\documentclass[onecolumn,letterpaper,left=1in, right=1in, bottom=1in, top=0.75in, 10pt]{IEEEtran}

\usepackage{float}
\usepackage[utf8]{inputenc} 
\usepackage[T1]{fontenc}
\usepackage{url}
\usepackage{ifthen}
\usepackage{cite}
\usepackage{algorithm}
\usepackage{algpseudocode}
\usepackage[cmex10]{amsmath}
\usepackage{wrapfig}
\usepackage{tikz}
\usepackage{circuitikz}
\usepackage{multirow}
\usepackage{makecell}
\usetikzlibrary{shapes, arrows, positioning}
\usepackage[includefoot,left=13mm,right=13mm,top = 13mm, bottom=13mm]{geometry}

\usepackage{cite}
\usepackage[colorlinks=false]{hyperref}%
\usepackage{mleftright}       
\mleftright                   
\usepackage{subcaption}
\usepackage{graphicx}         
\usepackage{booktabs}         
\usepackage{amsthm}
 \usepackage{graphicx}
\usepackage{amsmath}
\usepackage{graphics,graphicx,psfrag,color,float}
\usepackage{epsfig,bbold,bbm}
\usepackage{bm}
\usepackage{mathtools}
\usepackage{amsmath}
\usepackage{amssymb}
\usepackage{amsfonts}
\usepackage{amsfonts}
\usepackage {times}
\usepackage{bbm}
\usepackage{soul}
\usepackage{cleveref}
\usepackage{enumitem}
\usepackage{xcolor}

\makeatletter
\def\namedlabel#1#2{\begingroup
    #2%
    \def\@currentlabel{#2}%
    \phantomsection\label{#1}\endgroup
}
\makeatother
\DeclareUnicodeCharacter{0325}{}

\def\QED{\mbox{\rule[0pt]{1.5ex}{1.5ex}}}

\def\argmin{\mathop{\rm argmin}\limits}
\newtheorem{theorem}{\bf{Theorem}}
\newtheorem{fact}{Fact}
\newtheorem{assumption}{Assumption}

\newcommand{\supp}{\operatorname{supp}}
\newcommand{\RNum}[1]{\uppercase\expandafter{\romannumeral #1\relax}}
\newcommand{\tab}{\hspace*{2em}}
\newcommand{\beq}{\begin{equation}}
\newcommand{\enq}{\end{equation}}
\newcommand{\bel}{\begin{lemma}}
\newcommand{\enl}{\end{lemma}}

\newcommand{\bet}{\begin{theorem}}
\newcommand{\ent}{\end{theorem}}

\newcommand{\tr}{\mathrm{Tr}}

\newcommand{\nn}{\nonumber}

\newcommand{\ketbra}[1]{|#1\rangle\langle#1|}

\newcommand{\customfootnotetext}[2]{{
  \renewcommand{\thefootnote}{#1}
  \footnotetext[0]{#2}}}
\newcommand{\eps}{\varepsilon}

\newcommand*{\bbR}{\mathbb{R}}

\newcommand*{\bbM}{\mathbb{M}}

\newcommand*{\bbI}{\mathbb{I}}

\newcommand*{\bbE}{\mathbb{E}}

\newcommand*{\cA}{\mathcal{A}}

\newcommand*{\cH}{\mathcal{H}}
\newcommand*{\cM}{\mathcal{M}}

\newcommand*{\cB}{\mathcal{B}}
\newcommand*{\cD}{\mathcal{D}}

\newcommand*{\cL}{\mathcal{L}}

\newcommand*{\cS}{\mathcal{S}}

\newcommand*{\cT}{\mathcal{T}}
\newcommand*{\cX}{\mathcal{X}}
\newcommand*{\cW}{\mathcal{W}}
\newcommand*{\cZ}{\mathcal{Z}}
\newcommand*{\cE}{\mathcal{E}}

\newcommand*{\cY}{\mathcal{Y}}

\newcommand{\ket}[1]{|#1 \rangle}
\newcommand{\braket}[2]{\langle #1|#2\rangle}

\mathchardef\mhyphen="2D

\newcommand{\ketbrasys}[2]{\overset{#2}{|#1\rangle\langle#1|}}

\newcommand{\norm}[2]{{\left\lVert #1 \right\rVert}_{#2}}
\newcommand{\abs}[1]{\left\vert#1\right\vert}

\makeatletter
\newcommand*{\rom}[1]{\expandafter\@slowromancap\romannumeral #1@}
\makeatother
\mathchardef\mhyphen="2D
\newlist{steps}{enumerate}{1}
\setlist[steps, 1]{leftmargin = 1.1cm, label = Step \arabic*}
\newtheorem{remark}{Remark}
\newtheorem{definition}{Definition}

\newtheorem{lemma}{Lemma}
\newtheorem{corollary}{Corollary}
\newtheorem{proposition}{Proposition}
\usepackage[utf8]{inputenc} 
\usepackage[T1]{fontenc}
\usepackage{url}
\usepackage{ifthen}
\usepackage[cmex10]{amsmath}
\usepackage{amsfonts} 

\begin{document}

\title{Generalization Bounds for Quantum Learning via Rényi Divergences}

\author{
 Naqueeb Ahmad Warsi\textsuperscript{$*$}, Ayanava Dasgupta\textsuperscript{$*$} and  Masahito Hayashi\textsuperscript{$\dagger$}

}

\customfootnotetext{$*$}{
Indian Statistical Institute,
Kolkata 700108, India.
Email: 
{\sf 
naqueebwarsi@isical.ac.in, ayanavadasgupta\_r@isical.ac.in
}
}

\customfootnotetext{$\dagger$}{School of Data Science, The Chinese University of Hong Kong, Shenzhen, Longgang District, Shenzhen, 518172, China

International Quantum Academy, Futian District, Shenzhen 518048, China

Graduate School of Mathematics, Nagoya University, Nagoya, 464-8602, Japan.
Email: 
{\sf hmasahito@cuhk.edu.cn}}

\markboth{Warsi, Dasgupta and Hayashi:
Generalization Bounds for Quantum Learning via Rényi Divergences}{}

\maketitle

\begin{abstract}
This work advances the theoretical understanding of quantum learning by establishing a new family of upper bounds on the expected generalization error of quantum learning algorithms, leveraging the framework introduced by Caro et al. (2024) and a new definition for the expected true loss. 
Our primary contribution is the derivation of these bounds in terms of quantum and classical R\'enyi divergences, utilizing a variational approach for evaluating quantum R\'enyi divergences, specifically the Petz and a newly introduced modified sandwich quantum R\'enyi divergence. 
Analytically and numerically, we demonstrate the superior performance of the bounds derived using the modified sandwich quantum R\'enyi divergence compared to those based on the Petz divergence. Furthermore, we provide probabilistic generalization error bounds using two distinct techniques: one based on the modified sandwich quantum R\'enyi divergence and classical R\'enyi divergence, and another employing smooth max R\'enyi divergence.

\end{abstract}
\section{Introduction}
Quantum learning theory has emerged as a burgeoning field at the intersection of quantum computation and machine learning, promising enhanced capabilities for data analysis and pattern recognition. 
Recent years have witnessed significant advancements in developing theoretical frameworks and algorithms for quantum learning, as highlighted in works such as \cite{FOF25,FMG19,DHLYT22,PMTL21,AK25,SB25,HRF23}. 
Among these, the work by Caro et al. \cite{Caro23} has provided a crucial foundation by introducing a comprehensive quantum learning framework. This framework uniquely addresses the challenges inherent in learning from quantum data, particularly when considering the interplay between classical training data and quantum data states. A key contribution of Caro et al. \cite{Caro23} lies in its analysis of quantum learning algorithms' performance through the lens of quantum information theory. By carefully considering the process of training on classical data in conjunction with quantum systems, and the subsequent generation of both classical and quantum hypotheses, Caro et al. \cite{Caro23} laid the groundwork for a rigorous evaluation of generalization in quantum learning scenarios. Their framework specifically tackles the issue of evaluating the learned hypothesis against a loss operator, a process complicated by the inherent measurements and post-processing that can perturb the quantum data. To overcome this, Caro et al. \cite{Caro23} introduced the concept of a test-train bipartition of the quantum data, allowing for the analysis of potential correlations and entanglement across these partitions and enabling the derivation of upper bounds on the expected generalization error in terms of both classical and quantum information-theoretic quantities. This prior work thus provides a vital starting point for further investigations into the theoretical limits and practical capabilities of quantum learning algorithms.

Building upon these foundational concepts from learning theory, the study of quantum learning aims to understand the capabilities and limitations of learning algorithms that utilize quantum resources. 
In classical learning theory, a central challenge lies in bridging the gap between the expected empirical loss, directly estimated from training data, and the theoretically critical expected true loss on unseen data. This difference is defined as the generalization error \cite{VC1971,Valiant1984,McAllester1999,Zhang2006,HGDR2024}, 
a key figure of merit for assessing how well a learning algorithm's output generalizes to new, independent data. 
A substantial body of work in classical learning theory has been dedicated to analyzing this generalization error \cite{Mcallester2013,RZ2016,XR_2017,Bu_2020,Negrea2019,haghifam2020sharpened,hellstrom2022new,Modak21, hellstrom2022new,Sefidgaran2022,Sefidgaran2022dist,Catoni_2007,Esposito21}. Notably, an information-theoretic approach, as employed in \cite{RZ2016,XR_2017,Bu_2020,Modak21,Esposito21,hellstrom2022new,Sefidgaran2022,Sefidgaran2022dist}, often utilizes the combination of variational lower bounds on divergence \cite{DV1975} and concentration inequalities like Hoeffding's Lemma \cite[Lemma $2.2$]{BLM_Concentration_2013} (or sub-Gaussianity assumptions \cite[Section $2.3$]{BLM_Concentration_2013} for unbounded loss functions).

The generalization error, inherently dependent on the stochastic nature of the training data and the learned hypothesis, is typically studied through two primary ways: 
its expectation and its probabilistic behavior. Investigations into the expected generalization error in classical learning have successfully employed the aforementioned information-theoretic tools. 
For instance, the references \cite{RZ2016,XR_2017} established bounds on the expected generalization error based on the mutual information between the training data and the resulting hypothesis. 
Furthermore, Modak et al. \cite{Modak21} derived upper bounds on the expected generalization error by leveraging the variational form of R\'{e}nyi divergence \cite{BDMRW2021}.

Complementarily, a significant amount of research in classical learning theory has also focused on the probabilistic generalization error. 
Works such as \cite{Catoni_2007,Esposito21} have proven upper bounds on the generalization error that hold with a certain probability. Specifically, Esposito et al. \cite{Esposito21} demonstrated that the generalization error is probabilistically upper-bounded by R\'{e}nyi divergence (up to additive and multiplicative constants) with high probability, achieving this through a change of measure technique rooted in H\"{o}lder's inequality (refer to Fact \ref{holder_classic}).

It is worth noting that many existing works in both classical \cite{Mcallester2013,RZ2016,XR_2017,Bu_2020,Negrea2019,haghifam2020sharpened,hellstrom2022new,Modak21,Esposito21,hellstrom2022new,Sefidgaran2022,Sefidgaran2022dist} and quantum \cite{Caro23} learning theory often assume the loss function (or loss observable in the quantum case) to be sub-Gaussian. However, as observed in the classical setting due to Hoeffding's inequality (Fact \ref{fact_hoefding_lemma}), this assumption is less restrictive for bounded loss functions. In this work, we establish a similar result for quantum learning by first proving a quantum analogue of Hoeffding's inequality, which subsequently allows us to demonstrate that a bounded loss operator is trivially sub-Gaussian in the quantum context.

The main contribution of this paper is to extend classical results on generalization error, specifically those obtained in \cite[Theorem $3$]{Modak21} and \cite[Corollary $2$]{Esposito21}, to the quantum setting within the framework of \cite{Caro23}. 
We achieve this by proving upper bounds on the expected generalization error for quantum learning algorithms using variational forms involving R\'{e}nyi-divergence-type quantities. 
Given that our evaluation is based on measurements, we obtain variational bounds based on measured R\'{e}nyi divergence. 
However, these bounds involve optimization over the choice of measurement, making their calculation challenging. 
To circumvent this, we focus on two specific types of quantum R\'{e}nyi divergences: Petz quantum R\'{e}nyi divergence \cite{Petz1986} and sandwich quantum R\'{e}nyi divergence 
\cite{MDSF_Renyi_2013,Wilde2014_Sandwich}. While sandwich quantum R\'{e}nyi divergence is generally smaller than Petz quantum R\'{e}nyi divergence, it only provides an upper bound on measured R\'{e}nyi divergence for $\alpha \ge 1/2$, whereas Petz quantum R\'{e}nyi divergence is applicable for $\alpha \in (0,1)\cup(1,\infty)$
(See Fact \ref{Mdata_processing_petz_renyi}). 
To address this limitation, we introduce a modification of the sandwich quantum R\'{e}nyi divergence using the reverse sandwich quantum R\'{e}nyi divergence with parameter $\alpha < 1/2$. By combining this modified sandwich quantum R\'{e}nyi divergence with a quantum Hoeffding bound, we derive tighter variational bounds for the expected generalization error. 
Furthermore, we investigate the generalization error in probability using various techniques adapted to the quantum setting.

In addition, Caro et al. \cite{Caro23} introduced a definition of true loss (see Definition \ref{Ctrue_loss_def}) for a quantum learning algorithm, which we believe is conceptually misleading. 
Thus, we in this work, propose a new definition of true loss (see Definition \ref{true_loss_def}) and also explain the motivation behind this definition (see Remark \ref{motivation_remark}). 

\begin{table}[h]
\caption{Relationship between results obtained in this work and related studies}\label{tab}
\begin{center}
\begin{tabular}{|c|c|c|c|}
\hline
 \multicolumn{2}{|c|}{Types of upper-bounds} & {Classical Learning Setting} & {Quantum Learning Setting}\\
\hline
 
\multirow{3}{*}{\makecell{Upper-bounds on the \\ generalization error \\ in expectation}} 
& {\makecell{Whole sample (as defined in \eqref{class_gen_ws_intro})\\ based bounds in terms of \\divergence}} & \makecell{\cite[Theorem $1$]{XR_2017}\\ (Proposition \ref{XR_2017_result} in this paper)}& \makecell{\cite[Theorem $17$ ]{Caro23},\\ \cite[ Corollary $23$]{Caro23}\\
(Proposition \ref{Caro23_result} in this paper)\\ and Theorem \ref{Caro23_result_mod} in this paper.} \\
\cline{2-4}
 & {\makecell{Whole sample based bounds\\ in terms of R\'enyi divergence}}  & & \makecell{Theorem \ref{theo_exxp_gen_err_bound_renyi} and Corollaries \ref{corr_exxp_gen_err_bound_renyi} and \ref{corr_exxp_gen_err_bound_renyi_weak}\\ of this paper.} \\
\cline{2-4}
 & {\makecell{Individual sample (as defined in \ref{class_ind_gen_ws_intro}) \\ based bounds in terms of \\ divergence}} & \makecell{\cite[Proposition $1$]{Bu_2020}\\(Proposition \ref{Bu_2019_result} in this paper)} & \cite[ Corollary $24$]{Caro23}\\
\cline{2-4}
 & {\makecell{Individual sample based bounds\\ in terms of R\'enyi divergence}} & \makecell{\cite[Theorem $1$]{Modak21}\\(Proposition \ref{Modak_2021_result} in this paper)} & \makecell{Corollary \ref{cor_exxp_gen_err_bound_mod_renyi_frag}\\ of this paper.} \\
\cline{1-4}
 \multirow{2}{*}{\makecell{Upper-bounds on the \\ generalization error\\ in probability}} 
&  {\makecell{Using H\"older's inequality \\ (see Fact \ref{holder_classic})}} & \makecell{\cite[Corollary $2$]{Esposito21}\\ (Proposition \ref{lemma_tail_pac_class} in this paper)} & {\makecell{Theorem \ref{lemma_tail_pac}  of this paper.}}\\
\cline{2-4}
 &{\makecell{Using smooth max R\'enyi divergence\\ (see Definition \ref{fact_smooth_max_divergence})}} & Theorem \ref{lemma_tail_pac_naive} of this paper. & Theorem \ref{lemma_tail_pac_smooth} of this paper.\\
\hline

\end{tabular}
\end{center}
\end{table}

\subsection*{Organization of this Paper and our contributions}
\begin{itemize}
\item 
In Section \ref{sec:bg}, we introduce the notations, facts, and definitions used in this manuscript. In particular, we present the modified sandwich quantum R\'{e}nyi divergence. Additionally, we provide a variational lower bound for this modified sandwich quantum R\'{e}nyi divergence.

\item 
In Section \ref{sec:variation_lb}, we discuss the usefulness of the variational lower bound on KL-divergence. Following this, we extend this discussion to the quantum setting for bounded linear operators. To achieve this, we prove a quantum version of Hoeffding's lemma in Lemma \ref{quantum_hoeffding_lemma}. This lemma is analogous to its classical counterpart. Consequently, it allows us to deduce that every bounded linear operator is sub-Gaussian. Furthermore, we show that the modified sandwich quantum R\'{e}nyi divergence is asymptotically close to the measured R\'{e}nyi divergence, even without the i.i.d. assumption.

\item 
In Section \ref{sec:learning_framework}, we begin by discussing the quantum learning framework introduced by \cite{Caro23}. However, in this section, we also present a new definition for the expected true loss of quantum learning algorithms, offering an alternative to the one found in \cite{Caro23}. Subsequently, we justify and explain the motivation for this new definition.

\item 
In Section \ref{sec:prev_work}, we discuss a relation between the $L_1$ distance (between distributions) and the expected generalization error for bounded loss functions. Following this, we discuss a similar relation between the $L_1$ distance (between quantum states) and the expected quantum generalization error for bounded linear loss observables. Additionally, we review existing works in the literature that study upper bounds on the expected generalization error in both classical and quantum learning scenarios.

\item 
In Section \ref{sec:gen_bound_exp}, we prove a family of upper bounds for the expected quantum generalization error. These bounds use the modified sandwich quantum R\'{e}nyi divergence and also the classical R\'{e}nyi divergence. Although it is possible to replace the modified sandwich quantum R\'{e}nyi divergence with the Petz quantum R\'{e}nyi divergence, our bounds using the modified version show better performance. Furthermore, we demonstrate that the bounds obtained in \cite{Caro23} can be derived from our bounds, which are based on quantum divergence. Finally, we numerically compare these three bounds, and the results show that our bounds based on the modified sandwich quantum R\'{e}nyi divergence are superior.

\item 
Section \ref{sec:gen_bound_exp2} presents two distinct evaluations of the probabilistic behavior of the generalization error for quantum learning algorithms. Specifically, one probabilistic bound uses the modified sandwich quantum R\'{e}nyi divergence and the classical R\'{e}nyi divergence. The other probabilistic bound we obtain uses the smooth max R\'{e}nyi divergence.

\end{itemize}

In particular, Table \ref{tab} above summarizes the relation between our target problems and existing studies.

\section{Notations, Facts and Definitions}\label{sec:bg}
We use $\cH$ to denote a finite-dimensional Hilbert space and we denote its dimension with $\abs{\cH}$, $\mathcal{D}(\cH)$ to represent the set of all state density matrix acting on $\cH$, $\mathcal{B}(\cH)$ to represent the set of all bounded self-adjoint operators acting on $\cH$ and $\cL(\cH)$ represents the set of all linear operators over $\cH$, $\cL_{\geq 0}(\cH)$ represents the set of all positive operators over $\cH$ ($\cD(\cH)\subset\cL_{\geq 0}(\cH)$), $T(\cH)$ represents the set of all trace class operators over $\cH$   and $\abs{\cH}$ represents the dimension of the Hilbert space $\cH$.

\begin{definition}\label{def_prob_l1_diff}
    Let $P$ and $Q$ be two probability distributions (measures) over a common metric space $\cX$. Then, for any $p \in [1,\infty)$ the $L_p$ distance between $P$ and $Q$ is defined as follows,

    \begin{equation*}
        \norm{P-Q}{p} := \left(\sum_{x \in \cX}\abs{P(x) - Q(x)}^{p}\right)^{\frac{1}{p}},
    \end{equation*}
    where the supremum is taken over all measurable sets. Note that $\norm{P-Q}{p}$ is a strictly decreasing function of $p$.
\end{definition}

\begin{definition}[R\'enyi Divergence {\cite[Equation $3.3$]{Renyi1961}}]\label{def_renyi_class}
    Consider two probability distributions $P$ and $Q$ and $\gamma \in (0,1) \cup (1,+\infty)$. Then, the Re\'nyi divergence of order $\gamma$ is defined as follows,
    \begin{align*}
        D^{c}_{\gamma} (P \| Q) := \begin{cases}
            \frac{1}{\gamma - 1}\log\bbE_{Q}\left[\left(\frac{dP}{dQ}\right)^{\gamma}\right], & {\mbox{if }} P \ll Q, \\
            +\infty, & \mbox{else}.
        \end{cases} 
    \end{align*}
\end{definition}

\begin{definition}[Smooth max R\'enyi divergence {\cite[Definition $10$]{Warsi2016}}]\label{fact_smooth_max_divergence}
    Consider two probability distributions $P$ and $Q$ over a finite set $\cX$. Then, the smooth max R\'enyi divergence $D^{(\eps)}_{\max}(P\|Q)$ of order $\eps$  is defined as follows,
    \begin{equation}
        D^{(\eps)}_{\max}(P\|Q) := \inf\left\{ a : \Pr_{P}\left\{x \in \cX : \log\frac{P(x)}{Q(x)} < a\right\} \geq 1 - \eps\right\}.\label{fact_smooth_max_divergence_eq}
    \end{equation}
\end{definition}

\begin{definition}[Convex conjugate of a function {\cite{fenchel2014conjugate}}]\label{convconj}
    Given a convex function $f : \bbR \to \bbR$, the convex conjugate $f^{*}$ corresponding to $f$ is given as,
    \begin{equation*}
        f^{*}(y) := \sup_{x \in \bbR} \left(xy - f(x)\right).
    \end{equation*}
\end{definition}

\begin{definition}[sub-Gaussianity of random variables {\cite[Section $2.3$]{BLM_Concentration_2013}}]\label{classical_sub_gaussian}
    For some $\beta > 0$, a random variable $X$ is defined to be $\beta$-sub-Gaussian if, $\forall \lambda \in \bbR$, the logarithmic moment-generating function of $X$ satisfies the following condition,
    \begin{equation}
        \log\bbE\left[e^{\lambda(X - \bbE\left[X\right])}\right] \leq \frac{\lambda^2\beta^2}{2}.\label{classical_sub_gaussian_eq}
    \end{equation}
\end{definition}

\begin{definition}[{\cite[Equation $\textnormal{IV}.31$]{bhatia2013matrix}}]\label{def_schatten_norm}
    Given any operator $O\in \cL(\cH)$, for any $p \in [1,\infty)$, the Schatten-$p$-norm $\norm{O}{p}$ of $O$ is defined as follows,
    \begin{equation*}
        \norm{O}{p} := \left(\tr\left[\left(\sqrt{O^{\dagger} O}\right)^{p}\right]\right)^{\frac{1}{p}},
    \end{equation*}
    where $O^{\dagger}$ is the conjugate-transpose of $O$ and we let $\norm{O}{\infty} := \lim_{p \to \infty}\norm{O}{p}$, which turns out to be the largest singular value  of $O$, if $O \in \cB(\cH)$.
\end{definition}

\begin{definition}
    Consider $\rho,\sigma \in \cD(\cH)$. 
Quantum Divergence between $\rho$ and $\sigma$ is defined as follows
    \begin{equation*}
        D(\rho \| \sigma) := \begin{cases}
            \tr[\rho(\log\rho -\log\sigma)], &\text{  if } \rho \ll \sigma, \\
             +\infty, &\text{ else}.
        \end{cases}
    \end{equation*}
        Then, Measured Quantum Divergence between $\rho$ and $\sigma$ is defined as follows
    \begin{align*}
        D^{\bbM}(\rho \| \sigma) := \sup_{\cX,\{\Lambda_x\}_{x \in \cX}}
                \sum_{x \in \cX}\tr[\Lambda_x \rho]
      (  \log \tr[\Lambda_x \rho]-\log \tr[\Lambda_x \sigma] ),
    \end{align*}
    where the supremum is over the  choices of finite sets $\cX$ and POVMs $,\{\Lambda_x\}_{x \in \cX}$.
\end{definition}

\begin{definition}[Measured R\'enyi Divergence { \cite[Eqs. $3.116$-$3.117$]{fuchs1996distinguishabilityaccessibleinformationquantum}}]\label{def_renyi_meas}
    Consider $\rho,\sigma \in \cD(\cH)$ and $\alpha \in (0,1) \cup (1,+\infty)$. 
    Then, Measured R\'enyi Divergence of order $\alpha$ between $\rho$ and $\sigma$ is defined as follows,
    \begin{align*}
        D^{\bbM}_{\alpha} (\rho \| \sigma) := \sup_{\cX,\{\Lambda_x\}_{x \in \cX}}\frac{1}{\alpha - 1}\log\sum_{x \in \cX}(\tr[\Lambda_x \rho])^{\alpha}(\tr[\Lambda_x \sigma])^{1 -\alpha},
    \end{align*}
    where the supremum is over the  choices of finite sets $\cX$ and POVMs $,\{\Lambda_x\}_{x \in \cX}$.
\end{definition}

\begin{definition}[Petz Quantum R\'enyi Divergence {\cite{Petz1986}}]\label{def_renyi_petz}
    Consider $\rho,\sigma \in \cD(\cH)$ and $\alpha \in (0,1) \cup (1,+\infty)$. Then, Petz Quantum R\'enyi Divergence of order $\alpha$ between $\rho$ and $\sigma$ is defined as follows,
    \begin{align*}
        D_{\alpha} (\rho \| \sigma) := \begin{cases}
            \frac{1}{\alpha - 1}\log\tr\left[\rho^{\alpha}\sigma^{1 - \alpha}\right], & {\mbox{if }} (\alpha < 1 \cap \rho \not\perp \sigma) \cup (\rho \ll \sigma),\\
            +\infty, & \mbox{else}.
        \end{cases}
    \end{align*}
\end{definition}

\begin{definition}[Minimal/ Sandwiched Quantum R\'enyi Divergence {\cite{MDSF_Renyi_2013}, \cite[Definition $4$]{Wilde2014_Sandwich}}]\label{def_renyi_sandwiched}
    Consider $\rho,\sigma \in \cD(\cH)$ and $\alpha \in (0,1) \cup (1,+\infty)$. Then, `minimal/sandwiched Quantum R\'enyi Divergence' or `Quantum R\'enyi Divergence' of order $\alpha$ between $\rho$ and $\sigma$ is defined as follows,
    \begin{align*}
        \Tilde{D}_{\alpha} (\rho \| \sigma) := \begin{cases}
            \frac{1}{\alpha - 1}\log\tr\left[\left(\sigma^{\frac{1 - \alpha}{2\alpha}}\rho\sigma^{\frac{1 - \alpha}{2\alpha}}\right)^{\alpha}\right], & {\mbox{if }} (\alpha < 1 \cap \rho \not\perp \sigma) \cup (\rho \ll \sigma),\\
            +\infty, & \mbox{else}.
        \end{cases}
    \end{align*}
\end{definition}

\begin{definition}[Reverse Sandwiched Quantum R\'enyi Divergence]\label{def_rev-renyi_sandwiched}
    Consider $\rho,\sigma \in \cD(\cH)$ and $\alpha \in (0,1)\cup (1,+\infty)$. Then, reverse sandwiched Quantum R\'enyi Divergence
    of order $\alpha$ between $\rho$ and $\sigma$ is defined as follows,
    \begin{align*}
            \Tilde{D}_{\alpha}^R (\rho \| \sigma):= 
\frac{\alpha}{1-\alpha}\Tilde{D}_{1-\alpha} (\sigma \| \rho).
\end{align*}
\end{definition}

\begin{definition}[Modified Sandwiched Quantum R\'enyi Divergence]\label{def_renyi_mod_sandwiched}
    Consider $\rho,\sigma \in \cD(\cH)$ and $\alpha \in (0,1)$. Then, modified sandwiched Quantum R\'enyi Divergence
    of order $\alpha$ between $\rho$ and $\sigma$ is defined as follows,
    \begin{align*}
            \overline{D}_{\alpha} (\rho \| \sigma):= 
        \begin{cases}
            \Tilde{D}_{\alpha}^R (\rho \| \sigma)
        & {\mbox{if }} (\alpha < 1/2),\\
            \Tilde{D}_{\alpha} (\rho \| \sigma)
        & {\mbox{if }} (\alpha \ge 1/2).
        \end{cases}
\end{align*}
\end{definition}

\begin{definition}[sub-Gaussianity of observables]\label{quantum_sub_gaussian}
    For some $\mu > 0$, a self-adjoint operator $O \in \cB(\cH)$ is defined to be $\mu$-sub-Gaussian with respect to a quantum state $\rho \in \cD(\cH)$ if, $\forall \lambda \in \bbR$, $O$ satisfies the following condition,
    \begin{equation*}
        \log\tr\left[e^{\lambda(O - \tr\left[O\rho\right]\bbI_{\cH})}\rho\right] \leq \frac{\lambda^2\mu^2}{2}.
    \end{equation*}
\end{definition}

\begin{fact}[Jensen's inequality {\cite{Jensen1906}}]\label{fact_jensen}
    Given $X$ is a random variable and $\psi$ and $\phi$ are convex and concave functions, respectively. Then,
    \begin{align}
        \psi\left(\bbE[X]\right) &\leq \bbE\left[\psi(X)\right],\label{fact_jensen_conv}\\
        \phi\left(\bbE[X]\right) &\geq \bbE\left[\phi(X)\right]\label{fact_jensen_conc}.
    \end{align}
\end{fact}

\begin{fact}[H\"older's inequality]\label{holder_classic}
    Given two random variables $X$ and $Y$ and two real numbers $p,q \in [1,\infty)$ such that $\frac{1}{p}+\frac{1}{q} = 1$, we have,
    \begin{equation}
        \bbE[\abs{XY}] \leq \left(\bbE[\abs{X}^{p}]\right)^{1/p}\left(\bbE[\abs{X}^{q}]\right)^{1/q}.\label{Holder_fact_classical}
    \end{equation}
\end{fact}

\begin{fact}[Variational form of $L_p$ distance]\label{fact_var_form_l1_distance}
    Let $P$ and $Q$ be two probability distributions (measures) over a common metric space $\cX$. Then for any $p \in [1,\infty)$ the $L_p$ distance between $P$ and $Q$ has the following variational representation,
    \begin{equation}
        \norm{P-Q}{p} := \frac{1}{B} \sup_{\norm{f}{q} \leq B} \left(\bbE_{X \sim P}\left[f(X)\right] - \bbE_{X \sim Q}\left[f(X)\right]\right),\label{lpvarformeq}
    \end{equation}
    where $q$ is the H\"older conjugate of $p$ i.e. $q$ satisfies $\frac{1}{p}+\frac{1}{q}=1$ and the supremum is taken over all bounded functions $f$ over $\cX$ within range $[-B,B]$ i.e. $\norm{f}{q} \leq B$ (where $\norm{f}{q} := \left(\sum_{x \in \cX} \abs{f(x)}^{q}\right)^{\frac{1}{q}}$). Further, setting $p=1$, the $L_1$ distance between $P$ and $Q$ has the following variational representation,
    \begin{equation}
        \norm{P-Q}{1} := \frac{1}{B} \sup_{\norm{f}{\infty} \leq B} \left(\bbE_{X \sim P}\left[f(X)\right] - \bbE_{X \sim Q}\left[f(X)\right]\right),\label{l1varformeq}
    \end{equation}
    where the supremum is taken over all bounded functions $f$ over $\cX$ within range $[-B,B]$ i.e. $\norm{f}{\infty} \leq B$ (where $\norm{f}{\infty} := \sup_{x \in \cX} \abs{f(x)}$).
\end{fact}

\begin{fact}\label{fact_limit_renyi_divergence}
    $D^c_{\gamma}$ is monotonically increasing in $\gamma > 0$ and
    \begin{equation*}
        D^c_{1}(P\|Q) := \lim_{\gamma \to 1} D^c_{\gamma} (P\|Q) = D^{c}(P\|Q),
    \end{equation*}
    where, 
    \begin{equation*}
        D^{c}(P\|Q) := \begin{cases}
            
            \bbE_{P}\left[\log\left(\frac{dP}{dQ}\right)\right], &\text{  if } P \ll Q, \\
             +\infty, &\text{ else}.
        \end{cases}
    \end{equation*}
\end{fact}

\begin{fact}[Donsker-Varadhan variational form for divergence {\cite[Corollary $4.15$]{DV1975}}]
\label{dv_rela_ent}
    Let $P$ and $Q$ be probability measures on $\cX$. Then,  we have the following dual form of $D^{c}(P\|Q)$,
\begin{align*}
D^{c}(P \| Q) = \sup_{G \in \cM_{b}(\cX)} \left\{\bbE_{X \sim P}\left[G(X)\right]- \log \bbE_{X \sim Q}\left[e^{ G(X)}\right]\right\}.
\end{align*}
where $\cM_{b}(\cX)$ denotes the set of bounded measurable real-valued functions on $\cX$. The above dual form of $D^{c}(P\|Q)$ is also known as a \textbf{variational form}, which is a direct consequence of convex duality (see Definition \ref{convconj}).
\end{fact}

\begin{fact}[Variational form for R\'enyi divergence {\cite[Theorem $3.1$]{BDMRW2021}}]
\label{dv_renyi}
    Let $P$ and $Q$ be probability measures on $\cX$ and $\gamma \in \mathbb{R}\setminus\{0,1\}$. Then, we have,
\begin{align}
D^{c}_{\gamma}(P \| Q) = \sup_{G \in \Gamma} \left\{\frac{\gamma}{\gamma-1}\log\bbE_{X \sim P}\left[e^{(\gamma-1)G(X)}\right]- \log \bbE_{X \sim Q}\left[e^{ \gamma G(X)}\right]\right\}.\label{class_kl_var_lb}
\end{align}
where $\cM_{b}(\cX) \subset \Gamma \subset \cM(\cX)$ ($\cM(\cX)$ denotes the set of all
real-valued measurable functions $\cX$).
\end{fact}

\begin{fact}[Hoeffding's Lemma {\cite[Lemma $2.2$]{BLM_Concentration_2013}}]\label{fact_hoefding_lemma}
    Let $X_i$, for $i \in[n]$, be i.i.d. random variables distributed according to a probability distribution $P_X$ taking values in a bounded interval $[a,b]$ and $\mathbb{E}\left[X_i\right]=\mu$. Let $X=\sum_{i=1}^n X_i / n$ denote the average of the $\{X_i\}_{i=1}^{n}$. Then, $\forall\lambda \in \mathbb{R}$, the following holds,

\begin{align}
\log \mathbb{E}\left[e^{\lambda X}\right] \leq \lambda \mu + \frac{\lambda^2 (b-a)^2}{8n}.\label{fact_hoefding_lemma_eq}
\end{align}

\end{fact}

\begin{fact}[{\cite{Esposito21},\cite{Modak21}}]\label{Fact_change_measure}
    Consider a probability distribution $P_{AB}$ over $\cA \times \cB$ (where $\abs{\cA},\abs{\cB} > 0$) and if $P_{AB} \ll P_A \times P_B$ (where $P_A$ and $P_B$ are the corresponding marginal of $P_{AB}$), then, for any event $E \subseteq \cA\times\cB$, the following holds,
    \begin{equation}
        \Pr_{(A,B) \sim P_{AB}}\{E\} \leq e^{\frac{\gamma - 1}{\gamma} \left(\log\left(\bbE_{P_{A}}\left[\Pr_{B\sim P_{B}}\{E_{B}\}\right]\right) + I^{c}_{\gamma}[A;B]\right)},\label{Fact_change_measure_eq}
    \end{equation}
    where $\forall b \in \cB$, $E_{b}  := \{a \in \cA : (a,b) \in E\}$.
\end{fact}

\begin{fact}[{\cite[Proposition $2.5$]{Wainwright_2019}}]\label{average_sub_gaussian}
    Consider a collection of $n$ $\tau$-sub-Gaussian i.i.d. random variable $X_1,\cdots,X_n$ and let $S = \frac{1}{n}\sum_{i=1}{n} X_i$. Then, we have,
    \begin{equation}
    \begin{split}
        \Pr\{S - \bbE[S] > t\} &\leq  e^{-\frac{n t^2}{2\tau^2}},\\
        \Pr\{S - \bbE[S] < -t\} &\leq  e^{-\frac{n t^2}{2\tau^2}}.
    \end{split}
        \label{average_sub_gaussian_eq}
    \end{equation}
\end{fact}

\begin{fact}\label{average_sub_gaussian_unnorm}
    Consider a collection of $n$ $\tau$-sub-Gaussian i.i.d. random variable $X_1,\cdots,X_n$ and let $S = \frac{1}{n}\sum_{i=1}{n} X_i$. Then, we have,
    \begin{equation}
        \Pr\{\abs{S - c}< t\} \leq  2e^{-\frac{n (t-\abs{\bbE[S]-c})^2}{2\tau^2}}.\label{average_sub_gaussian_eq_unnorm}
    \end{equation}
\end{fact}
\begin{proof}
    Consider the following series of inequalities,
    \begin{align*}
        \Pr\{\abs{S - c}> t\} &\leq \Pr\{S > c+ t\} + \Pr\{S < c- t\}\nn\\
        &= \Pr\{S-\bbE[S] > t - (\bbE[S]-c)\} + \Pr\{S - \bbE[S] < (c-\bbE[S]) -t \}\nn\\
        &\overset{a}{\leq} e^{-\frac{n (t - (\bbE[S]-c))^2}{2\tau^2}} + e^{-\frac{n ((c-\bbE[S]) - t)^2}{2\tau^2}}\nn\\
        &= e^{-\frac{n (t - (\bbE[S]-c))^2}{2\tau^2}} + e^{-\frac{n (t - (c-\bbE[S]))^2}{2\tau^2}}\nn\\
        &\leq 2e^{-\frac{n (t-\abs{\bbE[S]-c})^2}{2\tau^2}},
    \end{align*}
    where $a$ follows from Fact \ref{average_sub_gaussian}. This proves Fact \ref{average_sub_gaussian_unnorm}.
\end{proof}

\begin{fact} (Variational lower-bound for Schatten $L_p$ distance)\label{fact_var_form_schatten1}
    Given two $\rho,\sigma \in \cD(\cH)$, for $p \in [1,\infty)$, we have the following variational form for $\norm{\rho -\sigma}{p}$,
    \begin{equation}
        \norm{\rho -\sigma}{p} \geq \frac{1}{B} \sup_{\substack{H \in \cB(\cH):\\ \norm{H}{q} \leq B}}\left(\tr[H\rho] -\tr[H \sigma]\right),\label{oplpvarform}
    \end{equation}
    where $q$ is H\"older conjugate of $p$ i.e. $q$ satisfies $\frac{1}{p} + \frac{1}{q} = 1$. Further, setting $p = 1$, we get the following variational form for $\norm{\rho -\sigma}{1}$,
    \begin{equation}
        \norm{\rho -\sigma}{1} \geq \frac{1}{B} \sup_{\substack{H \in \cB(\cH):\\-B\bbI \preceq H \preceq B\bbI}}\left(\tr[H\rho] -\tr[H \sigma]\right).\label{opl1varform}
    \end{equation}
\end{fact}

\begin{proof}
    For any $H \in \cB(\cH): \norm{H}{q} \preceq B$, consider the following series of inequalities,
    \begin{align}
        \frac{1}{B}\tr\left[H(\rho - \sigma)\right] &\overset{a}{\leq} \frac{1}{B}\norm{H}{q}\norm{\rho -\sigma}{p}\nn\\
        &\overset{b}{\leq} \norm{\rho -\sigma}{p},\label{fact_var_form_schatten1_eq}
    \end{align}
    where $a$ follows from Fact \ref{Holder_fact} and $b$ follows from the fact that $\norm{H}{q} \leq B$. Then, taking supremum over $H$ on both sides of \eqref{fact_var_form_schatten1_eq} completes the proof for Fact \ref{fact_var_form_schatten1}.
\end{proof}

\begin{fact}\label{fact_limit_petz_divergence}
    $D_{\alpha}$ is monotonically increasing in $\alpha > 0$ and,
    \begin{equation*}
        D_{1}(\rho\|\sigma) := \lim_{\alpha \to 1} D_{\alpha} (\rho \| \sigma) = D(\rho \| \sigma).
    \end{equation*}
\end{fact}

\begin{fact}\label{fact_limit_sandwiched_divergence}
    $\Tilde{D}_{\alpha}$ is monotonically increasing in $\alpha > 0$ and,
    \begin{equation*}
        \Tilde{D}_{1}(\rho\|\sigma) := \lim_{\alpha \to 1} \Tilde{D}_{\alpha} (\rho \| \sigma) = D(\rho \| \sigma).
    \end{equation*}
\end{fact}

\begin{fact}\label{fact_quantum_renyi_addit}
    Consider two quantum states $\rho := \bigotimes^{n}_{i=1}\rho^{(i)}$ and  $\sigma := \bigotimes^{n}_{i=1}\sigma^{(i)}$. Then, for any $\alpha \in (0,1) \cup (1,+\infty)$, the following holds,
    \begin{align}
        D_{\alpha}(\rho\|\sigma) &= \sum_{i=1}^{n} D_{\alpha}(\rho^{(i)} \| \sigma^{(i)}),\label{fact_quantum_renyi_addit_eq1}\\
        \Tilde{D}_{\alpha}(\rho\|\sigma) &= \sum_{i=1}^{n} \Tilde{D}_{\alpha}(\rho^{(i)} \| \sigma^{(i)})\label{fact_quantum_renyi_addit_eq2}.
    \end{align}
\end{fact}

\begin{fact}
[Variational characterization of the quantum divergence 
{\cite{HIAI1993153},\cite[Theorem $5.9$]{H2017QIT},\cite{Berta2017}}]\label{petz_quant_var_kl}
     Let $\rho, \sigma \in \cD(\cH)$ be two quantum states. 
     Then, the divergence between $\rho$ and $\sigma$ can be rewritten as follows,
\begin{equation}
         D^{\bbM}\left(\rho \| \sigma \right)=\sup _{H\in \cB(\cH)}\left\{\tr\left[H \rho\right]-\log \tr\left[e^H\sigma\right]\right\}.\label{quant_kl_var_form}
\end{equation}
\end{fact}

\begin{fact}[{\cite[Lemma $4$]{Frank_Lieb2013}}]\label{sandwich_unres_var_form}
 Let $\rho, \sigma \in \cD(\cA)$. Then, for $\alpha\in(0,1)\cup(1,\infty)$, we have,
\begin{equation}
    \Tilde{D}_{\alpha}(\rho\|\sigma) =\sup _{\substack{H \in \cL(\cA) :\\ H \succeq 0}}\left\{ \frac{1}{\alpha-1}\log\left(\alpha \tr \left[H \rho\right] -(\alpha-1) \tr\left[\left(H^{\frac{1}{2}} \sigma^{\frac{\alpha-1}{\alpha}} H^{\frac{1}{2}}\right)^{\frac{\alpha}{\alpha-1}}\right]\right)\right\}. \label{sandwich_unres_var_eq}
\end{equation}
\end{fact}

\begin{fact}[{\cite[Lemma $3$ and Theorem $4$]{Berta2017}}]\label{meas_renyi_var_form}
Let $\rho, \sigma \in \cD(\cA)$. Then, for $\alpha\in(0,1)\cup(1,\infty)$, we have,
    \begin{equation}
     D^{\bbM}_{\alpha} (\rho \| \sigma) = \sup_{\substack{H \in \cL(\cA): \\ H \succ 0}}\left\{\frac{\alpha}{\alpha - 1}\log \tr\left[e^{(\alpha - 1)H}\rho\right] - \log\tr\left[e^{\alpha H}\sigma\right]\right\}, \tab \forall \alpha \in (0,1)\cup(1,\infty).\nn
\end{equation}

Recently, Fang et al. in \cite{FFF2025} obtained a variational expression for measured $f$-divergences (for operator convex functions).
\end{fact}

\begin{fact}[Araki-Lieb-Thirring trace inequality \cite{Araki1990,Lieb_Thirring_2005}]\label{fact_thirring}
    Consider $X,Y \in \cL_{\geq 0}(\cH)$. Then, we have the following inequality:
    \begin{align*}
        \begin{cases}
        \tr\left[\left[YXY\right]^{r}\right] &\leq \tr\left[Y^{r}X^{r}Y^{r}\right], \text{\quad\quad\quad if } r \geq 1,\\
        \tr\left[\left[YXY\right]^{r}\right] &\geq \tr\left[Y^{r}X^{r}Y^{r}\right], \text{\quad\quad\quad if } r \in [0,1].
        \end{cases}
    \end{align*}
\end{fact}
\begin{fact}\label{petz_bounded_sandw}
From Fact \ref{fact_thirring}, it directly follows that for any $\rho,\sigma \in \cD(\cH) : \rho \ll \sigma$ and $\forall \alpha \in (0,1)\cup(1,\infty)$, we have, $$D_{\alpha} (\rho \| \sigma) \geq \Tilde{D}_{\alpha} (\rho \| \sigma).$$
\end{fact}

\begin{fact}\label{trace_log_ineq}
    Consider a positive operator $A \in \cL_{\geq 0}(\cH)$ and a state-density operator $\rho \in \cD(\cH)$. Then, we have,
    \begin{equation*}
        \log\tr\left[\rho A\right] \geq \tr\left[\rho \log A \right].
    \end{equation*}
\end{fact}
\begin{proof}
    Assume the following eigen decomposition of $A$ and $\rho$.
    \begin{equation*}
        A = \sum_{i = 1}^{\abs{\cH}} \alpha_{i} \ketbra{i} \text{ and }
        \rho = \sum_{j = 1}^{\abs{\cH}} \beta_{j}\ketbra{j}, \text{ where } \forall j \in [\abs{\cH}], 0 < \beta_j < 1 \text{ and } \sum_{j=1}^{\abs{\cH}}\beta_j = 1.
    \end{equation*}
    Then, consider the following series of inequalities,
    \begin{align}
        \log\tr\left[\rho A\right] & = \log \left(\sum_{i=1}^{\abs{\cH}}\sum_{j=1}^{\abs{\cH}}\alpha_i\beta_j \abs{\braket{i}{j}}^{2}\right)\nn\\
        & = \log \left(\sum_{i=1}^{\abs{\cH}}\sum_{j=1}^{\abs{\cH}}\beta_j \abs{\braket{i}{j}}^{2} \alpha_i\right).\label{fact_tr_log_ineq1}
    \end{align}
For all $i \in [\abs{\cH}],$ let $p_i := \sum_{j=1}^{\abs{\cH}}\beta_j \abs{\braket{i}{j}}^{2}$. It is easy to see that $\forall i \in [\abs{\cH}],$ $p_i \geq 0$ and $\sum_{i = }^{|\cH|}p_i =1.$
  
Thus, we can now lower-bound \eqref{fact_tr_log_ineq1} as follows,
    {\allowdisplaybreaks\begin{align*}
        \log\tr\left[\rho A\right] &= \log \left(\sum_{i=1}^{\abs{\cH}}p_i \alpha_i\right)\\
        &\overset{a}{\geq} \sum_{i=1}^{\abs{\cH}}p_i\log\alpha_i\\
        &= \sum_{i=1}^{\abs{\cH}}\sum_{j=1}^{\abs{\cH}}\beta_j \abs{\braket{i}{j}}^{2}\log\alpha_i\\
        &= \tr\left[\left(\sum_{j = 1}^{\abs{\cH}} \beta_{j}\ketbra{j}\right)\left( \sum_{i = 1}^{\abs{\cH}} \log\alpha_{i} \ketbra{i}\right)\right]\\
        &= \tr\left[\rho\log A\right],
    \end{align*}}
    where $a$ follows from Jensen's inequality.
\end{proof}

\begin{fact}[H\"older's Inequality for operators {\cite[Equation $12.6$]{Wilde_2013}}]\label{Holder_fact}
     Given two positive semidefinite operator $A,B \in \cL(\cH)$ and two real numbers $p,q \in [1,\infty)$ such that $\frac{1}{p}+\frac{1}{q} = 1$, we have,
    \begin{equation}
        \abs{\tr[AB]} \leq \left( \tr[A^p]\right)^{\frac{1}{p}}\left( \tr[B^q]\right)^{\frac{1}{q}}.\label{holder_ineq_quan}
    \end{equation}
 \end{fact}

 \begin{fact}[Data-processing inequality of sandwiched quantum R\'enyi divergence {\cite[Theorem $1$]{Frank_Lieb2013}}]\label{data_processing_sandwiched_renyi}
      For any $\rho,\sigma \in \cD(\cH)$ and $\forall \alpha \in [\frac{1}{2},1)\cup(1,\infty)$, $\Tilde{D}_{\alpha}(\rho\|\sigma)$ satisfies the following,
     \begin{equation}
         \Tilde{D}_{\alpha}(\rho\|\sigma) \geq \Tilde{D}_{\alpha}(\cE(\rho)\|\cE(\sigma)),\label{data_processing_sandwiched_renyi_ineq}
     \end{equation}
     where $\cE$ is any completely positive and trace-preserving (CP-TP) map.
 \end{fact}

 \begin{fact}[{\cite[Eq. (3.17)]{H2017QIT}}]\label{limit-measurement}
      For any $\rho,\sigma \in \cD(\cH)$ and $\forall \alpha \in [\frac{1}{2},1)\cup(1,\infty)$, $\Tilde{D}_{\alpha}(\rho\|\sigma)$ satisfies the following,
     \begin{equation}
         \Tilde{D}_{\alpha}(\rho\|\sigma) 
=\lim_{n\to \infty}      \frac{1}{n}   D^{\bbM}_{\alpha}(\rho^{\otimes n}\|\sigma^{\otimes n}) 
         \label{limit-measurement2}.
     \end{equation}
 \end{fact}

\begin{fact}[Data-processing inequality of quantum divergence]\label{data_processing_div}
      For any $\rho,\sigma \in \cD(\cH)$, $D(\rho\|\sigma)$ satisfies the following,
     \begin{equation}
         D(\rho\|\sigma) \geq D(\cE(\rho)\|\cE(\sigma)),\label{data_processing_div_ineq}
     \end{equation}
     where $\cE$ is any completely positive and trace-preserving (CP-TP) map.
 \end{fact}

 \begin{fact}[Data-processing inequality of modified sandwiched quantum R\'enyi divergence]\label{data_processing_mod-sandwiched_renyi}
      For any $\rho,\sigma \in \cD(\cH)$ and $\forall \alpha \in 
      (0,1)\cup(1,\infty)$, $\overline{D}_{\alpha}(\rho\|\sigma)$ satisfies the following,
     \begin{equation}
         \overline{D}_{\alpha}(\rho\|\sigma) \geq \overline{D}_{\alpha}(\cE(\rho)\|\cE(\sigma)),\label{data_processing_mod-sandwiched_renyi_ineq}
     \end{equation}
     where $\cE$ is any completely positive and trace-preserving (CP-TP) map.
 \end{fact}

 \begin{fact}
 \label{limit-measurement3}
      For any $\rho,\sigma \in \cD(\cH)$ and $\forall \alpha \in (0,1)\cup(1,\infty)$, $\overline{D}_{\alpha}(\rho\|\sigma)$ satisfies the following,
     \begin{equation}
         \overline{D}_{\alpha}(\rho\|\sigma) 
=\lim_{n\to \infty}      \frac{1}{n}  {D}^{\bbM}_{\alpha}(\rho^{\otimes n}\|\sigma^{\otimes n}) 
         \label{limit-measurement4}.
     \end{equation}
 \end{fact}
\begin{proof}
    Consider the following series of inequalities:
    \begin{align*}
        \overline{D}_{\alpha}(\rho\|\sigma) 
&\overset{a}{=}\lim_{n\to \infty}      \frac{1}{n}  \overline{D}_{\alpha}(\rho^{\otimes n}\|\sigma^{\otimes n})\\
&\overset{b}{=}\lim_{n\to \infty}      \frac{1}{n}  {D}^{\bbM}_{\alpha}(\rho^{\otimes n}\|\sigma^{\otimes n}), 
    \end{align*}
    where $a$ follows from \eqref{fact_quantum_renyi_addit_eq2} of Fact \ref{fact_quantum_renyi_addit} and $b$ follows from Lemma \ref{lemmaB} by setting $\rho_n = \rho^{\otimes{n}}$ and $\sigma_n = \sigma^{\otimes{n}}$. This proves Fact \ref{limit-measurement3}
\end{proof}

 \begin{fact}[Data-processing inequality of Petz quantum R\'enyi divergence {\cite{Petz1986}}]\label{data_processing_petz_renyi}
      For any $\rho,\sigma \in \cD(\cH)$ and $\forall \alpha \in (0,1)\cup(1,2]$, ${D}_{\alpha}(\rho\|\sigma)$ satisfies the following,
     \begin{equation}
         D_{\alpha}(\rho\|\sigma) \geq D_{\alpha}(\cE(\rho)\|\cE(\sigma)),\label{data_processing_petz_renyi_ineq}
     \end{equation}
     where $\cE$ is any completely positive and trace-preserving (CP-TP) map.
 \end{fact}

 \begin{fact}[Measurement-data-processing inequality of Petz quantum R\'enyi divergence {\cite[Eq. (3.23)]{H2017QIT}}]\label{Mdata_processing_petz_renyi}
      For any $\rho,\sigma \in \cD(\cH)$ and $\forall \alpha \in (0,1)\cup(1,\infty)$, ${D}_{\alpha}(\rho\|\sigma)$ satisfies the following,
     \begin{equation}
         D_{\alpha}(\rho\|\sigma) \geq {D}^{\bbM}_{\alpha}(\rho\|\sigma).\label{M-data_processing_petz_renyi_ineq}
     \end{equation}
 \end{fact}

The combination of Facts \ref{data_processing_mod-sandwiched_renyi}, \ref{limit-measurement3}, and \ref{Mdata_processing_petz_renyi} implies the following fact.

 \begin{fact}\label{MD-Petz}
      For any $\rho,\sigma \in \cD(\cH)$ and $\forall \alpha \in (0,1)\cup(1,\infty)$, ${D}_{\alpha}(\rho\|\sigma)$ satisfies the following,
     \begin{equation}
         D_{\alpha}(\rho\|\sigma) \geq \overline{D}_{\alpha}(\rho\|\sigma)
         \geq {D}^{\bbM}_{\alpha}(\rho\|\sigma).
         \label{MD-Petz2}
     \end{equation}
 \end{fact}

\begin{proof}
Fact \ref{Mdata_processing_petz_renyi} shows that
     \begin{equation}
         D_{\alpha}(\rho\|\sigma) =
           \frac{1}{n} D_{\alpha}(\rho^{\otimes n}\|\sigma^{\otimes n}) 
            \geq \frac{1}{n} {D}^{\bbM}_{\alpha}(\rho^{\otimes n}\|\sigma^{\otimes n}).
         \label{MD-Petz4}
     \end{equation}
Taking the limit $n\to \infty$ and using Fact \ref{limit-measurement3}, we have
     \begin{equation}
         D_{\alpha}(\rho\|\sigma) 
            \geq \lim_{n\to \infty}
            \frac{1}{n} {D}^{\bbM}_{\alpha}(\rho^{\otimes n}\|\sigma^{\otimes n})=
             \overline{D}_{\alpha}(\rho\|\sigma).
         \label{MD-Petz5}
     \end{equation}
Also, Fact \ref{data_processing_mod-sandwiched_renyi} implies 
     \begin{equation}
\overline{D}_{\alpha}(\rho\|\sigma)
         \geq {D}^{\bbM}_{\alpha}(\rho\|\sigma).
         \label{MD-Petz6}
     \end{equation}
\end{proof}

\section{Discussion on Variational lower-bound on Divergence and its application }\label{sec:variation_lb}

In this section, we start with discussing the usefulness of a ``change of measure'' based variational lower-bounds of divergence, both in classical and quantum settings, by portraying two toy examples. Further, we observe that the upper-bound obtained in the classical example can be further tightened by using a variational lower-bound of R\'enyi divergence. To perform the same in the quantum scenario, we require a variational lower-bound for Petz quantum R\'enyi divergence and later in this section, we give a new proof for the variational lower-bound for Petz quantum R\'enyi divergence.

 Variational representation for KL-divergences (as discussed in Fact \ref{dv_rela_ent}) has several applications in learning theory and related areas \cite{RZ2016,XR_2017,Modak21,Bu_2020,HD2020}. In particular, consider a scenario where we want to analyze a function {$-\beta < f(X) < \beta,$} under $X \sim P$. However, the analysis under $P$ might be hard to perform. {The variational form discussed in Fact \ref{dv_rela_ent} along with Hoeffding's lemma mentioned in Fact \ref{fact_hoefding_lemma} (or the sub-Gaussianity assumption mentioned in Definition \ref{classical_sub_gaussian}, if we don't assume the function to be bounded), are beneficial in such situations.} Instead of analyzing $f(X)$ under $P$, we will analyze it under some distribution $Q$ such that $P<<Q.$ Even though this ``change of measure'' may allow easier analysis of $f(X)$ under $Q$ it will also force us to incur a penalty in terms of the divergence between $P$ and $Q.$ Formally, consider the following series of inequalities for any $\lambda>0,$
\begin{align}
\mathbb{E}_P[f(X)] & = \frac{1}{\lambda} \mathbb{E}_P[\lambda f(X)] \nonumber\\
&\overset{a} \leq \frac{1}{\lambda}\left( \log \mathbb{E}_Q[e^{\lambda f(X)}] + D^{c}(P\|Q)\right)\nonumber\\
&\overset{b}\leq \mathbb{E}_Q[f(X)] + \frac{D^{c}(P\|Q)}{\lambda} + \frac{\lambda \beta^2}{2}, \label{cm}
\end{align}
where $a$ follows from Fact \ref{dv_rela_ent} and $b$ follows from Fact \ref{fact_hoefding_lemma}. Thus, optimizing over the choice of $\lambda$, we get, 
$$\mathbb{E}_P[f(X)] \leq \mathbb{E}_Q[f(X)] + \beta \sqrt{2{D^{c}(P\|Q)}}.$$ Similarly, by optimizing over $\lambda$ (for $\lambda <0$) we have, 
\begin{equation}
    \mathbb{E}_P[f(X)] \geq \mathbb{E}_Q[f(X)] - \beta \sqrt{2{D^{c}(P\|Q)}}.\label{classical_toy_eq1}
\end{equation}
{From the derivation of \eqref{cm} it follows that Fact \ref{dv_rela_ent} and Fact \ref{fact_hoefding_lemma} (Definition \ref{classical_sub_gaussian}, if we don’t assume the function to be bounded) together are equivalent to change of measure.}
Before mentioning the quantum version of this discussion we note here that to obtain the bound on $\mathbb{E}_P[f(X)]$ we only needed a lower-bound on $D^{c}(P\|Q)$ and this is one of the motivation for our lower-bound discussed in Lemma \ref{lemma1}. Further, the proof of the upper-bound on $\mathbb{E}_P[f(X)]$ crucially needs Hoeffding's lemma (Fact \ref{fact_hoefding_lemma}). This leads us to prove a quantum version of Hoeffding's lemma mentioned below. 
\begin{lemma}[Quantum Hoeffding's lemma]\label{quantum_hoeffding_lemma}
     Given a quantum state $\rho \in \cD(\cH)$ and a self-adjoint operator $L \in \cB(\cH)$ such that $ a\bbI \preceq L \preceq b\bbI$ (where $a\geq b$ and $a,b \in \bbR$ and $\bbI$ denotes the projection over $\cH$). Then, $\forall \lambda \in \bbR$,
    \begin{align}
            \label{quantum_hoeffding_lemma_eq1}
        \log\tr\left[e^{\lambda ( L - \tr[L\rho]\bbI)}\rho\right] &\leq  \frac{\lambda^2(b-a)^2}{8},
        \end{align}
        or equivalently,
        \begin{align}
        \log\tr\left[e^{\lambda L}\rho\right] &\leq \lambda\tr[L\rho] +  \frac{\lambda^2(b-a)^2}{8}.\label{quantum_hoeffding_lemma_eq2}
    \end{align}
\end{lemma}
\begin{proof}
    See Appendix \ref{proof_quantum_hoeffding_lemma} for the proof.
\end{proof}
The following corollary directly follows from Lemma \ref{quantum_hoeffding_lemma}.
\begin{corollary}\label{bounded_trace_sub_gaussianity}
    Every $L \in \cB(\cH)$ is $\mu^2$-sub-Gaussian. That is for every $\rho \in \cD(\cH)$,
    
    \begin{equation}
        \log\tr\left[e^{\lambda ( L - \tr[L\rho]\bbI)}\rho\right] \leq  \frac{\lambda^2\mu^2}{2},\label{bounded_trace_sub_gaussianity_eq}
    \end{equation}
    where $\mu:= \frac{\norm{L}{\infty}}{2}$ (where $\norm{\cdot}{\infty}$ is defined in Definition \ref{def_schatten_norm}) and $\lambda\in \mathbb{R}$.
\end{corollary}

We now discuss a quantum version of the problem discussed above in the context of \eqref{classical_toy_eq1}. Consider a quantum state $\rho \in \cD(\cH)$ and an observable $L \in \cL(\cH)$ such that {$-\mu \bbI \preceq L \preceq \mu \bbI,$} where $\mu < \infty$. Suppose we want to analyze $L$ under $\rho$. However, this might be hard to analyze. As discussed in the classical setting, we will perform this analysis by using a `change of measure'. For this, we will require a variational form of $D(\rho\|\sigma)$ mentioned in Fact \ref{petz_quant_var_kl}, along with a quantum version of Hoeffding's lemma (quantum sub-Gaussianity assumption mentioned in  Definition \ref{quantum_sub_gaussian}, if we don’t assume the  observable $L$ to be bounded) mentioned above. This change of measure may allow an easier analysis of $L$ under $\sigma$, but it will also force us to incur a loss in terms of $D(\rho\|\sigma)$. Formally, consider the following series of inequalities for any $\lambda>0,$

\begin{align}
    \tr[L \rho] & = \frac{1}{\lambda}\tr[\lambda L \rho]\nn\\
    &\overset{a}{\leq} \frac{1}{\lambda}\left(D(\rho\| \sigma) + \log\tr[e^{\lambda L}(\sigma)]\right),\nn\\
    &\overset{b}{\leq} \frac{1}{\lambda}\left(D(\rho\|\sigma) + \lambda \tr[{L}\sigma] + \frac{\lambda^2 \mu^2}{2}\right),\nn\\
    &= \frac{D(\rho\|\sigma)}{\lambda} + \frac{\lambda \mu^2}{2} + \tr[{L}\sigma],\label{quant_toy_bound_eq}
\end{align}
where $a$ follows from Facts \ref{petz_quant_var_kl} and \ref{data_processing_div} and $b$ follows from Lemma \ref{quantum_hoeffding_lemma}. Thus, in \eqref{quant_toy_bound_eq}, if we minimize the RHS over $\lambda$, (for $\lambda >0$) we get the following bound,

\begin{equation}
    \tr[L \rho] \leq \mu\sqrt{2D(\rho\|\sigma)}+ \tr[{L}\sigma],\nn
\end{equation}
Similarly, by optimizing over $\lambda$ (for $\lambda <0$) we have,

\begin{equation}
    \tr[L \rho] \geq \tr[{L}\sigma]-\mu\sqrt{2D(\rho\|\sigma)}. \nn
\end{equation}

{From the derivation of \eqref{quant_toy_bound_eq} it follows that Fact \ref{quant_kl_var_form} and Lemma \ref{quantum_hoeffding_lemma} (Definition \ref{quantum_sub_gaussian}, if we don’t assume the observable $L$ to be bounded) together are equivalent to change of measure.}

We can derive a tighter bound in \eqref{cm} (likewise in \eqref{quant_toy_bound_eq}), 
if instead of using a variational form (a lower-bound is sufficient) for the KL divergence (quantum KL divergence), 
we use a variational form (see Fact \ref{dv_renyi}) for the R\'enyi divergence (quantum R\'enyi divergence). 

Measured R\'enyi divergence has a variational lower-bound, which can be used in the context of the problem discussed above. 
However, it is not so easy to calculate the measured R\'enyi divergence. 
Instead of this, we can employ the modified sandwiched quantum R\'enyi divergence.

Indeed, for $\alpha\in (1,\infty)$ it follows via the data-processing inequality for sandwiched quantum R\'enyi divergence and Fact \ref{meas_renyi_var_form}. 
But, as mentioned in Fact \ref{data_processing_sandwiched_renyi}, the sandwiched R\'enyi divergence satisfies the data processing inequality only for $\alpha \in [\frac{1}{2},\infty)$. 
Therefore, 
we can not use 
Sandwiched quantum R\'enyi divergence for $\alpha \in (0,\frac{1}{2})$
to get a variational lower-bound for the variational form for $D^{\bbM}_{\alpha}(\cdot\|\cdot)$. 
Further, as mentioned in Fact \ref{sandwich_unres_var_form} (variational form for sandwiched quantum R\'enyi divergence), the terms involving $\rho$ (the original state) and $\sigma$ (change in measure state) are sitting inside the $\log$ term. 
Therefore, it is not very clear how to use this variational form for sandwiched quantum R\'enyi divergence to further tighten this bound.
To resolve this problem,
we employ 
modified sandwiched quantum R\'enyi divergence 
instead of sandwiched quantum R\'enyi divergence. 
In fact, as mentioned in Fact \ref{MD-Petz},
modified sandwiched quantum R\'enyi divergence is upper bounded by 
Petz quantum R\'enyi divergence, which implies that 
modified sandwiched quantum R\'enyi divergence 
gives a better bound than Petz quantum R\'enyi divergence.

Although Fact \ref{limit-measurement} holds for the iid case,
as a generalization of Fact \ref{limit-measurement3},
the following lemma shows that 
modified sandwiched quantum R\'enyi divergence 
gives a good approximation of measured R\'enyi divergence.

\begin{lemma} \label{lemmaB}
    Consider sequences of states $\rho_n,\sigma_n \in \cD(\cA)$ 
    with $\rho_n \ll \sigma_n$. 
We denote the minimum eigenvalue of $\sigma_n$ by $\lambda_n$.    
When $\log \lambda_n$ behaves as a polynomial order for $n$,
    $\forall \alpha \in (0,1) \cup (1,\infty)$ we have the following,
         \begin{equation}
         \lim_{n\to \infty}      \frac{1}{n}  \overline{D}_{\alpha}(\rho_n\|\sigma_n) 
=\lim_{n\to \infty}      \frac{1}{n}  {D}^{\bbM}_{\alpha}(\rho_n\|
\sigma_n) 
         \label{limit-measurement5}.
     \end{equation}
     \end{lemma}

This lemma suggest the following.
When the system size is not so large, we can use measured R\'enyi divergence
for the calculation of the variational form.
However, when the system size is too large to calculate measured R\'enyi divergence,
we can use modified sandwiched quantum R\'enyi divergence 
as an upper bound of the variational form, which is close to 
the variational form.

\begin{proof}
It is sufficient to show 
         \begin{equation}
         \lim_{n\to \infty}      \frac{1}{n}  \Tilde{D}_{\alpha}(\rho_n\|\sigma_n) 
=\lim_{n\to \infty}      \frac{1}{n}  {D}^{\bbM}_{\alpha}(\rho_n\|
\sigma_n) 
         \label{limit-measurement6}
     \end{equation}
for $\forall \alpha \in [1/2,1) \cup (1,\infty)$
because the part $\alpha\in (0,1/2)$ follows from the part with 
$\alpha \in (1/2,1)$.
Since the inequality 
$  \frac{1}{n}  \Tilde{D}_{\alpha}(\rho_n\|\sigma_n) 
\ge \frac{1}{n}  {D}^{\bbM}_{\alpha}(\rho_n\|
\sigma_n) $ holds, 
it is sufficient to show 
         \begin{equation}
         \lim_{n\to \infty}      \frac{1}{n}  \Tilde{D}_{\alpha}(\rho_n\|\sigma_n) 
\le \lim_{n\to \infty}      \frac{1}{n}  {D}^{\bbM}_{\alpha}(\rho_n\|
\sigma_n) 
         \label{limit-measurement7}.
     \end{equation}
for $\forall \alpha \in [1/2,1) \cup (1,\infty)$.

We choose a polynomial $p(n)$ such that
$-\log \lambda_n \le p(n)$ and $p(n)$ is an positive integer.
We choose the spectral decomposition of $\sigma_n$ as
\begin{align}
{\sigma}_n:=\sum_j e^{ -s_j} E_j 
\end{align}
We define the operator $\tilde{\sigma}_n$ as
\begin{align}
\tilde{\sigma}_n:=\sum_j e^{ -\lceil \frac{ s_j}{p(n)} \rceil p(n)} E_j .
\end{align}
Hence, we have
\begin{align}
e^{ \frac{1}{p(n)}}\tilde{\sigma}_n\ge {\sigma}_n \ge \tilde{\sigma}_n.
\label{BNS}
\end{align}
Also, the number of eigenvalues of $\tilde{\sigma}_n$ is at most 
$p(n)+1$.
The relation \eqref{BNS} implies
\begin{align}
\Tilde{D}_{\alpha}(\rho_n\|\tilde{\sigma}_n)
\ge 
\Tilde{D}_{\alpha}(\rho_n\|\sigma_n)
\ge
\Tilde{D}_{\alpha}(\rho_n\|\tilde{\sigma}_n)
-\log p(n).
\label{BNS2}
\end{align}
Measured R\'enyi divergence also satisfies the relation;
\begin{align}
{D}^{\bbM}_{\alpha}(\rho_n\|\tilde{\sigma}_n)
\ge 
{D}^{\bbM}_{\alpha}(\rho_n\|\sigma_n)
\ge
{D}^{\bbM}_{\alpha}(\rho_n\|\tilde{\sigma}_n)
-\log p(n).
\label{BNS3}
\end{align}

In addition,  
we have \cite[Eq. (3.153) and the next equation of Eq. (3.156)]{H2017QIT}
\begin{align}
\log (p(n)+1)+{D}_{\alpha}^{\bbM}(\rho_n\|\tilde{\sigma}_n)
\ge 
\Tilde{D}_{\alpha}(\rho_n\|\tilde{\sigma}_n)
\label{BNS4}.
\end{align}
The combination of \eqref{BNS2}, \eqref{BNS3}, and \eqref{BNS4}
implies 
\begin{align}
2 \log (p(n)+1)+{D}_{\alpha}^{\bbM}(\rho_n\|{\sigma}_n)
\ge 
\log (p(n)+1)+{D}_{\alpha}^{\bbM}(\rho_n\|\tilde{\sigma}_n)
\ge 
\Tilde{D}_{\alpha}(\rho_n\|\tilde{\sigma}_n)
\ge 
\Tilde{D}_{\alpha}(\rho_n\|{\sigma}_n)
\label{BNS5}.
\end{align}
Thus, we have
\begin{align}
\frac{2 \log (p(n)+1)}{n}+\frac{1}{n}{D}_{\alpha}^{\bbM}(\rho_n\|{\sigma}_n)
\ge \frac{1}{n}
\Tilde{D}_{\alpha}(\rho_n\|{\sigma}_n)
\label{BNS6}.
\end{align}
Taking the limit in \eqref{BNS6}, we obtain \eqref{limit-measurement7}.

The idea for the discretization $\tilde{\sigma}_n$ of $\sigma_n$
was essentially used in \cite{Hayashi4}. This proves Lemma \ref{lemmaB}.
\end{proof}

Below \textit{without invoking data-processing inequality and Facts \ref{meas_renyi_var_form} and
\ref{Mdata_processing_petz_renyi}}, we present an alternative 
simple proof for 
the following lemma related to a variational lower-bound.
The proof of Lemma \ref{lemma1} follows from the H\"older's inequality (Fact \ref{Holder_fact}) and Araki-Lieb-Thirring trace inequality (see Fact \ref{fact_thirring}).
The following proof can be considered as another proof for Facts 
\ref{meas_renyi_var_form} and \ref{Mdata_processing_petz_renyi}.

\begin{lemma} \label{lemma1}
    Consider $\rho,\sigma \in \cD(\cA)$ with $\rho \ll \sigma$. Then, $\forall \alpha \in (0,1) \cup (1,\infty)$ we have the following,
    \begin{equation} 
        D_{\alpha} (\rho \| \sigma) \geq
        \sup_{\substack{H \in \cL(\cA): \\ H > 0}}\left\{\frac{\alpha}{\alpha - 1}\log \tr\left[e^{(\alpha - 1)H}\rho\right] - \log\tr\left[e^{\alpha H}\sigma\right]\right\}=
         {D}^{\bbM}_{\alpha}(\rho\|\sigma).\label{petz_unres_var_lb}
    \end{equation}
\end{lemma}
\begin{proof}

 We divide this proof into two cases. In the first case $\alpha \in (0,1)$ and in the second case $\alpha \in (1, \infty).$

{\textbf{Case} $\bf 1$ \textnormal{:} $\alpha \in (0,1)$}\\
For any strictly positive operator $H$, consider the following series of inequalities,
\begin{align*}
    \tr\left[\rho^{\alpha}\sigma^{1 - \alpha}\right] &\overset{a}{=} \tr\left[e^{\frac{-(1-\alpha)\alpha H}{2}}\rho^{\alpha}e^{\frac{-(1-\alpha)\alpha H}{2}} e^{\frac{(1-\alpha)\alpha H}{2}} \sigma^{1 - \alpha} e^{\frac{(1-\alpha)\alpha H}{2}}\right]\\
    &\overset{b}{\leq} \left(\tr \left[\left(e^{\frac{-(1-\alpha)\alpha H}{2}}\rho^{\alpha}e^{\frac{-(1-\alpha)\alpha H}{2}}\right)^{\frac{1}{\alpha}}\right]\right)^{\alpha} \left(\tr \left[\left(e^{\frac{(1-\alpha)\alpha H}{2}} \sigma^{1 - \alpha} e^{\frac{(1-\alpha)\alpha H}{2}}\right)^{\frac{1}{1 - \alpha}}\right]\right)^{1 -\alpha} \\
    &\overset{c}{\leq} \left(\tr \left[e^{\frac{-(1-\alpha) H}{2}}\rho e^{\frac{-(1-\alpha) H}{2}}\right]\right)^{\alpha} \left(\tr \left[e^{\frac{\alpha H}{2}} \sigma e^{\frac{\alpha H}{2}}\right]\right)^{1 -\alpha}\\
    &= \left(\tr \left[e^{(\alpha - 1) H}\rho \right]\right)^{\alpha} \left(\tr \left[e^{\alpha H} \sigma \right]\right)^{1 -\alpha},
\end{align*}
where $a$ follows from cyclicity of trace, $b$ follows from Fact \ref{Holder_fact}, $c$ follows from Fact \ref{fact_thirring} and $\frac{1}{\alpha} > 1$. Since $(1 - \alpha) > 0$, we have, 
\begin{align*}
    \frac{1}{1-\alpha}\log\tr\left[\rho^{\alpha}\sigma^{1 - \alpha}\right] &\leq \frac{\alpha}{1 - \alpha} \log \tr \left[e^{(\alpha - 1) H}\rho \right] + \log\tr \left[e^{\alpha H} \sigma \right].
\end{align*}
Thus, from the above inequality, we have, 
\begin{equation*}
    D_{\alpha} (\rho \| \sigma) = \frac{1}{\alpha - 1}\log\tr\left[\rho^{\alpha}\sigma^{1 - \alpha}\right] \geq \frac{\alpha}{\alpha - 1} \log \tr \left[e^{(\alpha - 1) H}\rho \right] - \log\tr \left[e^{\alpha H} \sigma \right],
\end{equation*}
which implies the inequality in \eqref{petz_unres_var_lb}.

When we fix a basis $\{ |u_j\rangle\}$ and restrict the range of $H$
into the Hermitian matrices diagonal under the basis $\{ |u_j\rangle\}$,
Fact \ref{dv_renyi} implies 
\begin{equation*}
    D_{\alpha} ({\cal E}_{\{ |u_j\rangle\}}(\rho) \| {\cal E}_{\{ |u_j\rangle\}}(\sigma)) 
    =
    \sup_{H>0} \left\{
    \frac{\alpha}{\alpha - 1} \log \tr \left[e^{(\alpha - 1) H}\rho \right] - \log\tr \left[e^{\alpha H} \sigma \right]
   \middle| H \hbox{ is diagonal to } \{ |u_j\rangle\} \right\}    ,
\end{equation*}
where ${\cal E}_{\{ |u_j\rangle\}}(\rho):=
\sum_j  |u_j\rangle\langle u_j| \rho |u_j\rangle\langle u_j| $.
Considering the supremum for the choice of the basis
$\{ |u_j\rangle\}$ in the above equation, we obtain 
the equality in \eqref{petz_unres_var_lb}.

{\textbf{Case} $\bf 2$ \textnormal{:} $\alpha \in (1,\infty)$}\\
For any strictly positive operator $H$, consider the following series of inequalities,
\begin{align*}
    \tr\left[e^{(\alpha - 1) H}\rho\right] &\overset{a}{=} \tr\left[\sigma^{\frac{-(1 - \alpha)}{2\alpha}}e^{(\alpha - 1) H} \sigma^{\frac{-(1 - \alpha)}{2\alpha}} \sigma^{\frac{(1 - \alpha)}{2\alpha}}\rho \sigma^{\frac{(1 - \alpha)}{2\alpha}}\right] \\
    &\overset{b}{\leq} \left(\tr \left[ \left(\sigma^{\frac{-(1 - \alpha)}{2\alpha}}e^{(\alpha - 1) H} \sigma^{\frac{-(1 - \alpha)}{2\alpha}}\right)^{\frac{\alpha}{\alpha - 1}}\right]\right)^{\frac{\alpha - 1}{\alpha}} \left(\tr \left[ \left(\sigma^{\frac{(1 - \alpha)}{2\alpha}}\rho \sigma^{\frac{(1 - \alpha)}{2\alpha}}\right)^{\alpha}\right]\right)^{\frac{1}{\alpha}}\\
    &\overset{c}{\leq} \left(\tr \left[ \sigma^{\frac{1}{2}}e^{\alpha  H} \sigma^{\frac{1}{2}}\right]\right)^{\frac{\alpha - 1}{\alpha}} \left(\tr \left[ \sigma^{\frac{(1 - \alpha)}{2}}\rho^{\alpha} \sigma^{\frac{(1 - \alpha)}{2}}\right]\right)^{\frac{1}{\alpha}} \\
    &= \left(\tr \left[e^{\alpha  H} \sigma\right]\right)^{\frac{\alpha - 1}{\alpha}} \left(\tr \left[\rho^{\alpha} \sigma^{1 - \alpha}\right]\right)^{\frac{1}{\alpha}},
\end{align*}
where $a$ follows from cyclicity of trace, $b$ follows from Fact \ref{Holder_fact}, $c$ follows from Fact \ref{fact_thirring} and because $\frac{\alpha}{\alpha - 1} > 1$. Since $\frac{\alpha}{\alpha - 1} > 0$, we have,
\begin{align*}
    \frac{\alpha}{\alpha - 1} \log \tr\left[e^{(\alpha-1) H}\rho\right] &\leq \log \tr \left[e^{\alpha H} \sigma\right] + \frac{1}{\alpha - 1 } \log \tr \left[\rho^{\alpha} \sigma^{1 - \alpha}\right]\\
    &= \log \tr \left[e^{\alpha H} \sigma\right] + D_{\alpha} (\rho \| \sigma).
\end{align*}
Thus, from the above inequality, we have, 
\begin{equation*}
    D_{\alpha} (\rho \| \sigma) = \frac{1}{\alpha - 1}\log\tr\left[\rho^{\alpha}\sigma^{1 - \alpha}\right] \geq \frac{\alpha}{\alpha - 1} \log \tr \left[e^{(\alpha - 1) H}\rho \right] - \log\tr \left[e^{\alpha H} \sigma \right],
\end{equation*}
which implies the inequality in \eqref{petz_unres_var_lb}.
In the same way, we obtain
the equality in \eqref{petz_unres_var_lb}. This proves Lemma \ref{lemma1}.

\end{proof}
Further, we give the following variational lowerbound for modified sandwiched quantum R\'enyi divergence.
\begin{lemma}\label{mod-san_renyi_var_form}
    Let $\rho, \sigma \in \cD(\cA)$. Then, for $\alpha\in(0,1)\cup(1,\infty)$, we have,
    \begin{equation*}
     \overline{D}_{\alpha} (\rho \| \sigma) \geq \sup_{\substack{H \in \cL(\cA): \\ H > 0}}\left\{\frac{\alpha}{\alpha - 1}\log \tr\left[e^{(\alpha - 1)H}\rho\right] - \log\tr\left[e^{\alpha H}\sigma\right]\right\}, \tab \forall \alpha \in (0,1)\cup(1,\infty).\nn
\end{equation*}
\end{lemma}

\begin{proof}
    The proof follows from the lower-bound of Facts \ref{MD-Petz} and \ref{meas_renyi_var_form}.
\end{proof}

In the subsequent sections, we will bound functions which are generalization errors of a quantum learning algorithm in terms of the Petz and modified sandwiched quantum R\'enyi divergence and thus, obtain a family of upper-bounds. As a special case our bounds will recover the bound obtained in \cite{Caro23}. Therefore, in the section below, we first discuss a framework developed by Caro et al. in \cite{Caro23} for quantum learning algorithms.

\section{Quantum learning framework proposed by Caro et al. \cite{Caro23}}\label{sec:learning_framework}
Before describing the framework discussed in this paper, we first describe a version of the classical learning scenario \cite{XR_2017,Esposito21,Modak21}.
This we believe, will help us to get the motivation behind the framework proposed in \cite{Caro23}.
\subsection{Generalized Classical Learning Framework}\label{subsec:gen_clas_learn}
Although generalized classical learning framework covers various settings,
we begin with supervised learning as a typical example.
Here, we consider an input space $\mathcal{X}$ and an output space $\mathcal{Y}$, where a training data point is given as a pair $(X, Y)$ with input $X \in \mathcal{X}$ and output $Y \in \mathcal{Y}$. Given $n$ training data points $(X_1, Y_1), \ldots, (X_n, Y_n)$, the learner aims to output a hypothesis that explains this data. A common type of hypothesis is a function $f: \mathcal{X} \to \mathcal{Y}$, such as an affine function $f(x) = ax + b$. In this case, the parameters $a$ and $b$ are often determined by minimizing the prediction error on the training data using the minimum mean square error method:
\begin{equation*}
\argmin_{(a,b)} \sum_{j=1}^n (f(x_j)-y_j)^2 = \argmin_{(a,b)} \sum_{j=1}^n (a x_j+b-y_j)^2.
\end{equation*}
This supervised learning setup can be generalized to a broader learning framework.

More generally, in a learning scenario, training data points reside in a sample space $\mathcal{Z}$. In the supervised learning example, $\mathcal{Z} = \mathcal{X} \times \mathcal{Y}$, with each data point $z \in \mathcal{Z}$ being an input-output pair $(x, y)$. We consider a training data sample $S = \{Z_{i}\}_{i=1}^{n}$ consisting of $n$ independent and identically distributed (i.i.d.) data points drawn from a distribution $P$. A learning algorithm $\mathcal{A}$ takes $S$ as input and produces a hypothesis $w \in \mathcal{W}$. In the supervised learning example, $\mathcal{W}$ is the set of functions from $\mathcal{X}$ to $\mathcal{Y}$. Since the algorithm $\mathcal{A}$'s output hypothesis $w$ is conditioned on the training data $S$, $\mathcal{A}$ can generally be represented by a conditional probability distribution $P^{\mathcal{A}}_{W|S}$.

Once the algorithm produces a hypothesis $w$, its performance is evaluated based on the training data $S$ using a loss function $l : \mathcal{W} \times \mathcal{Z} \to \mathbb{R}$. In the supervised learning example, $l((a,b),(x,y)) = (ax + b - y)^2$. A primary goal in a learning scenario is to design an algorithm $\mathcal{A}$ that minimizes the empirical loss, defined as
\begin{equation}
     \hat{l}_{S}(w) := \frac{1}{n} \sum_{i=1}^{n}l(w,Z_i) \label{class_emp_loss_intro}.
 \end{equation}
The expectation of this empirical loss, $\mathbb{E}_{S,W}[ \hat{l}_{S}(W)]$, is the expected empirical loss, which can be readily estimated as it is expected to be close to the observed $\hat{l}_{S}(w)$.

However, because the algorithm $\mathcal{A}$ determines the hypothesis based on the training data $S$, 
it can inadvertently learn patterns specific to this data, leading to a dependency (bias) between $W$ and $S$ that is undesirable in real-world applications. 
To address this, we need to evaluate the learned hypothesis on 
independently generated data.
For a more rigorous analysis, we consider splitting the available data into a training set $S_{tr}$ and a testing set $S_{te}$, where $(S_{te},S_{tr}) = \{(Z_{te,i},Z_{tr,i})\}_{i=1}^{n}$ is a sequence of $n$ i.i.d. pairs drawn from a joint distribution $P_{Z_{te},Z_{tr}}$. 
The training data $S_{tr}$ is used by the algorithm $\mathcal{A}$ to produce a hypothesis $W$, while the testing data $S_{te}$ is used to evaluate the loss. 
Since $\mathcal{A}$ outputs $W$ conditioned on $S_{tr}$, it is generally represented by the conditional probability distribution $P^{\mathcal{A}}_{W|S_{tr}}$.

After the algorithm produces a hypothesis $w$, its performance is evaluated on the test data $S_{te}$ using the loss function $l : \mathcal{W} \times \mathcal{Z} \to \mathbb{R}$. Thus, a key objective is to design an algorithm $\mathcal{A}$ that minimizes the empirical loss on the test data:
\begin{equation*}
     \hat{l}_{S_{te}}(w) := \frac{1}{n} \sum_{i=1}^{n}l(w,Z_{te,i}).
 \end{equation*}
This is the empirical test loss of $\mathcal{A}$ for hypothesis $w$ on $S_{te}$. We are also interested in the average performance of $\mathcal{A}$, given by the expected test empirical loss $\mathbb{E}_{S_{te},S_{tr},W}[ \hat{l}_{S_{te}}(W)]$. 
However, the correlation between the test data $S_{te}$ and the training data $S_{tr}$, along with the fact that $\mathcal{A}$'s output $W$ depends solely on $S_{tr}$, can introduce a bias between $W$ and $S_{te}$. Indeed, $W$ depends on $S_{te}$ through $S_{tr}$, forming a Markov chain $W - S_{tr} - S_{te}$, which is generally undesirable.

To mitigate the issue of dependency, we consider a scenario with an independent test set $\bar{S}_{te}$, which is independent of the training set $\bar{S}_{tr}$ and the learned hypothesis $\bar{W}$. 
The joint distribution is then $P_{\bar{S}_{te},\bar{S}_{tr},\bar{W}}= P_{\bar{S}_{te}}P_{\bar{S}_{tr},\bar{W}}$. Since the hypothesis $\bar{W}$ is derived from the training data $\bar{S}_{tr}$, we have $P_{S_{tr}, W}(s_{tr}, w) = P_{\bar{S}_{tr}, \bar{W}}(s_{tr}, w)$, implying $\bar{W} - \bar{S}_{tr} \perp \bar{S}_{te}$. The expected loss on this independent test set is calculated as:
\begin{align}
     \mathbb{E}_{\bar{S}_{te},\bar{S}_{tr},\bar{W}}[\hat{l}_{\bar{S}_{te}}(\bar{W})]
     &= \mathbb{E}_{\bar{S}_{te},\bar{W}}[\hat{l}_{\bar{S}_{te}}(\bar{W})]\nn\\
     &= \mathbb{E}_{\bar{W}}\left[\frac{1}{n}\sum_{i=1}^{n}\mathbb{E}_{\bar{Z}_{te,i}}l(\bar{W},\bar{Z}_{te,i})\right]\nn\\
     &= \mathbb{E}_{\bar{W}}[\mathbb{E}_{\bar{Z}_{te}}l(\bar{W},\bar{Z}_{te})]\label{class_exp_true_loss}.
 \end{align}
The quantity on the right-hand side of Equation \eqref{class_exp_true_loss} is known as the expected true loss of $\mathcal{A}$, representing the average performance of the algorithm over the entire data distribution.

Because of the previous discussions, we understand that a good learning algorithm must be unbiased. Specifically, the difference between its expected empirical loss and its expected true loss should be small. As a result, a learning algorithm achieves good generalization in expectation when this difference is small in magnitude.
Therefore, in this scenario of classical learning, we define the generalization error as follows:
\begin{align}
    \text{gen}(w,S_{te}):=
 \hat{l}_{S_{te}}(w) -
 \mathbb{E}_{\bar{W}}[\mathbb{E}_{\bar{Z}_{te}}l(\bar{W},\bar{Z}_{te})] \label{class_gen_ws_intro}.
\end{align}
Note that in the above equation, the generalization error $\text{gen}(\cdot,\cdot)$ is defined in terms of the whole sample $S_{te}$. 
Thus, it can be called a \textit{whole sample-based generalization error}. However, Equation \eqref{class_gen_ws_intro} can be rewritten as:
\begin{equation}
    \text{gen}(w,S_{te})
    = \frac{1}{n} \sum_{i=1}^{n} \text{gen}_{\text{ind}}(w,Z_i) \label{class_ind_gen_ws_intro},
\end{equation}
where $\text{gen}_{\text{ind}}(w,Z_i) := l(w,Z_i) - \mathbb{E}_{\bar{Z}_{te}}[l(w,\bar{Z}_{te})]$ is denoted as the \textit{individual sample generalization error} corresponding to the $i$-th sample $Z_i$.

Then, we define the expected generalization error as follows:
\begin{align*}
    \overline{\text{gen}} &:= \mathbb{E}_{S_{te},S_{tr},W}[ \hat{l}_{S_{te}}(W)]
    - \mathbb{E}_{\bar{W}}[\mathbb{E}_{\bar{Z}_{te}}l(\bar{W},\bar{Z}_{te})]\\
    &= \mathbb{E}_{S_{te},S_{tr},W}[ \hat{l}_{S_{te}}(W)] -\mathbb{E}_{\bar{S}_{te},\bar{S}_{tr},\bar{W}}[\hat{l}_{\bar{S}_{te}}(\bar{W})],
\end{align*}
where the first expectation is calculated with respect to the joint distribution $P_{S_{te},S_{tr},W} = P_{S_{te},S_{tr}} P_{W\mid S_{tr}}$ and the second expectation is calculated with respect to the marginal distribution i.e. $P_{S_{te}} P_{S_{tr}} P_{W}$.

In a special case of the above learning scenario,
the test and training data are perfectly correlated, i.e., $S_{te}$ and $S_{tr}$ can be represented by a single variable $S$. Then, the learning scenario becomes a weaker version of classical learning, studied in \cite{XR_2017,Esposito21,Modak21,HGDR2024}
and discussed in the introduction of this manuscript.

\subsection{Quantum Learning Framework}\label{subsec:quantum_frame}
To get a motivation for defining an input to a quantum learning algorithm, let us first define an input to a classical learning algorithm in a quantum formalism. This can be given as the following state,
\begin{equation*}
    \rho^{}_{\text{class}} := \bbE_{S_{te},S_{tr} \sim P^n_{Z_{te},Z_{tr}}} \left[\ketbrasys{S_{te}}{\text{Te}} \otimes \ketbrasys{S_{tr}}{\text{Tr}}\right].
\end{equation*}

In the quantum learning scenario, a quantum learning algorithm $\cA_{Q}$, apart from the systems $\text{Te}$ and $\text{Tr}$, also has access to another quantum system whose state is $\rho(S_{te},S_{tr}) \in \cH^{D}$ (where $\cH^{D}$ is denoted as data Hilbert space). Thus, the input to $\cA_{Q}$ can be represented by the following classical quantum state,
\begin{equation}
    \rho := \bbE_{S_{te},S_{tr} \sim P^n_{Z_{te},Z_{tr}}} \left[\ketbra{S_{te}} \otimes \ketbra{S_{tr}} \otimes \rho(S_{te},S_{tr})\right].\label{data_state}
\end{equation}

{Typically in the classical learning literature \cite{XR_2017,Esposito21,Modak21,HGDR2024}, the testing and training data are assumed to be perfectly correlated, i.e. $S_{te}$ and $S_{tr}$ can be represented by a single random variable $S$. Likewise, in \cite{Caro23}, it is assumed that $S_{te}$ and $S_{tr}$ are perfectly correlated. However, in \cite{Caro23}, the authors observe that the analysis for the performance of any quantum learning algorithm where $S_{te}$ and $S_{tr}$ are not perfectly correlated, is very similar to the analysis in the case when $S_{te}$ and $S_{tr}$ are perfectly correlated. Therefore, throughout this manuscript, we build on this assumption.

Thus, under the above assumption, the input to a quantum learning algorithm $\cA_Q$ can be represented as the following state,
\begin{equation}
    \rho := \bbE_{S \sim P^{n}} \left[\ketbra{S} \otimes \rho(S)\right].\label{data_state2}
\end{equation}}

{ Upon getting the input mentioned in \eqref{data_state2}, a quantum learning algorithm $\cA_{Q}$ applies certain POVMs and CP-TP maps over the quantum system residing in $\cH^{D}$ for learning. Due to the irreversible perturbation of $\rho(S)$ during learning, the processed state cannot be further used to evaluate the empirical loss. This is because we need the unperturbed quantum state to calculate the empirical loss.\\
\hspace*{10pt} The authors in \cite{Caro23} avoid this issue by bi-partioning $\cH^{D}$ into $\cH^{te} $ and $ \cH^{tr}$ i.e. $\cH^{D} := \cH^{te} \otimes \cH^{tr}$, such that $\cH^{te} \cong \mathbb{C}^{d_1}, \cH^{tr} \cong \mathbb{C}^{d_2} : 1 < d_1,d_2 < \infty$ are the test data Hilbert space and the train data Hilbert space respectively. Thus, $\forall s  \in \cZ^n, \rho(s) \in \cH^{te} \otimes \cH^{tr}$ now has two components i.e. the testing component residing in $\cH^{te}$ and the training component residing in $\cH^{tr}$. The learning algorithm acts only on $\cH_{tr}$ while learning, but the loss is calculated over the whole data Hilbert space $\cH^{D}$. In general, $\forall s\in \cZ^n, \rho(s)$ might be correlated or even entangled across $\cH^{te}$ and $\cH^{tr}$. We now formally illustrate a general quantum learning algorithm in \cite{Caro23}.}

 The quantum learner $\cA_{Q}$ is presumed to have a collection of extractor POVMs $\left\{\{E^{\cA_{Q}}_{s}(w)\}_{w \in \cW}\right\}_{s \in \cZ^n}$ which acts over $\cH^{tr}$, where $\cW$ is assumed to be a discrete measurable space which is called the Hypothesis space and each $w \in \cW$ is a classical hypothesis given as a map $w: \cX \to \cY$. Thus, $\cW$ can be considered as a subspace of $\cY^{\cX}$. The quantum learner $\cA_{Q}$ uses this collection of POVMs to extract classical information from training data states.

 Conditioned on the classical data $s := (z_1,\cdots,z_n) \in \cZ^{n}$, $\cA_{Q}$ performs the measurement using the POVM $\{E^{\cA_{Q}}_{s}(w)\}_{w \in \cW}$ and records the outcome $w$ classically, which results the following post measurement state over $\cH^{D}$,
 \begin{align}
     \rho^{\cA_{Q}}(w,s) &:= \frac{\sqrt{\left(\bbI_{\cH^{te}} \otimes E^{\cA_{Q}}_{s}(w)\right)} \rho(s)\sqrt{\left(\bbI_{\cH^{te}} \otimes E^{\cA_{Q}}_{s}(w)\right)}}{\tr[E^{\cA_{Q}}_{s}(w) \rho_{tr}(s)]},
 \end{align}
 $\rho_{tr}(s) = \tr_{te}[\rho(s)]$ and $\forall w \in \cW, s \in \cZ^{n}$, we define a conditional probability distribution $P^{\cA_{Q}}(w|s) := \tr[E^{\cA_{Q}}_{s}(w) \rho_{tr}(s)]$. 

The quantum learner also has an apriori access to a collection of quantum channels (CP-TP) $\left\{\Lambda^{\cA_{Q}}_{w,s}: \cT(\cH^{tr}) \rightarrow \cT(\cH^{{hyp}})\right\}_{\substack{w \in \cW \\ s \in \cZ^n}}$, {where we denote $\cH^{{hyp}}$ to be a quantum hypothesis Hilbert space.} Conditioned on both $w$ and $s$, $\cA_{Q}$ now applies the channel $\Lambda_{w,s}$ on the post-measurement state $\rho^{\cA_{Q}}(w,s)$ and the resultant state is given as follows:
\begin{align*}
    \sigma^{\cA_{Q}}(w,s) &:= \left(\bbI_{\cH^{te}} \otimes \Lambda_{w,s}\right) \left(\rho^{\cA_{Q}}(w,s)\right).
\end{align*}

 Hence, the overall action of $\cA_{Q}$ over the data state $\rho$ leads us to the following CQ state,
 \begin{align*}
     \sigma^{\cA_{Q}} &= \bbE_{S \sim P^n}\left[\ketbra{S} \otimes \bbE_{W \sim P^{\cA_{Q}}(\cdot|S)}\left[\ketbra{W} \otimes \sigma^{\cA_{Q}}(W,S)\right]\right]\\
     &= \bbE_{(W,S) \sim {P}^{\cA_{Q}}_{WS}}\left[\ketbra{W} \otimes\ketbra{S}\otimes\sigma^{\cA_{Q}}(W,S)\right]\\
    &= \bbE_{W \sim {P}^{\cA_{Q}}_{W}}\bbE_{S \sim {P}^{\cA_{Q}}_{S|W}(\cdot|W)}\left[\ketbra{W} \otimes\ketbra{S}\otimes\sigma^{\cA_{Q}}(W,S)\right],
 \end{align*}
 where $\forall (w,s) \in \cW \times \cZ^n , {P}^{\cA_{Q}}_{WS}(w,s) := P^{\cA_{Q}}(w|s) P^n(s)$, ${P}^{\cA_{Q}}_{W}(w) := \sum_{s \in Z^n}{P}^{\cA_{Q}}_{WS}(w,s)$ and ${P}^{\cA_{Q}}_{S|W}(s|w):= \frac{{P}^{\cA_{Q}}_{WS}(w,s)}{{P}^{\cA_{Q}}_{W}(w)}$ {is the posterior distribution of the data given the hypothesis}. In the following discussion, we define how to quantize the loss or error induced from the resultant state $\sigma^{\cA_{Q}}$. 

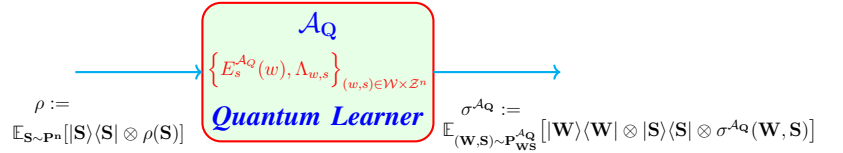
\begin{wrapfigure}{r}{0.6\textwidth}
\centering
\resizebox{110mm}{!}{
\begin{circuitikz}
\tikzstyle{every node}=[font=\large]
\draw [ color={rgb,255:red,255; green,0; blue,0} , fill= green!10, line width=1.3pt , rounded corners = 11.4] (4.25,10.25) rectangle  node {$\left\{E^{\cA_Q}_{s}(w),\Lambda_{w,s}\right\}_{(w,s) \in \cW\times\cZ^n}$} (9.25,7.25);
\draw [->, line width=1.3pt,color = cyan](1.5,8.75) to[short] (4.25,8.75);
\draw [ -
>, line width=1.3pt,color = cyan](9.25,8.75) to[short] (12,8.75);
\node [font=\LARGE,  blue] at (6.75,7.75) {\textbf{\textit{Quantum Learner}}};
\node [font=\LARGE,  blue] at (6.75,9.75) {$\mathbf{\cA_{Q}}$};

\node [font=\large, color = black] at (1,8) {$\mathbf{\rho:=}$};
\node [font=\large, color = black] at (2,7.4) {$\mathbf{\bbE_{S \sim P^{n}} \left[\ketbra{S} \otimes \rho(S)\right]}$};

\node [font=\large, color = black] at (10.5,8) {$\mathbf{\sigma^{\cA_Q}:=}$};
\node [font=\large, color = black] at (13.5,7.4) {$\mathbf{\bbE_{(W,S) \sim {P}^{\cA_{Q}}_{WS}}\left[\ketbra{W} \otimes\ketbra{S}\otimes\sigma^{\cA_{Q}}(W,S)\right]}$};
\end{circuitikz}
}\caption{Quantum learning algorithm structure proposed by \cite{Caro23}.}\label{fig:quantum_diag_caro}
\end{wrapfigure}

In a classical learning scenario, a loss function is generally defined as a map $l: \cW \times \cZ^n \to \bbR$. However, in the quantum learning scenario, since the data and hypothesis induced by the quantum learner are embedded into quantum states, we use observables (operators) and consider the expected value of the observables with respect to the quantum states that represent the data and the hypothesis induced from it. In \cite{Caro23}, the authors consider a family of non-negative self-adjoint loss observables $\left\{\hat{L}(w,s) \in \cL(\cH^{te} \otimes \cH^{hyp})\right\}_{\substack{(w,s) \in \cW \times \cZ^n}}$. Using these loss observables, we now define the two types of loss in terms of average values of the loss observables.
\begin{definition}[Empirical (observed) loss]\label{emp_loss_def}
    The expected empirical loss $\hat{l}_{\rho}(w,s)$ of ${\cA_{Q}}$ with an input $\rho$ and observable $\hat{L}(w,s)$ is defined as follows,
    \begin{align*}
        \hat{l}_{\rho}(w,s) &:= \tr[\hat{L}(w,s)\sigma^{\cA_{Q}}(w,s)].
    \end{align*}
\end{definition}
\begin{definition}[Expected empirical (observed) loss {\cite[Definition $11$]{Caro23}}]\label{exp_emp_loss_def}
    The expected empirical loss $L_{\rho}$ of ${\cA_{Q}}$ with an input $\rho$ and observables $\left\{\hat{L}(w,s)
    \right\}_{\substack{(w,s) \in \cW \times \cZ^n}}$ is defined as follows,
    \begin{align*}
        \hat{L}_{\rho} &:= \bbE_{(W,S) \sim {P}^{\cA_{Q}}_{WS}}[\hat{l}_{\rho}(W,S)].
    \end{align*}
\end{definition}

\begin{definition}[True loss \cite{Caro23}]\label{Ctrue_loss_def}
    The true loss $l^{(\textnormal{old})}_{\rho}(w)$ (here we use the phrase $\textnormal{(old)}$ since we propose a novel definition for true loss in Definition \ref{true_loss_def}) of ${\cA_{Q}}$ with an input $\rho$ and observables $\left\{\hat{L}(w,s)\right\}_{\substack{s \in \cZ^n}}$ is defined as follows,
    \begin{align*}
        l^{(\textnormal{old})}_{\rho}(w) &:= \bbE_{\overline
{S} \sim  P^{n}}\left[\tr\left[\hat{L}(w,\overline{S})\left(\rho_{te}(\overline{S}) \otimes \sigma^{\cA_{Q}}_{hyp}(w,\overline{S})\right)\right]\right],
    \end{align*}
\end{definition}

\begin{definition}[True loss \textbf{proposed}]\label{true_loss_def}
    The true loss $l_{\rho}(w)$ of ${\cA_{Q}}$ with an input $\rho$ and observables $\left\{\hat{L}(w,s)\right\}_{\substack{s \in \cZ^n}}$ is defined as follows,
    \begin{align}
        l_{\rho}(w) &:= \bbE_{\overline
{S} \sim  P^{n}}\left[\tr\left[\hat{L}(w,\overline{S})\left(\rho_{te}(\overline{S}) \otimes \sigma^{\cA_{Q}}_{hyp}(w)\right)\right]\right],\nn
    \end{align}
\end{definition}

where $\rho_{te}(\overline{S}) := \tr_{tr}\left[\rho(\overline{S})\right] $, $\sigma^{\cA_{Q}}_{hyp}(w) := \tr_{te}\left[\sigma^{\cA_{Q}}(w)\right]$, where for any $w \in \cW,$ $\sigma^{\cA_{Q}}(w):= \mathbb{E}_{S \sim P^{\cA_Q}_{S|W}(\cdot|w)} \left[\sigma^{\cA_{Q}}(w,S)\right].$
\begin{definition}[Expected true loss {\cite[Definition $12$]{Caro23}}]\label{Cexp_true_loss_def}
    The expected true loss $L^{(\textnormal{old})}_{\rho}$ of ${\cA_{Q}}$ with an input $\rho$ and observables $\left\{\hat{L}(w,s)\right\}_{\substack{(w,s) \in \cW \times \cZ^n}}$ is defined as follows,
    \begin{align*}
        L^{(\textnormal{old})}_{\rho} &:= \bbE_{\overline{W} \sim {P}^{\cA_{Q}}_{{W}}}[l^{(\textnormal{old})}_{\rho}(\overline{W})]\\
        &=\bbE_{(\overline{W},\overline{S}) \sim {P}^{\cA_{Q}}_{{W}} \times P^{n}}\left[\tr\left[\hat{L}(\overline{W},\overline{S})\left(\rho_{te}(\overline{S}) \otimes \sigma^{\cA_{Q}}_{hyp}(\overline{W},\overline{S})\right)\right]\right].
    \end{align*}
\end{definition}
\begin{definition}[Our definition for Expected true loss] \label{exp_true_loss_def}
    The expected true loss $L_{\rho}$ of ${\cA_{Q}}$ with an input $\rho$ and observables $\left\{\hat{L}(w,s)\right\}_{\substack{(w,s) \in \cW \times \cZ^n}}$ is defined as follows,
    \begin{align*}
        L_{\rho} &:= \bbE_{\overline{W} \sim {P}^{\cA_{Q}}_{{W}}}[l_{\rho}(\overline{W})]\\
        &=\bbE_{(\overline{W},\overline{S}) \sim {P}^{\cA_{Q}}_{{W}} \times P^{n}}\left[\tr\left[\hat{L}(\overline{W},\overline{S})\left(\rho_{te}(\overline{S}) \otimes \sigma^{\cA_{Q}}_{hyp}(\overline{W})\right)\right]\right] \\
        &=\bbE_{(\overline{W},\overline{S}_{tr},\overline{S}_{te}) 
        \sim {P}^{\cA_{Q}}_{{W},S} \times P^{n}}\left[\tr\left[\hat{L}(\overline{W},\overline{S}_{te})\left(\rho_{te}(\overline{S}_{te}) \otimes \sigma^{\cA_{Q}}_{hyp}(\overline{W},\overline{S}_{tr})\right)\right]\right].
    \end{align*}
    \end{definition}

\begin{remark}\label{motivation_remark}
    \begin{itemize}
    \item[]%
\item 
Reference \cite{Caro23} adopts Definition \ref{Ctrue_loss_def}
as the true loss. However, the quantum testing system is correlated to the quantum hypothesis 
system in the definition of $l^{(\textnormal{old})}_{\rho}(w)$ after taking the average with respect to the classical variable $\overline{S}$.
Even though the testing classical data 
$S_{te}$ is perfectly correlated to the training classical data $S_{tr}$, the testing classical data $\bar{S}_{te}$ needs to be 
independent of the training classical data $\bar{S}_{tr}$ and 
the classical hypothesis $\bar{W}$ in the definition of the true loss.
When we have 
$\bar{W} - \bar{S}_{tr} \perp \bar{S}_{te}$ in the same way as the classical case,
it is suitable to define $l_{\rho}(w)$ in Definition \eqref{true_loss_def}
by choosing $\overline{S}$ to be $\bar{S}_{te}$.
\item
The expected true loss is defined as the expectation of 
the true loss. 
In Definition \ref{Cexp_true_loss_def} of 
the expected true loss,
the testing classical data 
is perfectly correlated with the training classical data.
However, in the definition of the expected true loss,
the testing classical data $\bar{S}_{te}$ needs to be 
independent of the training classical data $\bar{S}_{tr}$ and 
the classical hypothesis $\bar{W}$ in the same way as the classical case.
Hence, we adopt Definition \ref{exp_true_loss_def}.
The detailed derivation of Definitions \ref{true_loss_def}
and \ref{exp_true_loss_def} is given in Subsection \ref{explanantion_new_def}.
\end{itemize}
\end{remark}

In practice, after the testing and training systems are well-prepared and the learner applies the quantum operation $\cA_{Q}$, measuring the observable ${L}(w,s)$ is expected to yield an outcome close to the expected empirical loss $\hat{L}_{\rho}$. However, our primary interest lies in the expected true loss $L_{\rho}$ of ${\cA_{Q}}$. Consequently, we are interested in their difference, which defines the generalization error.

\begin{definition}[Generalization Error]\label{gen_ws_error}
    For any $(w,s) \in \cW \times \cZ^n$ we define the generalization error $\text{\textnormal{gen}}(w,s)$ with an input $\rho$ and observables $\left\{\hat{L}(w,s)\right\}_{\substack{(w,s) \in \cW \times \cZ^n}}$ is defined as follows,
    \begin{align}
        \text{\textnormal{gen}}(w,s) = l_{\rho}(w) - \hat{l}_{\rho}(w,s).\label{gen_ws_def}
    \end{align}
\end{definition}
\begin{definition}[Expected generalization error]\label{gen_ws_error_exp}
    The expected generalization error $\overline{\text{gen}}$ of ${\cA_{Q}}$ with an input $\rho$ and observables $\left\{\hat{L}(w,s) \right\}_{\substack{(w,s) \in \cW \times \cZ^n}}$ is defined as follows,
    \begin{align*}
        \overline{\text{\textnormal{gen}}} &:= \bbE_{(W,S) \sim P^{\cA_Q}_{WS}}\left[\text{\textnormal{gen}}(W,S)\right]\\
        &=  \bbE_{(W,S) \sim P^{\cA_Q}_{WS}}\left[l_{\rho}(W) - \hat{l}_{\rho}(W,S)\right]\\
        &= \bbE_{W \sim P^{\cA_Q}_{W}}\left[l_{\rho}(W)\right] - \bbE_{(W,S) \sim P^{\cA_Q}_{WS}}\left[\hat{l}_{\rho}(W,S)\right]\\
        &= L_{\rho} - \hat{L}_{\rho}.
    \end{align*}
\end{definition}
\begin{remark}
    See \cite[Appendix C]{Caro23} for various applications of the quantum learning framework discussed above.
\end{remark}
{\begin{remark}
    From Definitions \ref{exp_emp_loss_def},\ref{exp_true_loss_def} and \ref{gen_ws_error_exp}, we can come up with the classical notions of expected empirical loss, expected true loss, and expected generalization error by exploiting two special cases of the quantum learning framework mentioned above.\\
         \textnormal{\bf (a)} The first special case is based on a very trivial assumption, i.e. all the involved quantum systems are trivial ($\cH^{D} = \cH^{hyp} = \mathcal{C}$). Then, for each $(w,s) \in \cW \times \cZ^n$, the loss observable $\hat{L}(w,s)$ becomes a scalar quantity (we consider $\bbR$ instead of $\mathcal{C}$ for simplicity) and can be thought of as the value of a classical loss function when passed $(w,s)$ as an input to it.\\
         \textnormal{\bf (b)} Secondly, for each $(w,s) \in \cW \times \cZ^n$, if we consider the loss observable of a special form $\hat{L}(w,s) : =\hat{l}_{s}(w). \bbI^{\cH^{te}\otimes \cH^{hyp}}$, ( where $\times \cZ^n$, if we consider the loss observable of a special form $\hat{L}(w,s)$ is defined in \eqref{class_emp_loss_intro}), then, for each $(w,s) \in \cW \times \cZ^n$, we are left with the following,
         \begin{align}
             \text{\textnormal{gen}}(w,s) &\overset{a}{=} l_{\rho}(w) - \hat{l}_{\rho}(w,s)\nn\\
             &\overset{b}{=}  \bbE_{\overline
{S} \sim  P^{n}}\left[\tr\left[\hat{L}(w,\overline{S})\left(\rho_{te}(\overline{S}) \otimes \sigma^{\cA_{Q}}_{hyp}(w)\right)\right]\right] - \tr[\hat{L}(w,s)\sigma^{\cA_{Q}}(w,s)]\nn\\
&=  \bbE_{\overline
{S} \sim  P^{n}}\left[\hat{l}_{\overline{S}}(w)\tr\left[\left(\rho_{te}(\overline{S}) \otimes \sigma^{\cA_{Q}}_{hyp}(w)\right)\right]\right] - \hat{l}_{s}(w)\tr[\sigma^{\cA_{Q}}(w,s)]\nn\\
&\overset{c}{=} \bbE_{\overline
{S} \sim  P^{n}}\left[\hat{l}_{\overline{S}}(w)\right] - \hat{l}_{s}(w)\nn\\
&= \bbE_{\overline
{Z} \sim  P}\left[l(w,\overline{Z})\right] - \hat{l}_{s}(w),\label{gen_ws_error_red}
         \end{align}
where $a$ follows from Definition \ref{gen_ws_error}, $b$ follows from Definitions \ref{emp_loss_def} and \ref{true_loss_def}, $c$ follows since for each $(w,s) \in \cW \times \cZ^n$, $\sigma^{\cA_Q}(w,s)$ and $\rho_{te}(s)\otimes\sigma^{\cA_Q}_{hyp}(w,s)$ are quantum states. Observe that \eqref{gen_ws_error_red} coincides with the classical generalization error 
(mentioned in \eqref{class_gen_ws_intro}) in absolute value. 
\end{remark}}

\subsection{Justification for the validity of the proposed definition of expected true loss (Definition \ref{exp_true_loss_def})}\label{explanantion_new_def}
    To see the validity of Definition \ref{exp_true_loss_def},
we discuss the general one-shot setting.
Assume that 
$\cH_{te}$,
$\cH_{tr}$,
and $\cH_{hyp}$ are 
the test space, the training space, and the hypothesis space.
Initially, we have the initial state $\rho$ on  $\cH_{te} \otimes \cH_{tr}$.
The learner applies a CP-TP map $\Lambda$
from $\cH_{tr}$ to $\cH_{hyp}$.
Then, we have the final state $ (id \otimes \Lambda)(\rho) $
on $\cH_{te} \otimes \cH_{hyp}$.
The loss is determined by a Hermitian matrix $L$ on $\cH_{te} \otimes \cH_{hyp}$ as follows.
The empirical loss is given as
\begin{align}
\tr \left[L (id \otimes \Lambda)(\rho)\right].
\end{align} 
The true loss is given as
\begin{align}
\tr \left[L (\rho_{te} \otimes \Lambda(\rho_{tr}))\right].
\end{align} 

We now consider the case when 
$\cH_{te}$,
$\cH_{tr}$,
$\cH_{hyp}$ are composed of the classical and quantum parts, 
$\cH_{te}^C$,
$\cH_{tr}^C$,
$\cH_{hyp}^C$
and 
$\cH_{te}^Q$,
$\cH_{tr}^Q$,
$\cH_{hyp}^Q$.
Also, we assume that the classical information $S$ of the test space is identical to the classical information of the training space.
Hence, the initial state $\rho$ on 
$\cH_{te} \otimes \cH_{tr}$ can be written as
\begin{align}
\rho= \sum_s P_S(s)|s,s\rangle \langle s,s| \otimes 
\rho(s),
\end{align}
where $\rho(s)$ is a state on $\cH_{te}^Q \otimes \cH_{tr}^Q$.
The learner's operation $\Lambda$ is written as an instrument 
$\Gamma_s=(\Gamma_{w|s})$ as follows.
For a state $\sum_s Q(s)|s\rangle \langle s |\otimes \sigma(s)$,
we have,
\begin{align}
\Lambda\left( \sum_s Q(s)|s\rangle \langle s |\otimes \sigma(s)\right)
= \sum_s Q(s)
|w\rangle \langle w |\otimes \Gamma_{w|s}\left(\sigma(s)\right).
\end{align}
Then, the resultant state is as follows,
\begin{align}
 (id \otimes \Lambda)(\rho)
=\sum_{s,w} P_S(s)  
|s\rangle \langle s| \otimes  |w\rangle \langle w| 
\otimes (id \otimes \Gamma_{w|s})( \rho(s)).
\end{align} 
Also, the operator $L$ can be written as,
\begin{align}
L= \sum_{s',w'} |s'\rangle \langle s' |\otimes
|w'\rangle \langle w' |\otimes L(w',s').
\end{align}
Then, the empirical loss is given as, 
\begin{align}
\tr \left[L (id \otimes \Lambda)(\rho)\right]
=&
\tr \Big( \sum_{s',w'} |s'\rangle \langle s' |\otimes
|w'\rangle \langle w' |\otimes L(w',s')\Big)
\Big(\sum_{s,w} P_S(s)  
|s\rangle \langle s| \otimes  |w\rangle \langle w| 
\otimes (id \otimes \Gamma_{w|s})( \rho(s))\Big)\notag\\
=&\sum_{s,w} P_S(s) \tr\left[ L(w,s)(id \otimes \Gamma_{w|s})( \rho(s))\right].
\end{align} 
Since,
\begin{align}
\Lambda(\rho_{tr})
&=\sum_{s,w} P_S(s)   |w\rangle \langle w| 
\otimes \Gamma_{w|s}( \rho_{tr}(s)), \\
\rho_{te}
&=\sum_{\bar{s}} P_S(\bar{s})   |\bar{s}\rangle \langle \bar{s}| 
\otimes \rho_{te}(\bar{s}),
\end{align} 
the true loss is given as,
\begin{align}
&\tr \left[L \left(\rho_{te} \otimes \Lambda(\rho_{tr})\right)\right]\notag\\
=&
\tr \left[ \sum_{s',w'} |s'\rangle \langle s' |\otimes
|w'\rangle \langle w' |\otimes L(w',s')\Big)
\Big(
\Big(
\sum_{\bar{s}} P_S(\bar{s})   |\bar{s}\rangle \langle \bar{s}| \otimes \rho_{te}(\bar{s})\Big)
\otimes 
\Big(\sum_{s,w} P_S(s)   |w\rangle \langle w| 
\otimes \Gamma_{w|s}( \rho_{tr}(s))\Big)\Big)
\right]\notag\\
=&
\sum_{s,\bar{s},w} P_S(s) P_S(\bar{s})\tr \left[ L(w,\bar{s})
(\rho_{te,\bar{s}}\otimes \Gamma_{w|s}( \rho_{tr}(s)))\right].\label{BIK}
\end{align} 
We define the probability
$P_{W|S}$ as,
\begin{align}
P^{\cA_Q}_{W|S}(w|s):=\tr \left[\Gamma_{w|s}( \rho_{tr}(s))\right].
\end{align}
We set $\sigma^{\cA_{Q}}_{hyp}(w,{S}):=
\frac{1}{P^{\cA_Q}_{W|S}(w|s)}
\Gamma_{w|s}( \rho_{tr}(s))$. Then, we can rewrite \eqref{BIK} as follows,

\begin{align}
\tr \left[L \left(\rho_{te} \otimes \Lambda(\rho_{tr})\right)\right]
=&
\sum_{s,\bar{s},w} P^{\cA_Q}_{WS}(w,s) P_S(\bar{s})\tr \left[ L(w,\bar{s})
(\rho_{te,\bar{s}}\otimes \sigma^{\cA_{Q}}_{hyp}(w,{S}))\right]\nn\\
=&
\sum_{\bar{s},w} P^{\cA_Q}_{W}(w) P_S(\bar{s})\tr \left[ L(w,\bar{s})
(\rho_{te,\bar{s}}\otimes \sum_{s}P^{\cA_Q}_{S|W=w}(s)\sigma^{\cA_{Q}}_{hyp}(w,{S}))\right]\nn\\
=&
\sum_{\bar{s},w} P^{\cA_Q}_{W}(w) P_S(\bar{s})\tr \left[ L(w,\bar{s})
(\rho_{te,\bar{s}}\otimes \sigma^{\cA_{Q}}_{hyp}(w))\right].\label{BIK2}
\end{align} 
Therefore, the relation \eqref{BIK2} coincides with Definition \ref{exp_true_loss_def}.

In particular, 
when the state $\rho_S(s)$ on $\cH_{te}^Q\otimes \cH_{tr}^Q$
is given a product state $\rho_{te}(s)\otimes \rho_{tr}(s)$,
we have $\sigma^{\cA_{Q}}(w,s)=\rho_{te}(s)\otimes
\sigma_{hyp}^{\cA_{Q}}(w,s)$.
Since 
\begin{align}
\hat{l}_{\rho}(w,s) = \tr[\hat{L}(w,s)\sigma^{\cA_{Q}}(w,s)]
=
\tr[\hat{L}(w,s) (\rho_{te}(s)\otimes \sigma_{hyp}^{\cA_{Q}}(w,s))],
\end{align}
the expected empirical loss is 
\begin{align}
\hat{L}_{\rho} &= \bbE_{(W,S) \sim {P}^{\cA_{Q}}_{WS}}[\hat{l}_{\rho}(W,S)]
\notag\\
=& \bbE_{(W,S) \sim {P}^{\cA_{Q}}_{WS}}
\tr[\hat{L}(W,S) (\rho_{te}(S)\otimes \sigma_{hyp}^{\cA_{Q}}(W,S))],
\end{align}
Hence, the difference between
the expected empirical loss and the expected true loss defined by \cite[Definition $12$]{Caro23}
is characterized by the difference between 
the joint distribution ${P}^{\cA_{Q}}_{WS}$ and the product distribution
${P}^{\cA_{Q}}_{{W}} \times P^{n}$.

\section{Discussion on previous works}\label{sec:prev_work}
In this section, we first discuss an upper-bound on the expected generalization error for bounded loss functions/observables. We then use this bound to discuss the classical results obtained in \cite{XR_2017,Bu_2020,Modak21} and the quantum results discussed in \cite{Caro23}.

\subsection{Classical learning paradigm}
In a classical learning scenario, using \eqref{class_exp_true_loss}, we can write the following form of expected generalization error \\
$\bbE_{(W,S) \sim P^{\cA}_{WS}}[\text{gen}(W,S)]$ (where $\text{gen}(W,S)$ is defined in \eqref{class_gen_ws_intro}),

\begin{align}
        \bbE_{(W,S) \sim P^{\cA}_{WS}}[\text{gen}(W,S)] &= \bbE_{(W,S) \sim P^{\cA}_{WS}}[\hat{l}_{S}(W)] - \bbE_{(\overline{W},\overline{S}) \sim P^{\cA}_{W} \times P_{S}}[\hat{l}_{\overline{S}}(\overline{W})].\nn
\end{align}
If we assume that the loss function is bounded, then, using the variational form for $L_1$ distance we can obtain an upper-bound on $\bbE_{(W,S) \sim P^{\cA}_{WS}}[\text{gen}(W,S)].$
In particular, assume that $\forall (w, z) \in \cW \times \cZ,$ $l(w,z) < \tau,$ where $\tau < \infty.$ Then,

\begin{align}
    \bbE_{(W,S) \sim P^{\cA}_{WS}}[\text{gen}(W,S)] &= \bbE_{(W,S) \sim P^{\cA}_{WS}}[\hat{l}_{S}(W)] - \bbE_{(\overline{W},\overline{S}) \sim P^{\cA}_{W} \times P_{S}}[\hat{l}_{\overline{S}}(\overline{W})]\nn\\
    &= \frac{1}{n}\sum_{i=1}\bbE_{(W,Z_i) \sim P^{\cA}_{WZ_i}}[\hat{l}_{Z_i}(W)] - \bbE_{(\overline{W},\overline{Z}_i) \sim P^{\cA}_{W} \times P}[\hat{l}_{\overline{Z}_i}(\overline{W})]\nn\\
    &\leq \frac{1}{n}\sum_{i=1}\sup_{f:\norm{f}{\infty} \leq \tau} \left( \bbE_{(W,Z_i) \sim P^{\cA}_{WZ_i}}[f(W,Z_i)] - \bbE_{(\overline{W},\overline{Z}_i) \sim P^{\cA}_{W} \times P}[f(\overline{W},\overline{Z}_i)]\right)\nn\\
    &\overset{a}{=} \frac{1}{n}\sum_{i=1}\tau\norm{P^{\cA}_{WZ_i} - P^{\cA}_{W} \times P}{1}
    ,\label{naive_exp_gen_bound}
\end{align}

where, $a$ follows from \eqref{l1varformeq} of Fact \ref{fact_var_form_l1_distance}. 
The bound obtained in \eqref{naive_exp_gen_bound} holds under a strict assumption that the loss function $l$ is bounded. Further, for any $p>1$, one can obtain the following upper-bound on $\bbE_{(W,S) \sim P^{\cA}_{WS}}[\text{gen}(W,S)]$, under a stricter assumption that $\norm{l}{q} \leq \tau$, for a $q<\infty$ such that  $\frac{1}{p} + \frac{1}{q} = 1$,
\begin{align}
\bbE_{(W,S) \sim P^{\cA}_{WS}}[\text{gen}(W,S)] &\leq \frac{1}{n}\sum_{i=1}\sup_{f:\norm{f}{q} \leq \tau} \left( \bbE_{(W,Z_i) \sim P^{\cA}_{WZ_i}}[f(W,Z_i)] - \bbE_{(\overline{W},\overline{Z}_i) \sim P^{\cA}_{W} \times P}[f(\overline{W},\overline{Z}_i)]\right)\nn\\
    &\overset{a}{=} \frac{1}{n}\sum_{i=1}\tau\norm{P^{\cA}_{WZ_i} - P^{\cA}_{W} \times P}{p},\label{naive_exp_gen_bound_tight}
\end{align}
where $a$ follows from \eqref{lpvarformeq} of Fact \ref{fact_var_form_l1_distance}. Observe that \eqref{naive_exp_gen_bound_tight} is a comparatively tighter upper-bound than \eqref{naive_exp_gen_bound}, since $\norm{\cdot}{p}$ is a decreasing function of $p$.

Xu and Raginsky in \cite{XR_2017}, relaxed these strict assumptions mentioned above and state an upper-bound on the absolute value of the expected generalization error under the following assumption.

\begin{assumption}\label{class_sub_gaussian_assumption}(classical sub-Gaussianity assumption)
For each $w \in \cW,$ $l(w,Z)$ for some $0 < \tau < \infty$ and any $\lambda \in \bbR$ under the distribution $P$, satisfies the following:
\begin{equation*}
    \log \bbE_{Z \sim P}\left[e^{\lambda (l(w,Z) - \bbE_{Z \sim P}[l(w,Z)])}\right] \leq \frac{\lambda^2\tau^2}{2} .
\end{equation*}   
\end{assumption}

Using Assumption \ref{class_sub_gaussian_assumption}, Xu and Raginsky in \cite{XR_2017} proved the following.

\begin{proposition}[{\cite[Theorem $1$]{XR_2017}}]\label{XR_2017_result}
    Suppose for each  $w \in \cW$, $l(w,Z)$ satisfies Assumption \ref{class_sub_gaussian_assumption} for some $0 < \tau < \infty$, then,
    \begin{equation*}
        \abs{\bbE_{(W,S) \sim P^{\cA}_{WS}}[\textnormal{gen}(W,S)]} \leq \sqrt{\frac{2\tau^2}{n} I[S;W]},
    \end{equation*}
    where $I[S;W]$ is calculated with respect to $P^{\cA}_{WS}$.
\end{proposition}
Later, Bu et al. in \cite{Bu_2020} extend Proposition \ref{XR_2017_result} by proposing a tighter individual sample-based upper-bound on the expected generalization error as follows

\begin{proposition}[{\cite[Proposition $1$]{Bu_2020}}]\label{Bu_2019_result}
    Suppose for each  $w \in \cW$, $l(w,Z)$ satisfies Assumption \ref{class_sub_gaussian_assumption} for some $0 < \tau < \infty$, then,
    \begin{equation*}
        \abs{\bbE_{(W,S) \sim P^{\cA}_{WS}}[\textnormal{gen}(W,S)]} \leq \frac{1}{n}\sum_{i=1}^{n}\sqrt{2\tau^2 I[Z_i;W]},
    \end{equation*}
    where $I[Z_i;W]$ is calculated with respect to $P^{\cA}_{WZ_i}$.
\end{proposition}

Later, Modak et al. \cite{Modak21} extend the techniques of \cite{Bu_2020} and generalize the individual sample-based upper-bound on the expected generalization error mentioned in Proposition \ref{Bu_2019_result} in terms of $\alpha$-R\'enyi divergence, given in Proposition \ref{Modak_2021_result} below. To do so, they require the following additional sub-Gaussianity assumptions over Assumption \ref{class_sub_gaussian_assumption}.

\begin{assumption}\label{class_sub_gaussian_assumption_modak}(classical sub-Gaussianity assumption)
For each $w \in \cW, i \in [n],$ $l(w,Z_i)$ for some $\tau > 0$ and any $\lambda \in \bbR$ under the distribution $P$ and $P_{Z_i|W=w}$, satisfies the following:
\begin{align*}
    \log \bbE_{Z_i \sim P_{Z_i|W=w}}\left[e^{\lambda (l(w,Z_i) - \bbE_{Z_i \sim P_{Z_i|W=w}}[l(w,Z_i)])}\right] &\leq \frac{\lambda^2\tau^2}{2}.
\end{align*}   
\end{assumption}

\begin{proposition}[{\cite[Theorem $1$ and Remark $2$]{Modak21}}]\label{Modak_2021_result}
    Suppose for each  $w \in \cW,i\in [n]$, $l(w,Z_i)$ satisfies Assumptions \ref{class_sub_gaussian_assumption} and \ref{class_sub_gaussian_assumption_modak} for some $0 < \tau < \infty$, then,
    \begin{equation*}
        \abs{\bbE_{(W,S) \sim P^{\cA}_{WS}}[\textnormal{gen}(W,S)]} \leq \begin{cases}
            \frac{1}{n}\sum_{i=1}^{n}\bbE_{W \sim P^{\cA}_{W}}\left[\sqrt{\frac{2\tau^2 D^{c}_{\alpha}(P_{Z_i|W} \| P)}{\alpha} }\right], & \text{if } \alpha \in (0,1),\\
            \frac{1}{n}\sum_{i=1}^{n}\bbE_{W \sim P^{\cA}_{W}}\left[\sqrt{2\tau^2D^{c}_{\alpha}(P_{Z_i|W} \| P)}\right], & \text{if } \alpha \in (1,\infty).\\
            \end{cases}
    \end{equation*}
    
\end{proposition}

It is important to note that for $\alpha \in (0,1)$, the upper-bound obtained in Proposition \ref{Modak_2021_result} can potentially be tighter than that obtained in Proposition \ref{Bu_2019_result}. Further, since $D^{c}_{\alpha}(.\|.)$ is well-defined for $\alpha \to 1$ (from Fact \ref{fact_limit_renyi_divergence}) and thus setting $\alpha \to 1$ we can recover results mentioned in Proposition \ref{Bu_2019_result}.

\subsection{Quantum learning paradigm}
We now show that one can obtain a bound analogous to the one obtained in \eqref{naive_exp_gen_bound} for the expected quantum generalization error (defined in  Definition \ref{gen_ws_error_exp}), if we assume that the loss observables are bounded. Under this strict assumption, one can easily derive an upper-bound of similar flavour to one of the main results of \cite{Caro23}. Towards this, 
 for each $(w,s) \in \cW \times \cZ^n$, we assume $-\mu\bbI \preceq \hat{L}(w,s)\preceq \mu\bbI$ for some $\mu < \infty$ and we define 
\begin{align}
    \textnormal{gen}^{(\textnormal{old})}(w,s) &:= l^{(\textnormal{old})}_{\rho}(w) - \hat{l}_{\rho}(w,s),\label{Cgen_ws_error_eq}\\
    \overline{\textnormal{gen}}^{(\textnormal{old})} &:= \bbE_{(W,S) \sim P^{\cA_Q}_{WS}}\left[\textnormal{gen}^{(\textnormal{old})}(w,s)\right]\nn\\
    &= L^{(\textnormal{old})}_{\rho} - \hat{L}_{\rho}, \label{Cexp_gen_ws_error_eq}
\end{align}
where $l^{(\textnormal{old})}_{\rho}(w)$, $\hat{l}_{\rho}(w,s)$, $L^{(\textnormal{old})}_{\rho}$ and $ \hat{L}_{\rho}$ are defined in Definitions \ref{Ctrue_loss_def}, \ref{emp_loss_def},  \ref{Cexp_true_loss_def} and \ref{exp_emp_loss_def}). Then, consider the following series of inequalities,
\begin{align}
    \hspace{10pt}{\overline{\textnormal{gen}}}^{(\textnormal{old})}
    &= \bbE_{(W,S)\sim P^{\cA_Q}_{WS}}\left[\tr\left[\hat{L}(W,S)\sigma^{\cA_Q}(W,S)\right]\right]-\bbE_{(\overline{W},\overline{S})\sim P^{\cA_Q}_{W} \times P_{S}}\left[\tr\left[\hat{L}(\overline{W},\overline{S})\left(\rho_{te}(\overline{S}) \otimes\sigma^{\cA_Q}_{hyp}(\overline{W},\overline{S})\right)\right]\right]\nn\\
    &\overset{a}{=} \bbE_{(W,S)\sim P^{\cA_Q}_{WS}}\left[\tr\left[\hat{L}(W,S)\sigma^{\cA_Q}(W,S)\right] - \tr\left[\hat{L}(W,S)\left(\rho_{te}(S) \otimes\sigma^{\cA_Q}_{hyp}(W,S)\right)\right]\right]\nn\\
    &\hspace{10pt}+ \bbE_{(W,S)\sim P^{\cA_Q}_{WS}}\left[\tr\left[\hat{L}(W,S)\left(\rho_{te}(S) \otimes\sigma^{\cA_Q}_{hyp}(W,S)\right)\right]\right] -\bbE_{(\overline{W},\overline{S})\sim P^{\cA_Q}_{W} \times P_{S}}\left[\tr\left[\hat{L}(\overline{W},\overline{S})\left(\rho_{te}(\overline{S}) \otimes\sigma^{\cA_Q}_{hyp}(\overline{W},\overline{S})\right)\right]\right]\nn\\
    &\overset{b}{\leq} \bbE_{(W,S)\sim P^{\cA_Q}_{WS}}\left[\mu\norm{\sigma^{\cA_Q}(W,S) - \rho_{te}(S) \otimes\sigma^{\cA_Q}_{hyp}(W,S)}{1}\right] \nn\\
    &\hspace{10pt}+ \bbE_{(W,S)\sim P^{\cA_Q}_{WS}}\left[\tr\left[\hat{L}(W,S)\left(\rho_{te}(S) \otimes\sigma^{\cA_Q}_{hyp}(W,S)\right)\right]\right] -\bbE_{(\overline{W},\overline{S})\sim P^{\cA_Q}_{W} \times P_{S}}\left[\tr\left[\hat{L}(\overline{W},\overline{S})\left(\rho_{te}(\overline{S}) \otimes\sigma^{\cA_Q}_{hyp}(\overline{W},\overline{S})\right)\right]\right]\nn\\
    &\overset{c}{\leq} \bbE_{(W,S)\sim P^{\cA_Q}_{WS}}\left[\mu\norm{\sigma^{\cA_Q}(W,S) - \rho_{te}(S) \otimes\sigma^{\cA_Q}_{hyp}(W,S)}{1}\right] + \mu \norm{P^{\cA_Q}_{WS} - P^{\cA_Q}_{W} \times P_{S}}{1}
    ,\label{naive_gen_exp_bound_quant}
\end{align}
where in $a$ we denote $P_{S} := P^n$, $b$ follows from \eqref{opl1varform} of Fact \ref{fact_var_form_schatten1} and $c$ follows from Fact \ref{fact_var_form_l1_distance} since for each $(w,s) \in \cW \times \cS$, $-\mu \leq \tr[\hat{L}(w,s)(\rho_{te}(S) \otimes \sigma^{\cA_Q}_{hyp}(w,s))] \leq \mu$ as $-\mu\bbI \preceq \hat{L}(w,s) \preceq \mu\bbI$. Further, using \eqref{oplpvarform} of Fact \ref{fact_var_form_schatten1}, for any $p > 1$, to obtain an upper-bound analogous to \eqref{naive_exp_gen_bound_tight},  ${\overline{\textnormal{gen}}}^{(\textnormal{old})}$, one can consider the following stricter assumption.
\begin{assumption}\label{assump_schatten_p}
    The collection of loss observables $\left\{\hat{L}(w,s)\right\}_{(w,s) \in \cW \times \cS}$ for some $0<\mu, \tau < \infty$ and $q<\infty$ such that $\frac{1}{p} + \frac{1}{q} =1$, satisfies the following,
    \begin{align}
        &\norm{\hat{L}(w,s)}{q} \leq \mu, ~\forall (w,s) \in \cW \times \cS,\label{assump_schatten_p_eq1}\\
        &\norm{l_{Q}}{q} \leq \tau, ~\text{ where } \forall (w,s) \in \cW \times \cS, ~ l_{Q}(w,s) := \tr\left[\hat{L}(w,s))\left(\rho_{te}(s) \otimes\sigma^{\cA_Q}_{hyp}(w,s)\right)\right].\label{assump_schatten_p_eq2}
    \end{align}
\end{assumption}

Then, under Assumption \ref{assump_schatten_p}, we can obtain the following upper-bound on ${\overline{\textnormal{gen}}}^{(\textnormal{old})}$, 
\begin{align}
    {\overline{\textnormal{gen}}}^{(\textnormal{old})} = &\overset{a}{=} \bbE_{(W,S)\sim P^{\cA_Q}_{WS}}\left[\tr\left[\hat{L}(W,S)\sigma^{\cA_Q}(W,S)\right] - \tr\left[\hat{L}(W,S)\left(\rho_{te}(S) \otimes\sigma^{\cA_Q}_{hyp}(W,S)\right)\right]\right]\nn\\
    &\hspace{10pt}+ \bbE_{(W,S)\sim P^{\cA_Q}_{WS}}\left[\tr\left[\hat{L}(W,S)\left(\rho_{te}(S) \otimes\sigma^{\cA_Q}_{hyp}(W,S)\right)\right]\right] -\bbE_{(\overline{W},\overline{S})\sim P^{\cA_Q}_{W} \times P_{S}}\left[\tr\left[\hat{L}(\overline{W},\overline{S})\left(\rho_{te}(\overline{S}) \otimes\sigma^{\cA_Q}_{hyp}(\overline{W},\overline{S})\right)\right]\right]\nn\\
    &\overset{a}{\leq} \bbE_{(W,S)\sim P^{\cA_Q}_{WS}}\left[\mu\norm{\sigma^{\cA_Q}(W,S) - \rho_{te}(S) \otimes\sigma^{\cA_Q}_{hyp}(W,S)}{p}\right] \nn\\
    &\hspace{10pt}+ \bbE_{(W,S)\sim P^{\cA_Q}_{WS}}\left[\tr\left[\hat{L}(W,S)\left(\rho_{te}(S) \otimes\sigma^{\cA_Q}_{hyp}(W,S)\right)\right]\right] -\bbE_{(\overline{W},\overline{S})\sim P^{\cA_Q}_{W} \times P_{S}}\left[\tr\left[\hat{L}(\overline{W},\overline{S})\left(\rho_{te}(\overline{S}) \otimes\sigma^{\cA_Q}_{hyp}(\overline{W},\overline{S})\right)\right]\right]\nn\\
    &= \bbE_{(W,S)\sim P^{\cA_Q}_{WS}}\left[\mu\norm{\sigma^{\cA_Q}(W,S) - \rho_{te}(S) \otimes\sigma^{\cA_Q}_{hyp}(W,S)}{p}\right] + \bbE_{(W,S)\sim P^{\cA_Q}_{WS}}\left[l_{Q}(W,S)\right] -\bbE_{(\overline{W},\overline{S})\sim P^{\cA_Q}_{W} \times P_{S}}\left[l_{Q}(\overline{W},\overline{S})\right]\nn\\
    &\overset{b}{\leq} \bbE_{(W,S)\sim P^{\cA_Q}_{WS}}\left[\mu\norm{\sigma^{\cA_Q}(W,S) - \rho_{te}(S) \otimes\sigma^{\cA_Q}_{hyp}(W,S)}{p}\right] + \tau \norm{P^{\cA_Q}_{WS} - P^{\cA_Q}_{W} \times P_{S}}{p}, \label{naive_gen_exp_bound_quant_tight}
\end{align}
where $a$ follows from \eqref{assump_schatten_p_eq1} and \eqref{oplpvarform} of Fact \ref{fact_var_form_schatten1}, $b$ follows from \eqref{assump_schatten_p_eq2} and Fact \ref{fact_var_form_l1_distance}. Observe that when $\mu = \tau$, \eqref{naive_gen_exp_bound_quant_tight} is a comparatively tighter upper-bound than the obtained in \eqref{naive_gen_exp_bound_quant}. This is becuase both Schatten-$p$ and $L_p$ norms are a decreasing function of $p$.

The upper-bounds obtained above in \cref{naive_gen_exp_bound_quant,naive_gen_exp_bound_quant_tight} require stringent assumptions on loss observables. However, Caro et al. in \cite{Caro23}, motivated by the results obtained in \cite{XR_2017} (mentioned as Proposition \ref{XR_2017_result}), relax the above stringent assumptions by considering the following sub-Gaussianity assumptions.
\begin{assumption}\label{sub_g_ass_Caro} (sub-Gaussianity assumptions mentioned in \cite{Caro23})
The collection of loss observables $\left\{\hat{L}(w,s)\right\}_{(w,s) \in \cW \times \cS}$ for some $0<\mu, \tau < \infty$ and any $\lambda \in \mathbb{R}$ satisfies the following,
    \begin{align}
    \log \tr\left[e^{\lambda\left(\hat{L}(w,s) -\tr\left[\hat{L}(w,s)\left(\rho_{te}(s) \otimes \sigma^{\cA_{Q}}_{hyp}(w,s)\right)\right](\bbI_{\cH_{te}} \otimes \bbI_{\cH_{hyp}})\right)}\left(\rho_{te}(s) \otimes \sigma^{\cA_{Q}}_{hyp}(w,s)\right)\right] &\leq \frac{\lambda^2\mu^2}{2}, ~~\forall (w,s) \in \cW \times \cS,\label{quantum_mgf_true_caro}\\
    \log\bbE_{S \sim P_{S}} \left[e^{\lambda\left(\tr\left[\hat{L}(w,S)\left(\rho_{te}(S) \otimes \sigma^{\cA_{Q}}_{hyp}(w,S)\right)\right] -  \bbE_{\overline{S} \sim P_{S}} \left[(\tr\left[\hat{L}(w,\overline{S})\left(\rho_{te}(\overline{S}) \otimes \sigma^{\cA_{Q}}_{hyp}(w,\overline{S})\right)\right]\right]\right)}\right] &\leq \frac{\lambda^2\tau^2}{2}, ~~ \forall w \in \cW.\label{classical_mgf_caro}
\end{align}
\end{assumption}
 
In Assumption \ref{sub_g_ass_Caro}, \eqref{quantum_mgf_true_caro} is a quantum sub-Gaussianity (see Definition \ref{quantum_sub_gaussian}) assumption. Similarly, \eqref{classical_mgf_caro} is a classical sub-Gaussianity (see Definition \ref{classical_sub_gaussian}) assumption.  We note here that if the loss observable has bounded norm, then Assumption \ref{sub_g_ass_Caro} directly follows from Corollary \ref{bounded_trace_sub_gaussianity} and Fact \ref{fact_hoefding_lemma}.  Using Assumption \ref{sub_g_ass_Caro}, \cite{Caro23} obtained the following upper-bound on the absolute value of ${\overline{\textnormal{gen}}}^{(\textnormal{old})}$ (mentioned in \eqref{Cexp_gen_ws_error_eq}), which is a quantum version of the result obtained in \cite[Theorem $1$]{XR_2017} (mentioned as Proposition \ref{XR_2017_result}).
\begin{proposition}[{\cite[Corollary $23$]{Caro23}}]\label{Caro23_result}
    Suppose for each $(w,s) \in \cW \times \cS$ $\hat{L}(w,s)$ satisfies Assumption \ref{sub_g_ass_Caro} for some $0 <\mu,\tau < \infty$. Then, the following holds,
    \begin{align}
        \abs{{\overline{\textnormal{gen}}}^{(\textnormal{old})}} \leq \bbE_{(W,S)\sim P^{\cA_Q}_{WS}}\left[\sqrt{2\mu^2 D(\sigma^{\cA_Q}(W,S) \| \rho_{te}(S) \otimes\sigma^{\cA_Q}_{hyp}(W,S))}\right] + \sqrt{2 \tau^2 I[S;W]}.\label{Caro23_result_eq}
    \end{align}
\end{proposition}

Note that under Assumption \ref{sub_g_ass_Caro}, the upper-bound obtained in Proposition \ref{Caro23_result} is weaker than the one obtained in \eqref{naive_gen_exp_bound_quant}. Further, if we use the proposed definition of quantum generalization error (see Definition \ref{gen_ws_error}), we get a modified version of the result obtained in Proposition \ref{Caro23_result}, which contains an extra quantum information theoretic quantity in terms of Petz quantum R\'enyi divergence. To prove this result, we require the following sub-Gaussian assumptions,
\begin{assumption}\label{sub_g_ass_Caro_modifed}
    The collection of loss observables $\left\{\hat{L}(w,s)\right\}_{(w,s) \in \cW \times \cS}$ for some $0<\mu, \tau < \infty$ and any $\lambda \in \mathbb{R}$ satisfies the following,
    \begin{align}
    \log \tr\left[e^{\lambda\left(\hat{L}(w,s) -\tr\left[\hat{L}(w,s)\left(\rho_{te}(s) \otimes \sigma^{\cA_{Q}}_{hyp}(w,s)\right)\right](\bbI_{\cH_{te}} \otimes \bbI_{\cH_{hyp}})\right)}\left(\rho_{te}(s) \otimes \sigma^{\cA_{Q}}_{hyp}(w,s)\right)\right] &\leq \frac{\lambda^2\mu^2}{2}, ~~\forall (w,s) \in \cW \times \cS,\label{quantum_mgf_true1_caro_mod}\\
     \log \tr\left[e^{\lambda\left(\hat{L}(w,s) -\tr\left[\hat{L}(w,s)\left(\rho_{te}(s) \otimes \sigma^{\cA_{Q}}_{hyp}(w)\right)\right](\bbI_{\cH_{te}} \otimes \bbI_{\cH_{hyp}})\right)}\left(\rho_{te}(s) \otimes \sigma^{\cA_{Q}}_{hyp}(w)\right)\right] &\leq \frac{\lambda^2\mu^2}{2}, ~~\forall (w,s) \in \cW \times \cS,\label{quantum_mgf_true2_caro_mod}\\
    \log\bbE_{S \sim P_{S}} \left[e^{\lambda\left(\tr\left[\hat{L}(w,S)\left(\rho_{te}(S) \otimes \sigma^{\cA_{Q}}_{hyp}(w)\right)\right] -  \bbE_{\overline{S} \sim P_{S}} \left[(\tr\left[\hat{L}(w,\overline{S})\left(\rho_{te}(\overline{S}) \otimes \sigma^{\cA_{Q}}_{hyp}(w)\right)\right]\right]\right)}\right] &\leq \frac{\lambda^2\tau^2}{2}, ~~ \forall w \in \cW.\label{classical_mgf_caro_mod}
\end{align}
\end{assumption}

Observe that \eqref{quantum_mgf_true2_caro_mod} is an additional sub-Gaussian assumption over Assumption \ref{sub_g_ass_Caro}, which was not required to prove Proposition \ref{Caro23_result}. However, since the first term in the proposed definition of generalization error (see Definition \ref{gen_ws_error}) involves a quantum state of the form $\left(\rho_{te}(s) \otimes \sigma^{\cA_{Q}}_{hyp}(w)\right)$, we require this additional quantum sub-Gaussianity assumption mentioned in \eqref{quantum_mgf_true2_caro_mod}, which involves a quantum state of the same form. Using Assumption \ref{sub_g_ass_Caro_modifed}, we state the following upper bound on the absolute value of the expected generalization error (see Definition \ref{gen_ws_error_exp}).
\begin{theorem}[Modified version of Proposition \ref{Caro23_result}]\label{Caro23_result_mod}
    Suppose for each $(w,s) \in \cW \times \cS,$ $\hat{L}(w,s)$ satisfies Assumption \ref{sub_g_ass_Caro_modifed} for some $0 <\mu,\tau < \infty$. Then, the following holds,
    \begin{align}
        \abs{\overline{\textnormal{gen}}} &\leq \bbE_{\substack{(W,S) \sim {P}^{\cA_{Q}}_{WS} }}\left[\sqrt{2\mu^2 D\left(\sigma^{\cA_{Q}}(W,S)||\rho_{te}(S) \otimes \sigma^{\cA_{Q}}_{hyp}(W,S)\right)} + \sqrt{2\mu^2 D\left( \sigma^{\cA_{Q}}_{hyp}(W,S)||\sigma^{\cA_{Q}}_{hyp}(W)\right)}\right] + \sqrt{2\tau^2 I[S;W]}.\label{Caro23_result_mod_eq}
    \end{align}
\end{theorem}
\begin{proof}
    See Appendix \ref{proof_Caro23_result_mod} for the proof.
\end{proof}

Observe that the second term in \eqref{Caro23_result_mod_eq} was not there in \eqref{Caro23_result_eq}. This is because of the difference between the definition of quantum generalization error proposed in \cite{Caro23} (see \eqref{Cgen_ws_error_eq}) and our proposed definition (see Definition \ref{gen_ws_error}).

\begin{remark}\label{remark_mod_san}
   As $\alpha \to 1$, $D^{\bbM}_{\alpha}(\cdot|\cdot)$ becomes equal to $D^{\bbM}(\cdot|\cdot)$ and thus as a consequence of Fact \ref{meas_renyi_var_form} and Fact \ref{data_processing_div} (data processing inequality for the quantum divergence), the upper-bounds mentioned in \eqref{Caro23_result_mod_eq} of Theorem \ref{Caro23_result_mod} can be tightened using $D^{\bbM}(\cdot|\cdot)$.

All the results obtained in this manuscript can be tightened using $D^{\bbM}_{\alpha}(\cdot|\cdot).$ However, the definition of $D^{\bbM}_{\alpha}(\cdot|\cdot)$ involves optimization over the choice of POVM. Therefore, for simplicity we prove all our results using Petz and modified sandwiched quantum R\'enyi divergences (Definition \ref{def_renyi_mod_sandwiched}). 

\end{remark}

\section{Quantum R\'enyi Divergences based Bounds on Expected Generalization error}\label{sec:gen_bound_exp}
In this section, we prove a quantum version of the results obtained in \cite{Modak21} and \cite{Esposito21}. In particular, we obtain   bounds on generalization error both in expectation and in probability in terms of quantum R\'enyi divergence and classical R\'enyi divergence. Thus, generalizing the results of  \cite{Modak21} and \cite{Esposito21}. Further, we recover the result of \cite{Caro23} for the expected generalization error.

\subsection{Bounds on the expected quantum generalization error}\label{subsec:non-iid_exp}

In this subsection, analogous to the classical sub-Gaussianity assumptions of \cite{Modak21} (mentioned as Assumption \ref{class_sub_gaussian_assumption_modak}), we require the following sub-Gaussianity assumption.

\begin{assumption}\label{Assumption_theo_gen}
The collection of loss observables $\left\{\hat{L}(w,s)\right\}_{(w,s) \in \cW \times \cS}$ for some $0<\mu, \tau < \infty$ and any $\lambda \in \mathbb{R}$ satisfies the following,
 \begin{align}
 \log \tr\left[e^{\lambda\left(\hat{L}(w,s) -\tr\left[\hat{L}(w,s)\sigma^{\cA_{Q}}(w,s)\right](\bbI_{\cH_{te}} \otimes \bbI_{\cH_{hyp}})\right)}\sigma^{\cA_{Q}}(w,s)\right] &\leq \frac{\lambda^2\mu^2}{2},~~\forall (w,s) \in \cW \times \cS,\label{quantum_mgf_emp}
    \end{align}
    \begin{align}
        \log\bbE_{S \sim {P}^{\cA_{Q}}_{S|W}(.|w)}\left[e^{\lambda\left(\tr\left[\hat{L}(w,S)\left(\rho_{te}(S) \otimes \sigma^{\cA_{Q}}_{hyp}(w)\right)\right] -  \bbE_{\overline{S}\sim{P}^{\cA_{Q}}_{S|W}(.|w)} \left[(\tr\left[\hat{L}(w,\overline{S})\left(\rho_{te}(\overline{S}) \otimes \sigma^{\cA_{Q}}_{hyp}(w)\right)\right]\right]\right)}\right] &\leq \frac{\lambda^2\tau^2}{2},~~ \forall w \in \cW.\label{classical_mgf2}
       \end{align}
\end{assumption}
    
To state a family of upper-bounds on the absolute value of the expected generalization error (mentioned in Definition \ref{gen_ws_error_exp}),
we require Lemma \ref{loss_var_form} below as its preparation.

    \begin{lemma}\label{loss_var_form}
Suppose for each $(w,s) \in \cW \times \cS,$ $\hat{L}(w,s)$ satisfies Assumptions \ref{sub_g_ass_Caro_modifed} and \ref{Assumption_theo_gen} for some $0 <\mu,\tau < \infty$. Then, the following holds,
    \begin{align}
     &\abs{\tr[\hat{L}(w,s)\sigma^{\cA_{Q}}(w,s)] - \tr\left[\hat{L}(w,s)\left(\rho_{te}(s) \otimes \sigma^{\cA_{Q}}_{hyp}(w,s)\right)\right]} \nn\\
     &\hspace{140pt}\leq
    \begin{cases}
          \sqrt{\frac{2\mu^2 \overline{D}_{\alpha}\left(\sigma^{\cA_{Q}}(w,s)||\rho_{{te}}(s) \otimes \sigma^{\cA_{Q}}_{hyp}(w,s)\right)}{\alpha}}, &\text{if } \alpha \in (0,1),\\
         \sqrt{2\mu^2 \overline{D}_{\alpha}\left(\sigma^{\cA_{Q}}(w,s)||\rho_{{te}}(s) \otimes \sigma^{\cA_{Q}}_{hyp}(w,s)\right)}, &\text{if } \alpha \in (1,\infty).
    \end{cases}\label{loss_var_form_eq}\\
    &\abs{\tr\left[\hat{L}(w,s)\left(\rho_{te}(s) \otimes \sigma^{\cA_{Q}}_{hyp}(w,s)\right)\right] - \tr\left[\hat{L}(w,s)\left(\rho_{te}(s) \otimes \sigma^{\cA_{Q}}_{hyp}(w)\right)\right]} \nn\\
    &\hspace{140pt}\leq
    \begin{cases}
          \sqrt{\frac{2\mu^2 \overline{D}_{\alpha}\left(\sigma^{\cA_{Q}}_{hyp}(w,s)||\sigma^{\cA_{Q}}_{hyp}(w)\right)}{\alpha}}, &\text{if } \alpha \in (0,1),\\
         \sqrt{2\mu^2 \overline{D}_{\alpha}\left(\sigma^{\cA_{Q}}_{hyp}(w,s)||\sigma^{\cA_{Q}}_{hyp}(w)\right)}, &\text{if } \alpha \in (1,\infty).
    \end{cases}\label{loss_var_form_eq1}
    \end{align}
\end{lemma}
\begin{proof} 
We first prove \eqref{loss_var_form_eq} in two cases and later we show that the proof of \eqref{loss_var_form_eq1} follows similarly.

{\textbf{Case} $\bf 1$ \textnormal{:} $\alpha \in (0,1)$}\\
From Lemma \ref{mod-san_renyi_var_form}, $\forall (w,s) \in \cW \times \cS$ and $\lambda \in \bbR$, we have the following,
\begin{align}
    &\overline{D}_{\alpha}\left(\sigma^{\cA_{Q}}(w,s)||\rho_{te}(s) \otimes \sigma^{\cA_{Q}}_{hyp}(w,s)\right) \geq \frac{\alpha}{\alpha - 1}\log \tr[e^{(\alpha - 1)\lambda \hat{L}(w,s)}\sigma^{\cA_{Q}}(w,s)] - \log \tr\left[e^{\alpha \lambda \hat{L}(w,s)}\left(\rho_{te}(s) \otimes \sigma^{\cA_{Q}}_{hyp}(s)\right)\right].\label{petz_form_learning}
\end{align}

We now bound the first term in the RHS of \eqref{petz_form_learning} as follows,
\begin{align}
    &\hspace{10pt}\log \tr\left[e^{(\alpha - 1)\lambda \hat{L}(w,s)}\sigma^{\cA_{Q}}(w,s)\right]\nn\\
    &= \log \tr\left[e^{(\alpha - 1)\lambda\left( \hat{L}(w,s) - \tr[\hat{L}(w,s)\sigma^{\cA_{Q}}(w,s)] (\bbI_{\cH_{te}} \otimes \bbI_{\cH_{hyp}})\right)}\sigma^{\cA_{Q}}(w,s)\right] - ((1 - \alpha)\lambda) \tr[\hat{L}(w,s)\sigma^{\cA_{Q}}(w,s)] \nn\\
    &\overset{a}{\leq} -((1 - \alpha)\lambda) \tr[\hat{L}(w,s)\sigma^{\cA_{Q}}(w,s)] + \frac{(\alpha - 1)^2\lambda^2\mu^2}{2},\label{petz_form_learning_term1}
\end{align}
where $a$ follows from \eqref{quantum_mgf_emp}. We now bound the second term in the RHS of \eqref{petz_form_learning} as follows,
\begin{align}
    &\hspace{10pt}\log \tr[e^{\alpha \lambda \hat{L}(w,s)}\left(\rho_{te}(s) \otimes \sigma^{\cA_{Q}}_{hyp}(w,s)\right)]\nn\\
    &= \log \tr\left[e^{\alpha\lambda\left( \hat{L}(w,s) - \tr\left[\hat{L}(w,s)\left(\rho_{te}(s) \otimes \sigma^{\cA_{Q}}_{hyp}(w,s)\right)\right] (\bbI_{\cH_{te}} \otimes \bbI_{\cH_{hyp}})\right)}\left(\rho_{te}(s) \otimes \sigma^{\cA_{Q}}_{hyp}(w,s)\right)\right]\hspace{75pt}\nn\\
    &\hspace{10pt}+ (\alpha\lambda) \tr\left[\hat{L}(w,s)\left(\rho_{te}(s) \otimes \sigma^{\cA_{Q}}_{hyp}(w,s)\right)\right]\nn\\
    &\overset{a}{\leq} (\alpha\lambda) \tr\left[\hat{L}(w,s)\left(\rho_{te}(s) \otimes \sigma^{\cA_{Q}}_{hyp}(w,s)\right)\right] + \frac{\alpha^2\lambda^2\mu^2}{2}\label{petz_form_learning_term2},
\end{align}
where $a$ follows from \eqref{quantum_mgf_true1_caro_mod}. 

 From \cref{petz_form_learning_term1,petz_form_learning_term2}, for $\alpha \in (0,1)$, we can rewrite \eqref{petz_form_learning} as follows,
\begin{align}
   &\hspace{10pt} \overline{D}_{\alpha}\left(\sigma^{\cA_{Q}}(w,s)||\rho_{te}(s) \otimes \sigma^{\cA_{Q}}_{hyp}(w,s)\right)\nn\\
   &\geq \alpha\lambda\left(\tr[\hat{L}(w,s)\sigma^{\cA_{Q}}(w,s)] - \tr\left[\hat{L}(w,s)\left(\rho_{te}(s) \otimes \sigma^{\cA_{Q}}_{hyp}(w,s)\right)\right]\right) +\frac{\alpha(\alpha - 1)\lambda^2\mu^2}{2} - \frac{\alpha^2\lambda^2\mu^2}{2}.\nn
\end{align}

We can rewrite the above inequality as follows,
\begin{equation}
    \left(\frac{\alpha\mu^2}{2}\right) \lambda^2 - \alpha\left(\tr[\hat{L}(w,s)\sigma^{\cA_{Q}}(w,s)] - \tr\left[\hat{L}(w,s)\left(\rho_{te}(s) \otimes \sigma^{\cA_{Q}}_{hyp}(w,s)\right)\right]\right)\lambda + \overline{D}_{\alpha}\left(\sigma^{\cA_{Q}}(w,s)||\rho_{te}(s) \otimes \sigma^{\cA_{Q}}_{hyp}(w,s)\right) \geq 0.\label{petz_form_learning_term3}
\end{equation}

Since the above inequality is a non-negative quadratic equation in $\lambda$ with the coefficient $\left(\frac{\alpha\mu^2}{2}\right) \geq 0,$ therefore its discriminant must be non-positive. Thus, we have the following inequality,

\begin{align}
    \alpha^2 \left(\tr[\hat{L}(w,s)\sigma^{\cA_{Q}}(w,s)] - \tr\left[\hat{L}(w,s)\left(\rho_{te}(s) \otimes \sigma^{\cA_{Q}}_{hyp}(w,s)\right)\right]\right)^2 &\leq 4 \left(\frac{\alpha\mu^2}{2}\right)\overline{D}_{\alpha}\left(\sigma^{\cA_{Q}}(w,s)||\rho_{te}(s) \otimes \sigma^{\cA_{Q}}_{hyp}(w,s)\right)\nn\\
    \Rightarrow \abs{\tr[\hat{L}(w,s)\sigma^{\cA_{Q}}(w,s)] - \tr\left[\hat{L}(w,s)\left(\rho_{te}(s) \otimes \sigma^{\cA_{Q}}_{hyp}(w,s)\right)\right]} &\leq \sqrt{\frac{2\mu^2 \overline{D}_{\alpha}\left(\sigma^{\cA_{Q}}(w,s)||\rho_{{te}}(s) \otimes \sigma^{\cA_{Q}}_{hyp}(w,s)\right)}{\alpha}}.\label{petz_form_learning_term4}
\end{align}

{\textbf{Case} $\bf 2$ \textnormal{:} $\alpha \in (1,\infty)$}\\
For $\alpha \in (1,\infty)$ the term $\log \tr[e^{(\alpha - 1)\lambda \hat{L}(w,s)}\sigma^{\cA_{Q}}(w,s)]$ in LHS of the inequality mentioned in \eqref{petz_form_learning_term1} can be upper-bounded as follows,
    \begin{equation*}
        \log \tr\left[e^{(\alpha - 1)\lambda \hat{L}(w,s)}\sigma^{\cA_{Q}}(w,s)\right] \overset{a}{\geq} (\alpha - 1)\lambda \tr[\hat{L}(w,s)\sigma^{\cA_{Q}}(w,s)],
    \end{equation*}
    where $a$ follows from Fact \ref{trace_log_ineq}.
The rest of the proof is similar to the Case $1$. This proves \eqref{loss_var_form_eq}. 

We now proceed to prove \eqref{loss_var_form_eq1} in the two following cases.

{\textbf{Case} $\bf 1$ \textnormal{:} $\alpha \in (0,1)$}\\
Using \cref{quantum_mgf_true1_caro_mod,quantum_mgf_true2_caro_mod} and a calculation similar to \cref{petz_form_learning_term1,petz_form_learning_term2,petz_form_learning_term3} we have the following inequality,

\begin{align}
    &\alpha^2 \left(\tr\left[\hat{L}(w,s)\left(\rho_{te}(s) \otimes \sigma^{\cA_{Q}}_{hyp}(w,s)\right)\right] - \tr\left[\hat{L}(w,s)\left(\rho_{te}(s) \otimes \sigma^{\cA_{Q}}_{hyp}(w)\right)\right]\right)^2\nn\\
    &\hspace{200pt}\leq 4 \left(\frac{\alpha\mu^2}{2}\right)\overline{D}_{\alpha}\left(\rho_{te}(s) \otimes \sigma^{\cA_{Q}}_{hyp}(w,s)||\rho_{te}(s) \otimes \sigma^{\cA_{Q}}_{hyp}(w)\right)\nn\\
    \Rightarrow &\abs{\tr\left[\hat{L}(w,s)\left(\rho_{te}(s) \otimes \sigma^{\cA_{Q}}_{hyp}(w,s)\right)\right] - \tr\left[\hat{L}(w,s)\left(\rho_{te}(s) \otimes \sigma^{\cA_{Q}}_{hyp}(w)\right)\right]}\nn\\
    &\hspace{200pt}\leq \sqrt{\frac{2\mu^2 \overline{D}_{\alpha}\left(\rho_{te}(s) \otimes \sigma^{\cA_{Q}}_{hyp}(w,s)||\rho_{{te}}(s) \otimes \sigma^{\cA_{Q}}_{hyp}(w)\right)}{\alpha}}\nn\\
    &\hspace{200pt}\overset{a}{=} \sqrt{\frac{2\mu^2 \overline{D}_{\alpha}\left(\sigma^{\cA_{Q}}_{hyp}(w,s)|| \sigma^{\cA_{Q}}_{hyp}(w)\right)}{\alpha}}\nn,
\end{align}
where $a$ follows from \eqref{fact_quantum_renyi_addit_eq2} of Fact \ref{fact_quantum_renyi_addit}.

{\textbf{Case} $\bf 2$ \textnormal{:} $\alpha \in (1,\infty)$}\\
For $\alpha \in (1,\infty)$ the term $\log \tr[e^{(\alpha - 1)\lambda \hat{L}(w,s)}\left(\rho_{te}(s) \otimes \sigma^{\cA_{Q}}_{hyp}(w,s)\right)]$ can be upper-bounded as follows,
    \begin{equation*}
        \log \tr\left[e^{(\alpha - 1)\lambda \hat{L}(w,s)}\left(\rho_{te}(s) \otimes \sigma^{\cA_{Q}}_{hyp}(w,s)\right)\right] \overset{a}{\geq} (\alpha - 1)\lambda \tr\left[\hat{L}(w,s)\left(\rho_{te}(s) \otimes \sigma^{\cA_{Q}}_{hyp}(w,s)\right)\right],
    \end{equation*}
    where $a$ follows from Fact \ref{trace_log_ineq}.
The rest of the proof is similar to the Case $1$ mentioned above.This proves \eqref{loss_var_form_eq1}. This completes the proof of Lemma \ref{loss_var_form}.
\end{proof}

\begin{theorem}[\textbf{Expected generalization error bound via modified sandwiched Quantum R\'enyi Divergence}]\label{theo_exxp_gen_err_bound_renyi}

    Suppose $\forall (w,s) \in \cW \times \cS,$ $L(w,s)$ satisfies Assumptions \ref{sub_g_ass_Caro_modifed} and \ref{Assumption_theo_gen}. Then, we have the following two upper bounds for $\abs{\overline{\textnormal{gen}}}$, 
    \begin{align}
        \abs{\overline{\textnormal{gen}}} &\leq \inf_{\alpha \in (0,1)}
             \bbE_{\substack{(W,S) \sim {P}^{\cA_{Q}}_{WS} }}\left(\sqrt{\frac{2\mu^2 \overline{D}_{\alpha}\left(\sigma^{\cA_{Q}}(W,S)||\rho_{te}(S) \otimes \sigma^{\cA_{Q}}_{hyp}(W,S)\right)}{\alpha}}\right.\notag\\
            &\hspace{100pt}\left. + \sqrt{\frac{2\mu^2 \overline{D}_{\alpha}\left( \sigma^{\cA_{Q}}_{hyp}(W,S)||\sigma^{\cA_{Q}}_{hyp}(W)\right)}{\alpha}}\right) 
            + \inf_{\gamma \in (0,1)} \bbE_{ W \sim {P}^{\cA_{Q}}_{W}}\left[\sqrt{\frac{2\tau^2 D^{c}_{\gamma}({P}^{\cA_{Q}}_{S|W} ||  P_{S})}{\gamma}}\right], 
\label{exp_gen_var_bound1}
    \end{align}
and 
     \begin{align}
          \abs{\overline{\textnormal{gen}}} &\leq \inf_{\alpha \in (1,\infty)}
            \bbE_{\substack{(W,S) \sim {P}^{\cA_{Q}}_{WS} }}\left(\sqrt{{2\mu^2 \overline{D}_{\alpha}\left(\sigma^{\cA_{Q}}(W,S)||\rho_{te}(S) \otimes \sigma^{\cA_{Q}}_{hyp}(W,S)\right)}}\right.\notag\\
&            \hspace{100pt}\left.+ \sqrt{{2\mu^2 \overline{D}_{\alpha}\left( \sigma^{\cA_{Q}}_{hyp}(W,S)||\sigma^{\cA_{Q}}_{hyp}(W)\right)}}\right) +
\inf_{\gamma \in (1,\infty)} \bbE_{ W \sim {P}^{\cA_{Q}}_{W}}\left[\sqrt{2\tau^2 D^{c}_{\gamma}({P}^{\cA_{Q}}_{S|W} ||  P_{S})}\right].
              \label{exp_gen_var_bound}
    \end{align}

\end{theorem}

\begin{remark}
From \eqref{exp_gen_var_bound}, it may appear that the bound for $\alpha,\gamma \in (1,\infty)$ follows trivially from the bound obtained in \eqref{exp_gen_var_bound1} for $\alpha,\gamma \in (0,1)$. However, this is not the case since the proof for the case when $\alpha,\gamma \in (1,\infty)$ is different from the proof for the case when $\alpha,\gamma \in (0,1)$, because of Lemma \ref{loss_var_form}. 
\end{remark}

\begin{remark}
    Observe that as discussed in Remark \ref{remark_mod_san} the terms involving modified sandwiched quantum R\'enyi divergence $\overline{D}_{\alpha}(\cdot|\cdot)$ in the RHS of \cref{exp_gen_var_bound,exp_gen_var_bound1} can be replaced by measured R\'enyi divergence $D^{\bbM}_{\alpha}(\cdot|\cdot)$ because of the lower-bound of Fact \ref{MD-Petz}, which results in tighter upper-bounds.
    
\end{remark}

Theorem \ref{theo_exxp_gen_err_bound_renyi} above can be viewed as a quantum version of Proposition \ref{Modak_2021_result} (\cite[Theorem $1$]{Modak21}). Proposition \ref{Modak_2021_result}  is a generalization of the results obtained in Propositions \ref{XR_2017_result} and \ref{Bu_2019_result} in terms of R\'enyi divergence. Similarly, Theorem \ref{theo_exxp_gen_err_bound_renyi} can be viewed as a generalization of Proposition \ref{Caro23_result} in terms of modified sandwiched R\'enyi divergence.
Further, observe that, unlike Proposition \ref{Modak_2021_result}, 
we get extra quantum information-theoretic quantities
as the second terms of \eqref{exp_gen_var_bound1} and \eqref{exp_gen_var_bound}. 
This is because, unlike in the classical learning setting, the generalization error as defined in 
Definition \ref{gen_ws_error} is asymmetric. Observe that in the classical case, the true loss for a fixed $w$, is defined as the expectation over the test data of the same loss function with respect to which the empirical loss is defined. However, this is not the case in the quantum setting (see Definitions \ref{emp_loss_def}, \ref{Ctrue_loss_def} and \ref{true_loss_def}). This asymmetric nature of the generalization error in the quantum case, as defined in \cite{Caro23} and the new definition of the true loss proposed in Definition \ref{true_loss_def}, is to mitigate the perturbations caused by the measurements and post-processing during the learning from quantum data.
That is, when we employ Definition \ref{Cexp_true_loss_def}, the second terms do not appear in the upper bounds.
Since Definition \ref{exp_true_loss_def} employs 
$\sigma^{\cA_{Q}}_{hyp}(W)$ instead of $\sigma^{\cA_{Q}}_{hyp}(W,S)$,
our upper bound has the second term, which comes from the difference between the definitions.

\begin{proof}
We calculate the absolute value of the expected generalized error as follows,
\begin{align*}
    \abs{\overline{\text{gen}}} &= \abs{\bbE_{W \sim {P}^{\cA_{Q}}_{W}}\left[\bbE_{\overline{S} \sim P_{S}}\left[\tr\left[\hat{L}(W,\overline{S})\left(\rho_{te}(\overline{S}) \otimes \sigma^{\cA_{Q}}_{hyp}(W)\right)\right]\right] - \bbE_{S \sim {P}^{\cA_{Q}}_{S|W}}\left[\tr[\hat{L}(W,S)\sigma^{\cA_{Q}}(W,S)]\right]\right]}\\
    &\leq  \bbE_{\substack{W \sim {P}^{\cA_{Q}}_{W} \\ S \sim {P}^{\cA_{Q}}_{S|W}}}\left[\abs{\tr[\hat{L}(W,S)\sigma^{\cA_{Q}}(W,S)] - \tr\left[\hat{L}(W,S)\left(\rho_{te}(S) \otimes \sigma^{\cA_{Q}}_{hyp}(W)\right)\right]}\right] \\
    &\hspace{10pt}+ \bbE_{W \sim {P}^{\cA_{Q}}_{W}}\left[\left|\bbE_{\overline{S}\sim P_{S}}\left[\tr\left[\hat{L}(W,\overline{S})\left(\rho_{te}(\overline{S}) \otimes \sigma^{\cA_{Q}}_{hyp}(W)\right)\right]\right] - \bbE_{S \sim \hat{P}^{\cA_{Q}}_{S|W}}\left[\tr\left[\hat{L}(W,S)\left(\rho_{te}(S) \otimes \sigma^{\cA_{Q}}_{hyp}(W)\right)\right]\right]\right|\right]\hspace{70pt}\\
    &\leq \bbE_{\substack{W \sim {P}^{\cA_{Q}}_{W} \\ S \sim {P}^{\cA_{Q}}_{S|W}}}\left[\abs{\tr[\hat{L}(W,S)\sigma^{\cA_{Q}}(W,S)] - \tr\left[\hat{L}(W,S)\left(\rho_{te}(S) \otimes \sigma^{\cA_{Q}}_{hyp}(W,S)\right)\right]}\right] \\
    &\hspace{10pt}+\bbE_{\substack{W \sim {P}^{\cA_{Q}}_{W} \\ S \sim {P}^{\cA_{Q}}_{S|W}}}\left[\abs{\tr\left[\hat{L}(W,S)\left(\rho_{te}(S) \otimes \sigma^{\cA_{Q}}_{hyp}(W,S)\right)\right] - \tr\left[\hat{L}(W,S)\left(\rho_{te}(S) \otimes \sigma^{\cA_{Q}}_{hyp}(W)\right)\right]}\right] \\
    &\hspace{10pt}+ \bbE_{W \sim {P}^{\cA_{Q}}_{W}}\left[\left|\bbE_{\overline{S}\sim P_{S}}\left[\tr\left[\hat{L}(W,\overline{S})\left(\rho_{te}(\overline{S}) \otimes \sigma^{\cA_{Q}}_{hyp}(W)\right)\right]\right] - \bbE_{S \sim \hat{P}^{\cA_{Q}}_{S|W}}\left[\tr\left[\hat{L}(W,S)\left(\rho_{te}(S) \otimes \sigma^{\cA_{Q}}_{hyp}(W)\right)\right]\right]\right|\right]\hspace{70pt}\\
    &\overset{a}{\leq}  \bbE_{\substack{W \sim {P}^{\cA_{Q}}_{W} \\ S \sim {P}^{\cA_{Q}}_{S|W}}}\left[\sqrt{\frac{2\mu^2 \overline{D}_{\alpha}\left(\sigma^{\cA_{Q}}(W,S)||\rho_{te}(S) \otimes \sigma^{\cA_{Q}}_{hyp}(W,S)\right)}{\alpha}}\right] + \bbE_{\substack{W \sim {P}^{\cA_{Q}}_{W} \\ S \sim {P}^{\cA_{Q}}_{S|W}}}\left[\sqrt{\frac{2\mu^2 \overline{D}_{\alpha}\left( \sigma^{\cA_{Q}}_{hyp}(W,S)||\sigma^{\cA_{Q}}_{hyp}(W)\right)}{\alpha}}\right] \\
    &\hspace{10pt}+ \bbE_{W \sim {P}^{\cA_{Q}}_{W}}\left[\left|\bbE_{\overline{S} \sim P_{S}}\left[\tr\left[\hat{L}(W,\overline{S})\left(\rho_{te}(\overline{S}) \otimes \sigma^{\cA_{Q}}_{hyp}(W)\right)\right]\right]   -\bbE_{S \sim {P}^{\cA_{Q}}_{S|W}}\left[\tr\left[\hat{L}(W,S)\left(\rho_{te}(S) \otimes \sigma^{\cA_{Q}}_{hyp}(W)\right)\right]\right]\right|\right]\\
    &\overset{b}{\leq} \bbE_{\substack{W \sim {P}^{\cA_{Q}}_{W} \\ S \sim {P}^{\cA_{Q}}_{S|W}}}\left[\sqrt{\frac{2\mu^2 \overline{D}_{\alpha}\left(\sigma^{\cA_{Q}}(W,S)||\rho_{te}(S) \otimes \sigma^{\cA_{Q}}_{hyp}(W,S)\right)}{\alpha}}\right]+ \bbE_{\substack{W \sim {P}^{\cA_{Q}}_{W} \\ S \sim {P}^{\cA_{Q}}_{S|W}}}\left[\sqrt{\frac{2\mu^2 \overline{D}_{\alpha}\left( \sigma^{\cA_{Q}}_{hyp}(W,S)||\sigma^{\cA_{Q}}_{hyp}(W)\right)}{\alpha}}\right]\\
    &\hspace{10pt} + \bbE_{ W \sim {P}^{\cA_{Q}}_{W}}\left[\sqrt{\frac{2\tau^2 D^{c}_{\gamma}({P}^{\cA_{Q}}_{S|W} ||  P_{S})}{\gamma}}\right]\\
    &=\bbE_{\substack{(W,S) \sim {P}^{\cA_{Q}}_{WS} }}\left[\sqrt{\frac{2\mu^2 \overline{D}_{\alpha}\left(\sigma^{\cA_{Q}}(W,S)||\rho_{te}(S) \otimes \sigma^{\cA_{Q}}_{hyp}(W,S)\right)}{\alpha}}+ \sqrt{\frac{2\mu^2 \overline{D}_{\alpha}\left( \sigma^{\cA_{Q}}_{hyp}(W,S)||\sigma^{\cA_{Q}}_{hyp}(W)\right)}{\alpha}}\right] \nn\\
    &\hspace{10pt}+ \bbE_{ W \sim {P}^{\cA_{Q}}_{W}}\left[\sqrt{\frac{2\tau^2 D^{c}_{\gamma}({P}^{\cA_{Q}}_{S|W} ||  P_{S})}{\gamma}}\right],
\end{align*}
    where $a$ follows from \cref{loss_var_form_eq,loss_var_form_eq1} and $b$ follows from  \cite[Lemma 2]{Modak21} under $\gamma \in (0,1)$ and the classical sub-Gaussianity assumptions mentioned in \eqref{classical_mgf_caro_mod},\eqref{classical_mgf2}. An important observation to make here is that $D (\sigma^{\cA_{Q}}_{hyp} (S,W) \|  $ $\sigma^{\cA_{Q}}_{hyp}(W))$ is well defined, since $\sigma^{\cA_{Q}}_{hyp}(S,W)<< \sigma^{\cA_{Q}}_{hyp}(W).$ 
    Taking the infimum with $\alpha \in (0,1)$ and $\gamma \in (0,1)$,
    we obtain the upper bound \eqref{exp_gen_var_bound1}.

For the case when $\alpha, \gamma \in (1,\infty)$, using Lemma \ref{loss_var_form} under the choice of $\alpha \in (1,\infty)$, we directly have the following inequality,
\begin{align*}
    \abs{\overline{\text{gen}}} &\leq \bbE_{\substack{(W,S) \sim {P}^{\cA_{Q}}_{WS} }}\left[\sqrt{2\mu^2 \overline{D}_{\alpha}\left(\sigma^{\cA_{Q}}(W,S)||\rho_{te}(S) \otimes \sigma^{\cA_{Q}}_{hyp}(W,S)\right)} + \sqrt{2\mu^2 \overline{D}_{\alpha}\left( \sigma^{\cA_{Q}}_{hyp}(W,S)||\sigma^{\cA_{Q}}_{hyp}(W)\right)}\right] \nn\\
    &\hspace{10pt}+ \bbE_{W \sim {P}^{\cA_{Q}}_{W}}\left[\left|\bbE_{\overline{S} \sim P_{S}}\left[\tr\left[\hat{L}(W,\overline{S})\left(\rho_{te}(\overline{S}) \otimes \sigma^{\cA_{Q}}_{hyp}(W)\right)\right]\right]   -\bbE_{S \sim {P}^{\cA_{Q}}_{S|W}}\left[\tr\left[\hat{L}(W,S)\left(\rho_{te}(S) \otimes \sigma^{\cA_{Q}}_{hyp}(W)\right)\right]\right]\right|\right]\\
    &\overset{a}{\leq} \bbE_{\substack{(W,S) \sim {P}^{\cA_{Q}}_{WS} }}\left[\sqrt{2\mu^2 \overline{D}_{\alpha}\left(\sigma^{\cA_{Q}}(W,S)||\rho_{te}(S) \otimes \sigma^{\cA_{Q}}_{hyp}(W,S)\right)} + \sqrt{2\mu^2 \overline{D}_{\alpha}\left( \sigma^{\cA_{Q}}_{hyp}(W,S)||\sigma^{\cA_{Q}}_{hyp}(W)\right)}\right] \nn\\
    &\hspace{10pt}+ \bbE_{ W \sim {P}^{\cA_{Q}}_{W}}\left[\sqrt{2\tau^2 D^{c}_{\gamma}({P}^{\cA_{Q}}_{S|W} ||  P_{S})}\right],
\end{align*}
where $a$ follows from \cite[Remark 2]{Modak21}. 
    Taking the infimum with $\alpha \in (1,\infty)$ and $\gamma \in (1,\infty)$,
    we obtain the upper bound \eqref{exp_gen_var_bound}.
This completes the proof of
Theorem \ref{theo_exxp_gen_err_bound_renyi}.
\end{proof}

\subsection*{Discussion on comparison of Theorem \ref{theo_exxp_gen_err_bound_renyi} with Proposition \ref{Caro23_result} (obtained in \cite{Caro23})}

    It is important to note that if there is no correlation between quantum testing and training data i.e. $\rho(S) := \rho_{te}(S) \otimes \rho_{tr}(S)$ the RHS of \eqref{Caro23_result_eq} in Proposition \ref{Caro23_result} only contains a classical quantity i.e. $ O(\sqrt{I[S;W]})$. However, under this assumption we will still have a quantum quantity in terms of modified sandwiched quantum R\'enyi divergence i.e. $O\left(\sqrt{{\overline{D}_{\alpha}\left( \sigma^{\cA_{Q}}_{hyp}(W,S)||\sigma^{\cA_{Q}}_{hyp}(W)\right)}}\right)$ along with a classical term in terms of R\'enyi divergence.

Observe that Theorem \ref{Caro23_result_mod} turns out to be a direct corollary of Theorem \ref{theo_exxp_gen_err_bound_renyi}, when $\alpha,\gamma \to 1$. This directly follows from Facts \ref{fact_limit_renyi_divergence} and \ref{fact_limit_sandwiched_divergence}. For clarity, we restate Theorem \ref{Caro23_result_mod} in the form of a corollary below.
  
\begin{corollary}\label{corr_exxp_gen_err_bound_renyi}
    Suppose $\forall (w,s) \in \cW \times \cS,$ $L(w,s)$ satisfies Assumptions \ref{sub_g_ass_Caro_modifed} and \ref{Assumption_theo_gen} for some $0 <\mu,\tau < \infty$. Then, we have the following upper-bound on $\abs{\overline{\text{gen}}}$,
    \begin{align*}
        \abs{\overline{\textnormal{gen}}} &\leq \bbE_{\substack{(W,S) \sim {P}^{\cA_{Q}}_{WS} }}\left[\sqrt{2\mu^2 D\left(\sigma^{\cA_{Q}}(W,S)||\rho_{te}(S) \otimes \sigma^{\cA_{Q}}_{hyp}(W,S)\right)} + \sqrt{2\mu^2 D\left( \sigma^{\cA_{Q}}_{hyp}(W,S)||\sigma^{\cA_{Q}}_{hyp}(W)\right)}\right] + \sqrt{2\tau^2 I[S;W]}.
    \end{align*}
\end{corollary}  

Further, as a consequence of Fact \ref{MD-Petz}, we get the following weakened upper-bound obtained in Theorem \ref{theo_exxp_gen_err_bound_renyi}, mentioned as a corollary below.

\begin{corollary}\label{corr_exxp_gen_err_bound_renyi_weak}
    Suppose $\forall (w,s) \in \cW \times \cS,$ $L(w,s)$ satisfies Assumptions \ref{sub_g_ass_Caro_modifed} and \ref{Assumption_theo_gen}. Then, we have the following two upper bounds for $\abs{\overline{\textnormal{gen}}}$,
    \begin{align}
        \abs{\overline{\textnormal{gen}}} &\leq \inf_{\alpha \in (0,1)}
             \bbE_{\substack{(W,S) \sim {P}^{\cA_{Q}}_{WS} }}\left(\sqrt{\frac{2\mu^2 D_{\alpha}\left(\sigma^{\cA_{Q}}(W,S)||\rho_{te}(S) \otimes \sigma^{\cA_{Q}}_{hyp}(W,S)\right)}{\alpha}}\right.\notag\\
            &\hspace{100pt}\left. + \sqrt{\frac{2\mu^2 D_{\alpha}\left( \sigma^{\cA_{Q}}_{hyp}(W,S)||\sigma^{\cA_{Q}}_{hyp}(W)\right)}{\alpha}}\right) 
            + \inf_{\gamma \in (0,1)}\bbE_{ W \sim {P}^{\cA_{Q}}_{W}}\left[\sqrt{\frac{2\tau^2 D^{c}_{\gamma}({P}^{\cA_{Q}}_{S|W} ||  P_{S})}{\gamma}}\right], 
\label{exp_gen_var_bound_weak1}
    \end{align}
and 
     \begin{align}
          \abs{\overline{\textnormal{gen}}} &\leq \inf_{\alpha \in (1,\infty)}
            \bbE_{\substack{(W,S) \sim {P}^{\cA_{Q}}_{WS} }}\left(\sqrt{{2\mu^2 D_{\alpha}\left(\sigma^{\cA_{Q}}(W,S)||\rho_{te}(S) \otimes \sigma^{\cA_{Q}}_{hyp}(W,S)\right)}}\right.\notag\\
&            \hspace{100pt}\left.+ \sqrt{{2\mu^2 D_{\alpha}\left( \sigma^{\cA_{Q}}_{hyp}(W,S)||\sigma^{\cA_{Q}}_{hyp}(W)\right)}}\right) 
+ \inf_{\gamma \in (1,\infty)} \bbE_{ W \sim {P}^{\cA_{Q}}_{W}}\left[\sqrt{2\tau^2 D^{c}_{\gamma}({P}^{\cA_{Q}}_{S|W} ||  P_{S})}\right].
              \label{exp_gen_var_bound_weak2}
    \end{align}
\end{corollary}

\begin{remark}\label{recoverability_to_caro}
    In Definition \ref{gen_ws_error}, for any $(w,s) \in \cW \times \cS$, if we assume $\textnormal{gen}(w,s) = \textnormal{gen}^{(\textnormal{old})}(w,s)$, (where $\textnormal{gen}^{(\textnormal{old})}(w,s)$ is  defined in \eqref{Cgen_ws_error_eq}) and $\hat{L}(w,s)$ satisfies the sub-Gaussianity assumptions mentioned in \cref{quantum_mgf_emp,quantum_mgf_true1_caro_mod,classical_mgf_caro_mod,classical_mgf2} for some $0 <\mu,\tau < \infty$,
       then from Theorem \ref{theo_exxp_gen_err_bound_renyi}, we get the following upper-bounds for $\abs{\overline{\textnormal{gen}}}$,
       \begin{align}
        \abs{\overline{\textnormal{gen}}} &\leq
        \begin{cases}
             \bbE_{\substack{(W,S) \sim {P}^{\cA_{Q}}_{WS} }}\left(\sqrt{\frac{2\mu^2 \overline{D}_{\alpha}\left(\sigma^{\cA_{Q}}(W,S)||\rho_{te}(S) \otimes \sigma^{\cA_{Q}}_{hyp}(W,S)\right)}{\alpha}}\right)\\
            \hspace{130pt}+\bbE_{ W \sim {P}^{\cA_{Q}}_{W}}\left[\sqrt{\frac{2\tau^2 D^{c}_{\gamma}({P}^{\cA_{Q}}_{S|W} ||  P_{S})}{\gamma}}\right], &\text{ if } \alpha,\gamma \in (0,1),\\
             \\
             \\
            \bbE_{\substack{(W,S) \sim {P}^{\cA_{Q}}_{WS} }}\left(\sqrt{2\mu^2 \overline{D}_{\alpha}\left(\sigma^{\cA_{Q}}(W,S)||\rho_{te}(S) \otimes \sigma^{\cA_{Q}}_{hyp}(W,S)\right)}\right)\\
             \hspace{130pt}+\bbE_{ W \sim {P}^{\cA_{Q}}_{W}}\left[\sqrt{2\tau^2 D^{c}_{\gamma}({P}^{\cA_{Q}}_{S|W} ||  P_{S})}\right], &\text{ if  } \alpha,\gamma \in (1,\infty).\\
        \end{cases}  \label{Caro_formulation}
    \end{align}
    For $\alpha,\gamma \to 1$, we obtain the following upper-bound on $\abs{\overline{\textnormal{gen}}},$
    \begin{align}
        \abs{\overline{\textnormal{gen}}} &\leq \bbE_{\substack{(W,S) \sim {P}^{\cA_{Q}}_{WS} }}\left[\sqrt{2\mu^2 D\left(\sigma^{\cA_{Q}}(W,S)||\rho_{te}(S) \otimes \sigma^{\cA_{Q}}_{hyp}(W,S)\right)}\right] + \sqrt{2\tau^2 I[S;W]}.\label{Caro_gen_bound}
    \end{align}
    The upper-bound obtained in \eqref{Caro_gen_bound} recovers the same upper-bound as mentioned in \textnormal{Proposition \ref{Caro23_result}}.

 \end{remark}

\subsection{Bounds on the expected quantum generalization error under i.i.d assumption of quantum data}\label{subsec:iid_data}

The upper-bounds obtained in Theorem \ref{theo_exxp_gen_err_bound_renyi} and Corollary \ref{corr_exxp_gen_err_bound_renyi} do not depend on the size of the training data. This happens because in Theorem \ref{theo_exxp_gen_err_bound_renyi}, we have performed all the calculations with respect to the whole data and not each of the classical and quantum data-points residing in the classical and quantum data, since we did not assume any i.i.d. structure of the classical and the quantum data. However, suppose we take an i.i.d. structure of the classical data as well as an i.i.d. structure of the quantum data as mentioned in \cite[eqs. ($4.60$) and ($4.61$)]{Caro23}. In that case, we can now fully utilize the i.i.d. assumption of classical data along with the quantum data. Formally, we assume that the test and train data Hilbert spaces are factorized as $\cH^{te} := \left(\cH^{\cZ_{te}}\right)^{\otimes n}$ and $\cH^{tr} :=\left(\cH^{\cZ_{tr}}\right)^{\otimes n}$ and $\forall s := (z_1,\cdots, z_n) \in \cZ^n, \rho(s) := \bigotimes_{i = 1}^{n}\rho({z_i})$ ( we assume for each $z \in \cZ$ $\rho(z) \in \cH^{\cZ_{te}} \otimes \cH^{\cZ_{tr}}$ might be correlated or even entangled across $\cH^{\cZ_{te}}$ and $\cH^{\cZ_{tr}}$).

The overall action of $\cA_Q$ over the data state $\rho$ leads us to the following CQ state,
 
 \begin{align*}
     \sigma^{\cA_{Q}} 
    &= \bbE_{W \sim {P}^{\cA_{Q}}_{W}}\left[\ketbra{W} \otimes\bbE_{S \sim {P}^{\cA_{Q}}_{S|W}(\cdot|W)}\left[\ketbra{S}\otimes\sigma^{\cA_{Q}}(W,S)\right]\right]\\
    &= \bbE_{W \sim {P}^{\cA_{Q}}_{W}}\left[\ketbra{W} \otimes \bigotimes_{i=1}^{n}\bbE_{Z_i \sim {P}^{\cA_{Q}}_{Z_i|W}(\cdot|W)}\left[\ketbra{Z_i}\otimes\sigma^{\cA_{Q}}(W,Z_i)\right]\right],\\
 \end{align*}
where for any $i \in [n]$, ${P}^{\cA_{Q}}_{Z_i|W}$ is the corresponding marginal of ${P}^{\cA_{Q}}_{S|W} $ and $\sigma^{\cA_{Q}}(W,Z_i) := \left(\bbI_{\cH^{\cZ_{te}}} \otimes \Lambda_{W,Z_i}\right) \left(\rho^{\cA_{Q}}(W,Z_i)\right)$, where $\{\Lambda_{w,z} : T(\cH^{\cZ_{tr}}) \rightarrow T(\cH^{\widehat{hyp}})\}_{(w,z) \in \cW \times \cZ}$ is a collection of quantum channels.

We consider a family of non-negative self-adjoint loss observables $\left\{\hat{L}(w,s) \in \cL(\cH^{te} \otimes \cH^{hyp})\right\}_{\substack{(w,s) \in \cW \times \cZ^n}}$ (where $\cH^{hyp} := {\cH^{\widehat{hyp}}}^{\otimes n}$), where for each $(w,s) \in \cW \times \cZ^n$, we $\hat{L}(w,s)$ is of the following local form,
\begin{equation}
    \hat{L}(w,s) := \frac{1}{n}\sum_{i=1}^{n}(\bbI_{\cH_{\cZ_{te}}} \otimes \bbI_{\cH^{\widehat{hyp}}})^{\otimes (i-1)} \otimes L(w,z_i) \otimes (\bbI_{\cH_{\cZ_{te}}} \otimes \bbI_{\cH^{\widehat{hyp}}})^{\otimes( n-i)},\label{loss_observables_frag}
\end{equation}
where for each $i \in [n],$ $L(w,z_i)$ is a local loss observable acting on the $i$-th iteration of the test and depends on $(w,z_i)$.

With respect to the loss observable mentioned in \eqref{loss_observables_frag}, we prove corollaries of Lemma \ref{loss_var_form} and Theorem \ref{theo_exxp_gen_err_bound_renyi} in terms of the modified sandwiched R\'enyi divergence. To prove these corollaries, we further require the following sub-Gaussianity assumption, which are special cases of Assumptions \ref{sub_g_ass_Caro_modifed} and \ref{Assumption_theo_gen}.

\begin{assumption}\label{Assumption_theo_gen_frag}
The collection of loss observables $\left\{L(w,z)\right\}_{(w,z) \in \cW \times \cZ}$ for some $0<\mu, \tau < \infty$ and any $\lambda \in \mathbb{R}$ satisfies the following,
    \begin{align}
    \log \tr\left[e^{\lambda\left(L(w,z) -\tr\left[L(w,z)\sigma^{\cA_{Q}}(w,z)\right](\bbI_{\cH_{\cZ_{te}}} \otimes \bbI_{\cH_{\widehat{hyp}}})\right)}\sigma^{\cA_{Q}}(w,z)\right] &\leq \frac{\lambda^2\mu^2}{2},~~\forall (w,z) \in \cW\times\cZ,\label{quantum_mgf_emp_frag}\\
        \log \tr\left[e^{\lambda\left(L(w,z) -\tr\left[L(w,z)\left(\rho_{\cZ_{te}}(z) \otimes \sigma^{\cA_{Q}}_{\widehat{hyp}}(w,z)\right)\right](\bbI_{\cH_{\cZ_{te}}} \otimes \bbI_{\cH_{\widehat{hyp}}})\right)}\left(\rho_{\cZ_{te}}(z) \otimes \sigma^{\cA_{Q}}_{\widehat{hyp}}(w,z)\right)\right] &\leq \frac{\lambda^2\mu^2}{2},~~\forall (w,z) \in \cW\times\cZ,\label{quantum_mgf_true_frag}\\
        \log \tr\left[e^{\lambda\left(L(w,z) -\tr\left[L(w,z)\left(\rho_{\cZ_{te}}(z) \otimes \sigma^{\cA_{Q}}_{\widehat{hyp}}(w)\right)\right](\bbI_{\cH_{\cZ_{te}}} \otimes \bbI_{\cH_{\widehat{hyp}}})\right)}\left(\rho_{\cZ_{te}}(z) \otimes \sigma^{\cA_{Q}}_{\widehat{hyp}}(w)\right)\right] &\leq \frac{\lambda^2\mu^2}{2},~~\forall (w,z) \in \cW\times\cZ,\label{quantum_mgf_true1_frag}
    \end{align}
    \begin{align}
        &\log\bbE_{Z_i \sim P} \left[e^{\lambda\left(\tr\left[L(w,Z_i)\left(\rho_{\cZ_{te}}(Z_i) \otimes \sigma^{\cA_{Q}}_{\widehat{hyp}}(w)\right)\right] -  \bbE_{\overline{Z}_i \sim P} \left[(\tr\left[L(w,\overline{Z}_i)\left(\rho_{\cZ_{te}}(\overline{Z}_i) \otimes \sigma^{\cA_{Q}}_{\widehat{hyp}}(w)\right)\right]\right]\right)}\right] \leq \frac{\lambda^2\tau^2}{2}\label{classical_mgf_frag},~~\forall w \in \cW,\\
        &\log\bbE_{Z_i \sim {P}^{\cA_{Q}}_{Z_i|W}(.|w)} \left[e^{\lambda\left(\tr\left[L(w,Z_i)\left(\rho_{\cZ_{te}}(Z_i) \otimes \sigma^{\cA_{Q}}_{\widehat{hyp}}(w)\right)\right] -  \bbE_{\overline{Z}_i \sim{P}^{\cA_{Q}}_{Z|W}(.|w)} \left[(\tr\left[L(w,\overline{Z}_i)\left(\rho_{\cZ_{te}}(\overline{Z}_i) \otimes \sigma^{\cA_{Q}}_{\widehat{hyp}}(w)\right)\right]\right]\right)}\right] \leq \frac{\lambda^2\tau^2}{2}\label{classical_mgf2_frag},~~\forall w \in \cW,
    \end{align}
    \end{assumption}

\begin{corollary}\label{loss_var_form_mod_san_frag}

Suppose $\forall (w,z) \in \cW \times \cZ,$ $L(w,z)$ satisfies Assumption \ref{Assumption_theo_gen_frag} for some $0 <\mu,\tau < \infty$. Then, we have,
    {\allowdisplaybreaks\begin{align}
    &\abs{\tr[L(w,z)\sigma^{\cA_{Q}}(w,z)] - \tr\left[L(w,z)\left(\rho_{\cZ_{te}}(z) \otimes \sigma^{\cA_{Q}}_{\widehat{hyp}}(w,z)\right)\right]}\nn\\
     &\hspace{140pt} \leq
    \begin{cases}
        \sqrt{\frac{2\mu^2 \overline{D}_{\alpha}\left(\sigma^{\cA_{Q}}(w,z)||\rho_{\cZ_{te}}(z) \otimes \sigma^{\cA_{Q}}_{\widehat{hyp}}(w,z)\right)}{\alpha}}, &\text{ if } \alpha \in (0,1),\\
        \sqrt{2\mu^2 \overline{D}_{\alpha}\left(\sigma^{\cA_{Q}}(w,z)||\rho_{\cZ_{te}}(z) \otimes \sigma^{\cA_{Q}}_{\widehat{hyp}}(w,z)\right)}, &\text{ if } \alpha \in (1,\infty),
    \end{cases}
             \label{loss_var_form_mod_san_frag_eq1}\\
    &\abs{\tr[L(w,z)\left(\rho_{\cZ_{te}}(z) \otimes \sigma^{\cA_{Q}}_{\widehat{hyp}}(w,z)\right)] - \tr\left[L(w,z)\left(\rho_{\cZ_{te}}(z) \otimes \sigma^{\cA_{Q}}_{\widehat{hyp}}(w)\right)\right]}\nn\\
     &\hspace{140pt} \leq
    \begin{cases}
        \sqrt{\frac{2\mu^2 \overline{D}_{\alpha}\left(\sigma^{\cA_{Q}}_{\widehat{hyp}}(w,z)|| \sigma^{\cA_{Q}}_{\widehat{hyp}}(w)\right)}{\alpha}}, &\text{ if } \alpha \in (0,1),\\
        \sqrt{2\mu^2 \overline{D}_{\alpha}\left(\sigma^{\cA_{Q}}_{\widehat{hyp}}(w,z)|| \sigma^{\cA_{Q}}_{\widehat{hyp}}(w)\right)}, &\text{ if } \alpha \in (1,\infty).\\
    \end{cases}
             \label{loss_var_form_mod_san_frag_eq2}
    \end{align}}
\end{corollary}

\begin{corollary}\label{cor_exxp_gen_err_bound_mod_renyi_frag}
    Suppose $\forall (w,z) \in \cW \times \cZ,$ $L(w,z)$ satisfies Assumption \ref{Assumption_theo_gen_frag} for some $0 <\mu,\tau < \infty$. Then, we have the following, 
    \begin{align*}
        \abs{\overline{\text{{\textnormal{gen}}}}} &\leq
        \begin{cases}
\inf_{\alpha \in (0,1)}\Bigg(             \sqrt{\frac{ 2\mu^2\bbE_{\substack{(W,S) \sim {P}^{\cA_{Q}}_{WS} }}\left[\overline{D}_{\alpha}\left(\sigma^{\cA_{Q}}(W,S)||\rho_{te}(S) \otimes \sigma^{\cA_{Q}}_{hyp}(W,S)\right)\right]}{n \alpha}}\nn\\
             \hspace{50pt}+ \sqrt{\frac{ 2\mu^2\bbE_{\substack{(W,S) \sim {P}^{\cA_{Q}}_{WS} }}\left[\overline{D}_{\alpha}\left(\sigma^{\cA_{Q}}_{hyp}(W,S)|| \sigma^{\cA_{Q}}_{hyp}(W)\right)\right]}{n \alpha}}\Bigg)
             +\inf_{\gamma \in (0,1)}\frac{1}{n}\sum_{i=1}^{n}\bbE_{W \sim P^{\cA}_{W}}\left[\sqrt{\frac{2\tau^2D_{\gamma}^c(P_{Z_i|W} \| P)}{\gamma} }\right], &
             \\
\inf_{\alpha \in (1,\infty)}\Bigg(            \sqrt{\frac{ 2\mu^2\bbE_{\substack{(W,S) \sim {P}^{\cA_{Q}}_{WS} }}\left[\overline{D}_{\alpha}\left(\sigma^{\cA_{Q}}(W,S)||\rho_{te}(S) \otimes \sigma^{\cA_{Q}}_{hyp}(W,S)\right)\right]}{n}} \nn\\
            \hspace{50pt}+ \sqrt{\frac{ 2\mu^2\bbE_{\substack{(W,S) \sim {P}^{\cA_{Q}}_{WS} }}\left[\overline{D}_{\alpha}\left(\sigma^{\cA_{Q}}_{hyp}(W,S)||\sigma^{\cA_{Q}}_{hyp}(W)\right)\right]}{n}} \Bigg)
            + \inf_{\gamma \in (1,\infty)}\frac{1}{n}\sum_{i=1}^{n}\bbE_{W \sim P^{\cA}_{W}}\left[\sqrt{{2\tau^2D_{\gamma}^c(P_{Z_i|W} \| P)} }\right]. &
        \end{cases}       
    \end{align*}
\end{corollary}

Observe that the above result is a quantum version of the individual sample-based upper-bounds mentioned in Proposition \ref{Bu_2019_result}.

\begin{figure}[h]
    \centering
    \begin{subfigure}[b]{\textwidth}
        \centering
\includegraphics[height = 0.39\linewidth]{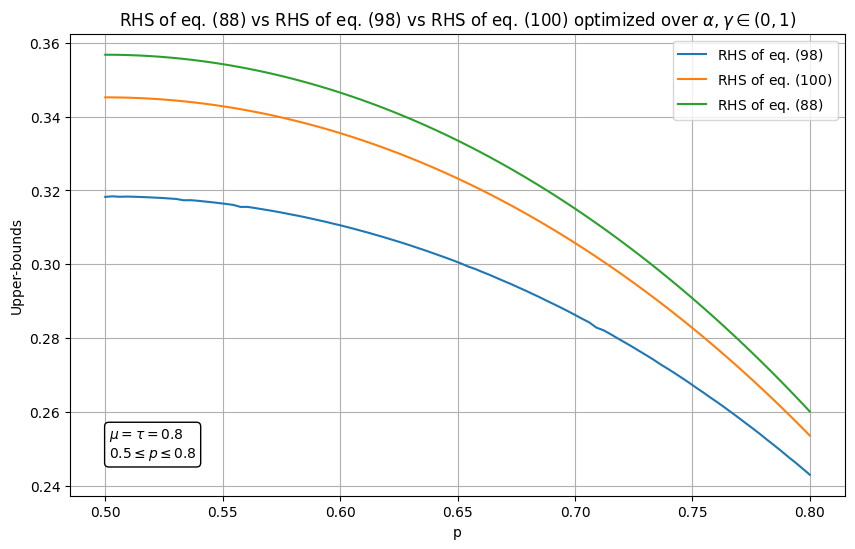}
\caption{}
    \label{fig:gen_ex_sub3}
    \end{subfigure}
   \newline 
     \begin{subfigure}[b]{\textwidth}
        \centering
        \includegraphics[height = 0.39\linewidth]{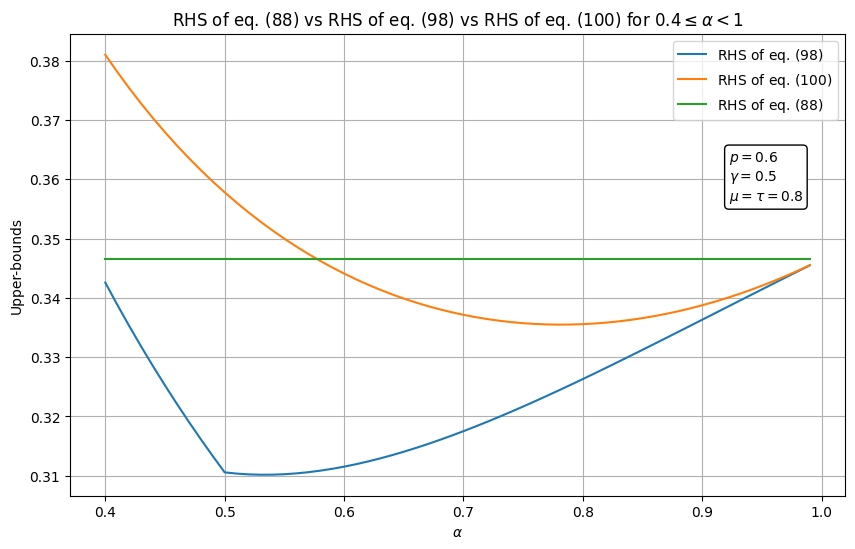}
        \caption{}
\label{fig:gen_ex_sub4}
    \end{subfigure}
    \caption{RHS of eq. \eqref{Caro23_result_mod_eq} vs  RHS of eq. \eqref{exp_gen_var_bound1} vs RHS of eq. \eqref{exp_gen_var_bound_weak1}.}
    \label{fig:gen_ex2}
\end{figure}

    \subsection{Comparision between the results obtained in Theorems \ref{Caro23_result_mod}, \ref{theo_exxp_gen_err_bound_renyi} and Corollary \ref{corr_exxp_gen_err_bound_renyi_weak}}

For the case when  $\alpha,\gamma \in (0,1)$, our bounds are in terms of $\frac{\overline{D}_{\alpha}(\cdot\|\cdot)}{\alpha}$, $\frac{D_{\alpha}(\cdot\|\cdot)}{\alpha}$ and $\frac{D_{\gamma}(\cdot\|\cdot)}{\gamma}$. Because of these extra multiplicative factors of $\frac{1}{\alpha}$ and $\frac{1}{\gamma}$, it is not clear whether the terms $\frac{\overline{D}_{\alpha}(\cdot\|\cdot)}{\alpha}$, $\frac{D_{\alpha}(\cdot\|\cdot)}{\alpha}$ and $\frac{D^{c}_{\gamma}(\cdot\|\cdot)}{\gamma}$ are smaller than $D(\cdot\|\cdot)$ and $D^{c}(\cdot\|\cdot)$ respectively. Using simulations for a toy example, in \cite{Modak21}, the authors showed cases when the bounds obtained on expected generalization error in terms of $D^{c}_{\gamma}(\cdot\|\cdot)$ is smaller than the bounds obtained on expected generalization error in terms 
 of $D(\cdot\|\cdot)$ for $\gamma = 0.5$. 
 
 In a similar spirit, to compare the bounds obtained in Theorems \ref{Caro23_result_mod}, \ref{theo_exxp_gen_err_bound_renyi} and Corollary \ref{corr_exxp_gen_err_bound_renyi_weak}, we consider an example, where for each $z \in \{0,1\}$ $\rho(z) := \rho_{te}(z) \otimes  \rho_{tr}(z)$. Therefore, any measurement and post-processing on the train part will not affect the test part. Hence, the first term of \eqref{Caro23_result_mod_eq}, \eqref{exp_gen_var_bound1} and \eqref{exp_gen_var_bound_weak1} is equal to zero (however, the second term in the above equations is not zero). For our example, the classical-quantum data state as mentioned in \eqref{data_state2} has the following form (parametrized by $p \in (0,1)$),
\begin{align}
    \rho = p \ketbra{0}  \otimes \ketbra{\psi_0}_{te} \otimes \ketbra{\psi_0}_{tr} + (1-p)\ketbra{1}\otimes \ketbra{\psi_1}_{te} \otimes \ketbra{\psi_1}_{tr}.
\end{align}
In the above for $z \in \{0,1\},$
\begin{align}
   \rho(z) &=  \ketbra{\psi_z}_{te} \otimes \ketbra{\psi_z}_{tr},
\end{align}
where, $\ket{\psi_0}$ and $\ket{\psi_1}$ are defined as follows,
\begin{align*}
    \ket{\psi_0} &:= \cos\theta\ket{\phi_{0}} + \sin\theta \ket{\phi^{\perp}_{0}},\\
    \ket{\psi_1} &:= \cos\beta \ket{\phi_{1}} + \sin\beta \ket{\phi^{\perp}_{1}},
\end{align*}
where, $\cos^2{\theta} = 0.45$, $\cos^2{\beta} = 0.5$, $\ket{\phi_{0}} = a\ket{0} + b\ket{1}$ and $\ket{\phi_{1}} = c\ket{0} + d\ket{1}$. Further, $a \approx -0.59 - 0.29i, b \approx -0.25+0.71i, c \approx 0.34-0.42i$ and $d \approx -0.83-0.12i$ (we list only approximate values of $a,b,c$ and $d$ for simplicity). We now consider the following measurements (dependent on the classical data $z$), which we will perform on $\rho$,
\begin{align*}
    \{E_{z}(w)\}_{(w,z) \in \{0,1\}^2} &= \left\{\ketbra{\phi_z},\ketbra{\phi^{\perp}_{z}}\right\},
\end{align*}

where, $\{E_{z}(w)\}_{(w,z) \in \{0,1\}^2}$ is in the similar spirit as $\left\{E^{\cA_Q}_{s}(w)\right\}_{(w,s) \in \cW \times \cS}$, discussed in Subsection \ref{subsec:quantum_frame}. After performing the measurement, $\forall(w,z) \in \{0,1\}^{2}$, we consider $\Lambda_{w,z} = \bbI$ (where $\Lambda_{w,z}$ is in the similar spirit as $\Lambda_{w,s}$, discussed in Subsection \ref{subsec:quantum_frame}). 

For each $(w, s) \in \cW \times \cS$, we only assume that 
the loss observable $\hat{L}(w, s)$, satisfies satisfies Assumptions \ref{sub_g_ass_Caro_modifed} and \ref{Assumption_theo_gen} for $\mu=\tau =0.8$, because
the expressions in the RHS of \eqref{Caro23_result_mod_eq}, \eqref{exp_gen_var_bound1} and \eqref{exp_gen_var_bound_weak1} depend only on $\mu$ and $\tau$ for
the loss observable  $\hat{L}(w, s)$. Therefore, we do not explicitly mention its choice here.

Under these settings, in Figure \ref{fig:gen_ex_sub3} above, we compare the optimal values of RHS of  \eqref{Caro23_result_mod_eq}, \eqref{exp_gen_var_bound1} and \eqref{exp_gen_var_bound_weak1} respectively by varying $p$ ($0.5 \leq p \leq 0.8$).

Further, under the same settings mentioned above, in Figure \ref{fig:gen_ex_sub4} above, we compare the RHS of  \eqref{Caro23_result_mod_eq}, \eqref{exp_gen_var_bound1} and \eqref{exp_gen_var_bound_weak1} by varying $\alpha \in [0.4,1)$ for $p = 0.6$.


For various cases of the example considered above, our simulations show that the modified sandwiched R\'enyi divergence always gives a better bound compared to the bounds obtained in terms of the Petz R\'enyi divergence and the quantum relative entropy, respectively.

\section{Quantum R\'enyi Divergences based Bounds on Generalization error in probability}\label{sec:gen_bound_exp2}
\subsection{Bounds on the generalization error in probability under i.i.d. assumption of quantum data}\label{sec:gen_bound_prob}
Expectation bounds derived in the earlier sections only provide average-case guarantees. Therefore, a more relevant metric to study the performance of a learning algorithm would be to obtain bounds on the generalization error in probability. In this section, we will study bounds of the following form,

\begin{equation*}
    \Pr_{(W,S)\sim P_{WS}}\left\{\abs{\text{gen}(W,S)} \leq \eps \right\} \geq 1 - \delta,
\end{equation*}
 where $W$ and $S$ are single-drawn according to the distribution $P_{W|S}$ induced by the learning algorithm and the distribution of the data $P_S$, $\eps > 0$ is the parameter used to denote the error allowed and $\delta > 0$  is the parameter used to denote the confidence $1-\delta$. This kind of probabilistic upper-bound on the generalization error is called a ``single-draw'' upper-bound on the generalization error.

In the classical learning scenario, Esposito et al. \cite{Esposito21} as a corollary of \cite[Theorem $4$]{Esposito21}, give an upper-bound on generalization error in probability mentioned below.

\begin{proposition}[{\cite[Corollary $2$]{Esposito21} }]\label{lemma_tail_pac_class}
 Assume that the loss function $l(w,Z)$ satisfies Assumption \ref{classical_sub_gaussian} for some $0<\tau<\infty$. Furthermore, assume that $P_{WS} \ll P_W P_S$ ( $P_S = P_Z^{\otimes n}$) . Then, for any $\gamma > 1, \delta \in (0,1)$, 

 \begin{equation}
     \cE := \left\{\abs{\textnormal{gen}(W,S)} \leq \sqrt{\frac{2 \tau^2}{n}\left(I^{c}_{\gamma}[W;S]+\log 2+\frac{\gamma}{\gamma-1} \log \left(\frac{1}{\delta}\right)\right)}\right\},\label{lemma_tail_pac_class_event}
 \end{equation}
 satisfies the following,
\begin{equation}
    \Pr_{(W,S) \sim P_{WS}}\left\{\cE\right\} \geq 1-\delta,\label{lemma_tail_pac_class_eq}
\end{equation}
where $\forall (w,s) \in \cW \times \cZ^n, \textnormal{gen}(w,s)$ is defined in \cite[Definition $8$]{Esposito21}.
\end{proposition}

From \cref{lemma_tail_pac_class_event,lemma_tail_pac_class_eq}, we note that if we aim to achieve at most $\eps$ ($\eps>0$) generalization error i.e. $\abs{\text{gen}(W,S)} \leq \eps$ with $(1-\delta)$ confidence, then, it is necessary to have $n \geq \frac{2 \tau^2}{\eps^2}\left(I^{c}_{\gamma}[W;S]+\log 2+\frac{\gamma}{\gamma-1} \log \left(\frac{1}{\delta}\right)\right)$ samples. Thus, for a fixed sample size, there is a trade-off between $\eps$ and $\delta.$

The proof of Proposition \ref{lemma_tail_pac_class} uses H\"older's inequality (see Fact \ref{holder_classic}) non-trivially and is arguably tedious. In comparison, Theorem \ref{lemma_tail_pac_naive} below proves an alternative strategy to upper-bound the generalization error in probability in terms of smooth max divergence (see Fact \ref{fact_smooth_max_divergence}, just using the definition of smooth max divergence and is simpler than the proof of Proposition \ref{lemma_tail_pac_class}.
\begin{theorem}\label{lemma_tail_pac_naive}
    Assume that the loss function $l(w,Z)$ satisfies Assumption \ref{classical_sub_gaussian} for some $0<\tau<\infty$.. Furthermore, assume that $P_{WS} \ll P_W P_S$ ( $P_S = P_Z^{\otimes n}$) . Then, for any $\delta \in (0,1), \nu < \delta$, the following holds,
    \begin{equation}
    \Pr_{(W,S) \sim {P}_{WS}}\left\{\abs{\textnormal{gen}(W,S)} \leq\sqrt{\frac{2\tau^2}{n}\left(  I^{(\nu)}_{\max}[W;S] + \log 2 + \log\left(\frac{1}{\delta - \nu}\right)\right)} \right\} \geq 1 - \delta,\label{lemma_tail_pac_naive_eq}
\end{equation}
where $I^{(\nu)}_{\max}[W;S] := D^{(\nu)}_{\max}(P_{WS}\|P_W\times P_S)$ (where $D^{(\nu)}_{\max}(\cdot\|\cdot)$ is defined in Fact \ref{fact_smooth_max_divergence}).
\end{theorem}
\begin{proof}
    See Appendix \ref{proof_lemma_tail_pac_naive} for the proof.
\end{proof}

{From \eqref{lemma_tail_pac_naive_eq}, it follows that with probability at most $\delta$, $\abs{\text{gen}(W,S)}$ will go beyond $O(\sqrt{I^{(\nu)}_{\max}[W;S]})$ for any $\nu < \delta$. Further, in earlier sections we observed that in any change of measure based upper-bounds on Generalization error, we try to approximate the the joint distribution $P_{WS}$ in terms of the marginal $P_{W}P_{S}$ along with some distance measure due to the change of the measure. Moreover, from \cite{RSW2017,AJW2018,AJW2019}, it can be realized that whenever we try to make a joint probability distribution ($P_{WS}$) approximately close to its marginal ($P_{W}P_{S}$), smooth max R\'enyi divergence (denoted as $D^{\nu}_{\max}$ and defined in Definition \ref{fact_smooth_max_divergence}) naturally comes into the picture. Thus, the upper-bound obtained in \eqref{lemma_tail_pac_naive_eq} in terms of smooth-max R\'enyi divergence seems to be well-justified.}

We now extend Proposition \ref{lemma_tail_pac_class} and Theorem \ref{lemma_tail_pac_naive} in the quantum learning scenario and compare them with the same in the classical scenario. Towards this, we require the following sub-Gaussianity assumption,
\begin{align}
\log\bbE_{Z \sim P} \left[e^{\lambda\left(\tr\left[L(w,Z_i)\sigma^{\cA_{Q}}(w,Z)\right] -  \bbE_{\overline{Z} \sim P} \left[(\tr\left[L(w,\overline{Z})\sigma^{\cA_{Q}}(w,\overline{Z})\right]\right]\right)}\right] &\leq \frac{\lambda^2\tau^2}{2}\label{classical_mgf3_frag}.
\end{align}
where $0 < \tau < \infty$. In the theorem below, we now mention a quantum version of Proposition \ref{lemma_tail_pac_class} (\cite[Corollary $2$]{Esposito21}), in the context of above above-discussed ``single-draw'' probabilistic bounds in a quantum learning scenario. 
\begin{theorem}\label{lemma_tail_pac}
    Given the distribution ${P}^{\cA_{Q}}_{WS}$ induced by a quantum learner ${\cA_{Q}}$ (such that ${P}^{\cA_{Q}}_{WS} \ll {P}^{\cA_{Q}}_{W} \times P^n$), and for any $\delta \in (0,1) , \alpha \in (0,1), \gamma >1$ if the sub-Gaussianity assumptions mentioned in \cref{quantum_mgf_emp_frag,quantum_mgf_true_frag,classical_mgf3_frag} holds for some $0<\mu,\tau<\infty$. Then the event
    \begin{equation}
        \cE := \left\{\abs{\textnormal{gen}(W,S)} \leq \sqrt{\frac{2\tau^2}{n}\left(\log 2 + I^{c}_{\gamma}[S;W] + \frac{\gamma}{\gamma - 1}\log\left(\frac{1}{\delta}\right)\right)} + \inf_{\alpha \in (0,1)} c_1(\alpha) \right\},\label{lemma_tail_pac_eq1}
    \end{equation}
    satisfies the following,
    \begin{align}  
        \Pr_{(W,S) \sim {P}^{\cA_{Q}}_{WS}}\left\{\cE\right\} &\geq 1 - \delta, \nn
    \end{align}
    where, 
    \begin{align}
    c_1(\alpha) &:= \sup_{w \in \textnormal{supp}(P^{\cA_Q}_{W})} \bbE_{S \sim P^n}\left[\sqrt{\frac{2\mu^2 \overline{D}_{\alpha}\left(\sigma^{\cA_{Q}}(w,S)||\rho_{{te}}(S) \otimes \sigma^{\cA_{Q}}_{hyp}(w,S)\right)}{n \alpha}} + \sqrt{\frac{2\mu^2 \overline {D}_{\alpha}\left(\sigma^{\cA_{Q}}_{hyp}(w,S)|| \sigma^{\cA_{Q}}_{{hyp}}(w)\right)}{n\alpha}}\right].\label{c1_expr}
    \end{align}
    
\end{theorem}

\subsection{Proof of Theorem \ref{lemma_tail_pac}}
    Consider the following event
    \begin{equation*}
        E := \{(w,s) \in \cW \times \cZ^n : \abs{\text{\textnormal{gen}}(w,s)} > \eps\},
    \end{equation*}
    where for any $w \in \cW$ , we define $E_{w} := \{s \in \cZ^n : (w,s) \in E\}$. 
    Then, from \eqref{Fact_change_measure_eq} of Fact \ref{Fact_change_measure}, for some $\gamma > 1$ we can write the following,

    \begin{align}
        \Pr_{(W,S) \sim {P}^{\cA_{Q}}_{WS}}\{E\} &\leq  \exp\left[{\frac{\gamma - 1}{\gamma} \left(\log\left(\bbE_{{P}^{\cA_{Q}}_{W}}\left[\Pr_{S\sim P^n}\{E_{W}\}\right]\right) + I^{c}_{\gamma}[S;W]\right)}\right].\label{probability_exp_2}
    \end{align}
    
From \eqref{loss_observables_frag}, it follows that,

\begin{align*}
        \hat{l}_{\rho}(w,s) 
        &=\frac{1}{n}\sum_{i=1}^{n}\tr[L(w,z_i)\sigma^{\cA_{Q}}(w,z_i)],\\
        l_{\rho}(w) &= \frac{1}{n}\sum_{i=1}^{n}\bbE_{\overline{Z}_i \sim  P}\left[\tr\left[L(w,\overline{Z}_i)\left(\rho_{Z_{te}}(\overline{Z}_i) \otimes \sigma^{\cA_{Q}}_{\widehat{hyp}}(w)\right)\right]\right].
    \end{align*}

    Further, for any $w \in \cW, s := (z_1,\cdots,z_n) \in \cZ^n$, we define $\Tilde{l}_{\rho}(w)$ as follows,
    \begin{equation*}
        \Tilde{l}_{\rho}(w) := \frac{1}{n}\sum_{i=1}^{n}\bbE_{\overline{Z}_i \sim P}\tr\left[L(w,\overline{Z}_i)\sigma^{\cA_{Q}}(w,\overline{Z}_i)\right].
    \end{equation*}
    Then, for any $w \in \textnormal{supp}(P^{\cA_Q}_{W})$, we have the following,
\begin{align}
        \hspace{10pt}\Pr_{S\sim P^n}\{E_{w}\}
        &=\Pr_{S\sim P^n}\left\{\abs{\textnormal{gen}(w,S)} > \eps\right\}\nn\\
        &\overset{a}{=} \Pr_{S\sim P^n}\left\{\abs{\hat{l}_{\rho}(w,S) - l_{\rho}(w) } > \eps\right\}\nn\\
        &\overset{b}{\leq} 2 e^{\frac{-n\left(\eps -\abs{\Tilde{l}_{\rho}(w) - l_{\rho}(w)}\right)^2}{2\tau^2}}\nn\\
        &\overset{c}{\leq} 2\exp{\left[\frac{-n\left(\eps - c_1(\alpha, w)\right)^2}{2\tau^2}\right]},\label{probability_exp_3}
    \end{align}
    where $a$ follows from \eqref{gen_ws_def}, $b$ follows from Fact \ref{average_sub_gaussian_unnorm} and in $c$ we define $\Tilde{c}(w)$ as follows,
    \begin{equation}
        c_1(\alpha, w) := \bbE_{S \sim P^n}\left[\sqrt{\frac{2\mu^2 \overline{D}_{\alpha}\left(\sigma^{\cA_{Q}}(w,S)||\rho_{{te}}(S) \otimes \sigma^{\cA_{Q}}_{hyp}(w,S)\right)}{n \alpha}} + \sqrt{\frac{2\mu^2 \overline{D}_{\alpha}\left(\sigma^{\cA_{Q}}_{hyp}(w,S)|| \sigma^{\cA_{Q}}_{{hyp}}(w)\right)}{n\alpha}}\right],\label{c1w}
    \end{equation}
    and the inequality $c$ follows from the following series of inequalities,

    {\allowdisplaybreaks\begin{align}
        \abs{\Tilde{l}_{\rho}(w) - l_{\rho}(w)} &= \abs{\frac{1}{n}\sum_{i=1}^{n}\bbE_{\overline{Z}_i \sim  P}\left[\tr\left[L(w,\overline{Z}_i)\sigma^{\cA_{Q}}(w,\overline{Z}_i)\right] -\tr\left[L(w,\overline{Z}_i)\left(\rho_{Z_{te}}(\overline{Z}_i) \otimes \sigma^{\cA_{Q}}_{\widehat{hyp}}(w)\right)\right] \right]}\nn\\
        &\leq \frac{1}{n}\sum_{i=1}^{n}\bbE_{\overline{Z}_i \sim  P}\left[\abs{\tr\left[L(w,\overline{Z}_i)\sigma^{\cA_{Q}}(w,\overline{Z}_i)\right] -\tr\left[L(w,\overline{Z}_i)\left(\rho_{Z_{te}}(\overline{Z}_i) \otimes \sigma^{\cA_{Q}}_{\widehat{hyp}}(w,\overline{Z}_i)\right)\right] }\right]\nn\\
        &\hspace{10pt}+\frac{1}{n}\sum_{i=1}^{n}\bbE_{\overline{Z}_i \sim  P}\left[\abs{\tr\left[L(w,\overline{Z}_i)\left(\rho_{Z_{te}}(\overline{Z}_i) \otimes \sigma^{\cA_{Q}}_{\widehat{hyp}}(w,\overline{Z}_i)\right)\right] -\tr\left[L(w,\overline{Z}_i)\left(\rho_{Z_{te}}(\overline{Z}_i) \otimes \sigma^{\cA_{Q}}_{\widehat{hyp}}(w)\right)\right]} \right]\nn\\
        &\overset{a}{\leq} \frac{1}{n}\sum_{i=1}^{n}\bbE_{\overline{Z}_i \sim  P}\left[\sqrt{\frac{2\mu^2 \overline{D}_{\alpha}\left(\sigma^{\cA_{Q}}(w,\overline{Z}_i)\|\rho_{Z_{te}}(\overline{Z}_i) \otimes \sigma^{\cA_{Q}}_{\widehat{hyp}}(w,\overline{Z}_i)\right)}{\alpha}}+\sqrt{\frac{2\mu^2 \overline{D}_{\alpha}\left(\sigma^{\cA_{Q}}_{\widehat{hyp}}(w,\overline{Z}_i)\|\sigma^{\cA_{Q}}_{\widehat{hyp}}(w)\right)}{\alpha}}\right]\nn\\
        &\leq\bbE_{S \sim P^n}\left[\sqrt{\frac{2\mu^2 \overline{D}_{\alpha}\left(\sigma^{\cA_{Q}}(w,S)||\rho_{{te}}(S) \otimes \sigma^{\cA_{Q}}_{hyp}(w,S)\right)}{n \alpha}} + \sqrt{\frac{2\mu^2 \overline{D}_{\alpha}\left(\sigma^{\cA_{Q}}_{hyp}(w,S)|| \sigma^{\cA_{Q}}_{{hyp}}(w)\right)}{n\alpha}}\right],\nn
    \end{align}}
    where $a$ follows from Corollary \ref{loss_var_form_mod_san_frag}. Thus, from \cref{probability_exp_2,probability_exp_3} we have, 
    \begin{align}
        \Pr_{(W,S) \sim {P}^{\cA_{Q}}_{WS}}\left\{E\right\} &\leq \exp\left[{\frac{\gamma - 1}{\gamma} \left(\log\left(\bbE_{{P}^{\cA_{Q}}_{W}}\left[\Pr_{S\sim P^n}\{E_{W}\}\right]\right) + I^{c}_{\gamma}[S;W]\right)}\right]\nn\\
        &\leq \exp\left[{\frac{\gamma - 1}{\gamma} \left(\log\left(\sup_{ w \in \supp\left({P}^{\cA_{Q}}_{W}\right)}\left[\Pr_{S\sim P^n}\{E_{w}\}\right]\right) + I^{c}_{\gamma}[S;W]\right)}\right]\nn\\
        &{\leq}\exp\left[{\frac{\gamma - 1}{\gamma} \left(\log\left(\sup_{ w \in \supp\left({P}^{\cA_{Q}}_{W}\right)}\left[2\exp{\left[\frac{-n\left(\eps - c_1(\alpha, w)\right)^2}{2\tau^2}\right]}\right]\right) + I^{c}_{\gamma}[S;W]\right)}\right]\nn\\
        &\overset{a}{=} \exp\left[{\frac{\gamma - 1}{\gamma} \left(\log\left(\left[2\exp{\left[\frac{-n\left(\eps - c_1(\alpha)(w*)\right)^2}{2\tau^2}\right]}\right]\right) + I^{c}_{\gamma}[S;W]\right)}\right]\nn\\
        &= \exp\left[{\frac{\gamma - 1}{\gamma} \left(\log{2} -\frac{n\left(\eps - c_1(\alpha)\right)^2}{2\tau^2} + I^{c}_{\gamma}[S;W]\right)}\right],\label{pac_bound_1}
    \end{align}
    where in $a$, $c_1(\alpha, w^{\star}) := \sup_{ w \in \supp\left({P}^{\cA_{Q}}_{W}\right)} c_1(\alpha, w) = c_1(\alpha)$. If we now assume $\delta := e^{\frac{\gamma-1}{\gamma}\left(\log2 - \frac{n \left(\eps -c_1(\alpha)\right)^2}{2\tau^2} + I^{c}_{\gamma}[S;W]\right)}$. Then, we can write $\eps$ as follows,
    \begin{equation}
        \eps = \sqrt{\frac{2\tau^2}{n}\left(\log 2 + I^{c}_{\gamma}[S;W] + \frac{\gamma}{\gamma - 1}\log\left(\frac{1}{\delta}\right)\right)} + c_1(\alpha).\label{pac_eps}
    \end{equation}
    
    Hence, from \cref{pac_bound_1,pac_eps} we have the following,
    \begin{align*}
        \Pr_{(W,S) \sim {P}^{\cA_{Q}}_{WS}}\left\{\abs{\text{gen}(W,S)} \leq\sqrt{\frac{2\tau^2}{n}\left(\log 2 + I^{c}_{\gamma}[S;W] + \frac{\gamma}{\gamma - 1}\log\left(\frac{1}{\delta}\right)\right)} + c_1(\alpha) \right\} \geq 1 - \delta.
    \end{align*}
    Since the above inequality with any $\alpha \in (0,1)$, we obtain \eqref{lemma_tail_pac_eq1}.
    This completes the proof of Theorem \ref{lemma_tail_pac}.\hfill\QED
    
\begin{remark}\label{remark_extra_term}
     Observe that, unlike the event mentioned in \eqref{lemma_tail_pac_class_event}, in \eqref{lemma_tail_pac_eq1}, we have an extra term $c_1(\alpha)$ (defined in \cref{c1_expr}). This is because of the asymmetric nature of the generalization error as discussed earlier below the statement of Theorem \ref{theo_exxp_gen_err_bound_renyi}. Further, it is because of this asymmetric nature of the generalization error, we do not get a uniform upper-bound on $\Pr_{S\sim P^n}\{E_{w}\}$ (mentioned in \eqref{probability_exp_3}). However, this is not the case in the classical setting (see \cite[eq. $(59)$]{Esposito21}).
\end{remark}

Further, we mention a quantum version of Theorem \ref{lemma_tail_pac_naive} in Theorem \ref{lemma_tail_pac_smooth} below. 

\begin{theorem}\label{lemma_tail_pac_smooth}
    Given the distribution ${P}^{\cA_{Q}}_{WS}$ induced by a quantum learner ${\cA_{Q}}$ (such that ${P}^{\cA_{Q}}_{WS} \ll {P}^{\cA_{Q}}_{W} \times P^n$), then for any $\delta \in (0,1), \alpha \in (0,1), \nu < \delta$ if the sub-Gaussianity assumptions mentioned in \cref{quantum_mgf_emp_frag,quantum_mgf_true_frag,classical_mgf3_frag} holds for some $0<\mu,\tau<\infty$, the following holds,
    \begin{align}  
        \Pr_{(W,S) \sim {P}^{\cA_{Q}}_{WS}}\left\{\abs{\textnormal{gen}(W,S)} \leq \sqrt{\frac{2\tau^2}{n}\left(\log 2 + I^{(\nu)}_{\max}[S;W] + \log\left(\frac{1}{\delta - \nu}\right)\right)} + \inf_{\alpha \in (0,1)} {c}_1(\alpha) \right\} &\geq 1 - \delta, \text{ if } \alpha \in (0,1), \label{lemma_tail_pac_smooth_eq1}
    \end{align}
    where, ${c}_1(\alpha)$ is defined in \eqref{c1_expr}.
    \end{theorem}

We omit the proof of Theorem \ref{lemma_tail_pac_smooth} for brevity, as it directly follows from the techniques used in the proof of Theorem \ref{lemma_tail_pac_naive}. The results mentioned above in Theorems \ref{lemma_tail_pac} and \ref{lemma_tail_pac_smooth} give a single-draw upper-bound on the quantum generalization error (defined in Definition \ref{gen_ws_error}) in probability and can be thought of as a quantum version of Proposition \ref{lemma_tail_pac_class} and Theorem \ref{lemma_tail_pac_naive} respectively. No results similar to Theorems \ref{lemma_tail_pac} and \ref{lemma_tail_pac_smooth}, have been studied in the literature.

\begin{remark}
\label{single_draw_comparision}
Here, we remark the comparison of the quantum upper-bounds obtained in Sections \ref{sec:gen_bound_exp} and \ref{sec:gen_bound_exp} 
with their classical counterparts.
From the definitions of the generalization error (Definition \ref{gen_ws_error}) it follows that if the quantum training and testing data is not entangled i.e. $\forall s \in \cZ^n, \rho(s) = \rho_{{te}}(s) \otimes \rho_{{tr}}(s)$ and if the quantum hypothesis state ($\sigma^{\cA_{Q}}_{hyp}(w,s)$) is classically independent (through $S$) with the quantum data ($\rho(s)$) state i.e. for each $(w,s) \in \cW \times \cS, $ $\sigma^{\cA_{Q}}_{hyp}(w,s) = \sigma^{\cA_{Q}}_{hyp}(w)$, then, all the quantum upper-bounds obtained in Sections \ref{sec:gen_bound_exp} and \ref{sec:gen_bound_exp2} 
will boil down to their corresponding classical counterparts as stated earlier in Table \ref{tab}.
\end{remark}

\section{Conclusion}
Given the inherent stochasticity of training data and learned hypotheses, we have investigated the generalization error in quantum learning through its expectation and probabilistic behavior as two primary avenues.

\begin{table}[ht]
    \centering
    \caption{Comparison between the result obtained in \cite{Caro23} and our result  (mentioned in Sections \ref{sec:gen_bound_exp} 
and \ref{sec:gen_bound_exp2})}
    \label{tab:table1}
    \begin{tabular}{ || m{3cm} || m{7cm} || m{7cm} ||}
    \hline
    \hline
   Comparison & Upper-bounds obtained in \cite{Caro23} & Upper-bounds obtained in Section \ref{sec:gen_bound_exp}\\
   [0.5ex]
    \hline
        Definition of generalization error & Defined in Definitions \ref{gen_ws_error} and \ref{gen_ws_error_exp}. & Defined in \cref{Cgen_ws_error_eq,Cexp_gen_ws_error_eq}.\\
        \hline
        Assumption on bounded moment generating function & \cite[Theorem $17$]{Caro23} is proven under a general assumption of bounded moment generating functions of the loss observables (see eqs. (QMGF) and (CMGF) in \cite[Theorem $17$]{Caro23}). However, \cite[Corollaries $23$ and $24$]{Caro23} assumes sub-Gaussianity of loss observables (see \cref{quantum_mgf_true_caro,classical_mgf_caro})& All the results are based on sub-Gaussianity assumptions (see Assumptions \ref{Assumption_theo_gen} and \ref{Assumption_theo_gen_frag}).\\
          \hline 
          Quantum terms involved & A single term based on quantum divergence (defined in Fact \ref{fact_limit_petz_divergence}). & Two terms based on modified sandwiched quantum $\alpha$-R\'enyi divergence (defined in Definition \ref{def_renyi_mod_sandwiched}).\\
          \hline
          Classical terms involved & A single term based on divergence (defined in Fact \ref{fact_limit_renyi_divergence}). & A term based on  $\gamma$-R\'enyi divergence (defined in Definition \ref{def_renyi_class}).\\ 
          \hline
         General Result & The bounds are not generalized. & Generalized family of upper-bounds depending on the values of $\alpha, \gamma \in (0,1) \cap (1,\infty)$.\\
           \hline
          Recoverability & These results can't recover the results obtained in Section \ref{sec:gen_bound_exp}. & These results can easily recover results obtained in \cite{Caro23} (see Remark \ref{recoverability_to_caro} for more details).\\
          \hline
          Whole sample-based bounds on the expected generalization error & \cite[Theorem $17$ and Corollary $23$]{Caro23} study such upper-bounds. & Theorem \ref{theo_exxp_gen_err_bound_renyi}, Corollaries \ref{corr_exxp_gen_err_bound_renyi} and \ref{corr_exxp_gen_err_bound_renyi_weak} in this paper study such upper-bounds.\\
           \hline
          Individual sample-based bounds on the expected generalization error & \cite[Corollary $24$]{Caro23} studies such upper-bounds. &  Corollary \ref{cor_exxp_gen_err_bound_mod_renyi_frag} in this paper study such upper-bounds.\\
          \hline
          Single-draw upper-bounds on the generalization error in probability& Not investigated in \cite{Caro23}. & Theorems \ref{lemma_tail_pac} and \ref{lemma_tail_pac_smooth} in this paper study such upper-bounds. \\
          \hline
          Techniques involved &  Variational lower-bound of classical and quantum relative entropies required (see Facts \ref{class_kl_var_lb} and \ref{petz_quant_var_kl}). & Variational lower-bound of classical R\'enyi divergence and modified sandwiched quantum  R\'enyi divergence required (see Fact \ref{dv_renyi} and Lemma \ref{mod-san_renyi_var_form}).\\
    
          \hline
          
    \end{tabular}
\end{table}

Regarding the expected generalization error, we have built upon the quantum learning framework introduced by \cite{Caro23} and have proposed a novel definition for true loss within the quantum learning context. Leveraging this framework, we have established a family of upper bounds on the expected generalization error, expressed in terms of modified sandwiched, Petz, and classical R\'{e}nyi divergences. Notably, our bounds have encompassed the bound derived in \cite{Caro23} as a special case. Furthermore, under standard i.i.d. assumptions for both classical and quantum data, we have presented a family of upper bounds on the generalization error by using the modified sandwiched and Petz quantum R\'{e}nyi divergences and the classical R\'{e}nyi divergence. 
A detailed comparison between the initial results of this work and those of \cite{Caro23} has been provided in Table \ref{tab:table1}.
For the probabilistic behavior of the generalization error, we have derived a probabilistic bound based on the modified sandwich quantum R\'{e}nyi divergence and the classical R\'{e}nyi divergence. Additionally, we have obtained another probabilistic bound formulated using the smooth max R\'{e}nyi divergence. 
Importantly, all the upper bounds derived in this manuscript have held under specific sub-Gaussian assumptions for the loss observables. Within this work, we have demonstrated that these sub-Gaussian assumptions are a direct consequence of the boundedness of the loss observables (though the converse is not necessarily true), achieved by proving a quantum analogue of Hoeffding's lemma for bounded self-adjoint operators.

Finally, the proofs of all the upper bounds presented herein have necessitated the evaluation of the variational form. 
To this end, we have newly introduced a variational lower bound for the modified quantum R\'{e}nyi divergence. 
Moreover, we have presented an alternative proof for the variational lower bound of the Petz quantum R\'{e}nyi divergence (Fact \ref{meas_renyi_var_form}) and the Measurement-data-processing inequality of Petz quantum R\'{e}nyi divergence (Fact \ref{Mdata_processing_petz_renyi}) by employing the operator H\"{o}lder's inequality (Fact \ref{Holder_fact}) and the Araki-Lieb-Thirring inequality (Fact \ref{fact_thirring}), thus circumventing the need for the standard data-processing inequality.

\section*{Acknowledgments}
The work of N. A. Warsi was supported in part by MTR/2022/000814, DST/INT/RUS/RSF/P-41/2021 from the Department of Science \& Technology, Govt. of India and DCSW grant provided by the Indian Statistical Institute. The work of MH was supported in part by the National Natural Science Foundation of China under Grant 62171212 and the General R\&D Projects of 1+1+1 CUHKCUHK(SZ)-GDST Joint Collaboration Fund (Grant No. GRDP2025-022).

\bibliographystyle{IEEEtran}
\bibliography{master}

\appendices
\section{Proof of Lemma \ref{quantum_hoeffding_lemma}}\label{proof_quantum_hoeffding_lemma}

    We define $L' := L - \tr[L\rho]\bbI$ and note that $\tr[L'\rho] = 0$. Let $L'$ and $\rho$ have the following eigen decomposition,
     \begin{align*}
        L' = \sum_{i = 1}^{\abs{\cH}} \alpha_{i} \ketbra{i}, \text{ where } \forall i \in [\abs{\cH}], a - \tr[L\rho] < \alpha_i < b - \tr[L\rho], \nn\\
        \rho = \sum_{j = 1}^{\abs{\cH}} \beta_{j}\ketbra{j}, \text{ where } \forall j \in [\abs{\cH}], 0 < \beta_j < 1 \text{ and } \sum_{j=1}^{\abs{\cH}}\beta_j = 1.
    \end{align*}
    
    Then, we have,
    \begin{align}
        \log\tr\left[e^{\lambda L'}\rho\right] &= \log\left(\sum_{i=1}^{\abs{\cH}}\sum_{j=1}^{\abs{\cH}}e^\lambda {\alpha_i}\beta_j \abs{\braket{i}{j}}^{2}\right).\label{proof_quantum_hoeffding_lemma_1}
    \end{align}

For all $i \in [\abs{\cH}],$ let $p_i := \sum_{j=1}^{\abs{\cH}}\beta_j \abs{\braket{i}{j}}^{2}$. It is easy to see that $\forall i \in [\abs{\cH}],$ $p_i \geq 0$ and $\sum_{i = }^{|\cH|}p_i =1.$ Thus,
\begin{align}
    \bbE_{\mathbf{a} \sim P}[\mathbf{a}] = \sum_{i=1}^{\abs{\cH}}\alpha_i p_i = \tr\left[\sum_{i = 1}^{\abs{\cH}}\alpha_i\ketbra{i} \left(\sum_{j = 1}^{\abs{\cH}} \beta_{j}\ketbra{j}\right)\right] = \tr[L'\rho] = 0.\label{proof_quantum_hoeffding_lemma_2}
\end{align}

We now upper-bound $\log\tr\left[e^{\lambda L'}\rho\right]$ as follows,
\begin{align}
    \log\tr\left[e^{\lambda L'}\rho\right] &= \log\left(\sum_{i=1}^{\abs{\cH}}e^{\lambda \alpha_i} p_i\right)\nn\\
    &=\log\bbE_{\mathbf{a} \sim P}[e^{\lambda \mathbf{a}}]\nn\\
    &\overset{a}{\leq} \bbE_{ \mathbf{a} \sim P}[\lambda \mathbf{a}] + \frac{\lambda^2 (b - \tr[L\rho] - a + \tr[L\rho])^2}{8}\nn\\
    &\overset{b}{=}  \frac{\lambda^2(b - a)^2}{8},\label{proof_quantum_hoeffding_lemma_3}
\end{align}
where $a$ follows from Fact \ref{fact_hoefding_lemma} and $b$ follows from \eqref{proof_quantum_hoeffding_lemma_2}. This proves \eqref{quantum_hoeffding_lemma_eq1}. 
We now prove \eqref{quantum_hoeffding_lemma_eq2} as follows,

\begin{align*}
    \log\tr\left[e^{\lambda L}\rho\right]&=\log\tr\left[e^{\lambda (L' + \tr[L\rho]\bbI)}\rho\right] \nn\\
    &\overset{a}{=}  \log\tr\left[e^{\lambda L'}e^{\lambda \tr[L\rho]\bbI}\rho\right]\nn\\
    &=\log\left(\tr\left[e^{\lambda L'}e^{\lambda \tr[L\rho]}\rho\right]\right)\nn\\
    &=\log\left(e^{\lambda \tr[L\rho]} \tr\left[e^{\lambda L'}\rho\right]\right)\nn\\
    & = \lambda \tr[L\rho] + \log\tr\left[e^{\lambda L'}\rho\right]\nn\\
    &\leq \lambda \tr[L\rho] + \frac{\lambda^2(b - a)^2}{8},
\end{align*}
where $a$ follows since $L'$ and $\bbI$ commute with each other (otherwise, this equality in $a$ is not true). 

This completes the proof.\hfill\QED

%

\section{Proof of Theorem \ref{Caro23_result_mod}}\label{proof_Caro23_result_mod}

Before proceeding to the proof of Theorem \ref{Caro23_result_mod}, we state the following lemma, which will be required to prove Theorem \ref{Caro23_result_mod}.

\begin{lemma}\label{loss_var_form_KL}
Suppose for each $(w,s) \in \cW \times \cS,$ $\hat{L}(w,s)$ satisfies Assumption \ref{sub_g_ass_Caro_modifed}. Then, the following holds,
    \begin{align}
     \abs{\tr[\hat{L}(w,s)\sigma^{\cA_{Q}}(w,s)] - \tr\left[\hat{L}(w,s)\left(\rho_{te}(s) \otimes \sigma^{\cA_{Q}}_{hyp}(w,s)\right)\right]} &\leq
         \sqrt{2\mu^2 D\left(\sigma^{\cA_{Q}}(w,s)||\rho_{{te}}(s) \otimes \sigma^{\cA_{Q}}_{hyp}(w,s)\right)}, \label{loss_var_form_KL_eq}\\
    \abs{\tr\left[\hat{L}(w,s)\left(\rho_{te}(s) \otimes \sigma^{\cA_{Q}}_{hyp}(w,s)\right)\right] - \tr\left[\hat{L}(w,s)\left(\rho_{te}(s) \otimes \sigma^{\cA_{Q}}_{hyp}(w)\right)\right]} &\leq
         \sqrt{2\mu^2 D\left(\sigma^{\cA_{Q}}_{hyp}(w,s)||\sigma^{\cA_{Q}}_{hyp}(w)\right)}.\label{loss_var_form_KL_eq1}
    \end{align}
\end{lemma}
\begin{proof}
    See Appendix \ref{proof_loss_var_form_KL} for the proof.
\end{proof}

Another way to prove Lemma \ref{loss_var_form_KL} is via Lemma \ref{loss_var_form} by taking $\alpha \to 1$. However, it will require two assumptions, i.e., Assumptions \ref{sub_g_ass_Caro_modifed} and \ref{Assumption_theo_gen}.

We now consider the following series of inequalities,
{\allowdisplaybreaks\begin{align*}
    \abs{\overline{\text{gen}}} &= \abs{\bbE_{W \sim {P}^{\cA_{Q}}_{W}}\left[\bbE_{\overline{S} \sim P_{S}}\left[\tr\left[\hat{L}(W,\overline{S})\left(\rho_{te}(\overline{S}) \otimes \sigma^{\cA_{Q}}_{hyp}(W)\right)\right]\right] - \bbE_{S \sim {P}^{\cA_{Q}}_{S|W}}\left[\tr[\hat{L}(W,S)\sigma^{\cA_{Q}}(W,S)]\right]\right]}\\
    &\leq  \bbE_{\substack{W \sim {P}^{\cA_{Q}}_{W} \\ S \sim {P}^{\cA_{Q}}_{S|W}}}\left[\abs{\tr[\hat{L}(W,S)\sigma^{\cA_{Q}}(W,S)] - \tr\left[\hat{L}(W,S)\left(\rho_{te}(S) \otimes \sigma^{\cA_{Q}}_{hyp}(W)\right)\right]}\right] \\
    &\hspace{10pt}+ \bbE_{W \sim {P}^{\cA_{Q}}_{W}}\left[\left|\bbE_{\overline{S}\sim P_{S}}\left[\tr\left[\hat{L}(W,\overline{S})\left(\rho_{te}(\overline{S}) \otimes \sigma^{\cA_{Q}}_{hyp}(W)\right)\right]\right] - \bbE_{S \sim \hat{P}^{\cA_{Q}}_{S|W}}\left[\tr\left[\hat{L}(W,S)\left(\rho_{te}(S) \otimes \sigma^{\cA_{Q}}_{hyp}(W)\right)\right]\right]\right|\right]\hspace{70pt}\\
    &\leq \bbE_{\substack{W \sim {P}^{\cA_{Q}}_{W} \\ S \sim {P}^{\cA_{Q}}_{S|W}}}\left[\abs{\tr[\hat{L}(W,S)\sigma^{\cA_{Q}}(W,S)] - \tr\left[\hat{L}(W,S)\left(\rho_{te}(S) \otimes \sigma^{\cA_{Q}}_{hyp}(W,S)\right)\right]}\right] \\
    &\hspace{10pt}+\bbE_{\substack{W \sim {P}^{\cA_{Q}}_{W} \\ S \sim {P}^{\cA_{Q}}_{S|W}}}\left[\abs{\tr\left[\hat{L}(W,S)\left(\rho_{te}(S) \otimes \sigma^{\cA_{Q}}_{hyp}(W,S)\right)\right] - \tr\left[\hat{L}(W,S)\left(\rho_{te}(S) \otimes \sigma^{\cA_{Q}}_{hyp}(W)\right)\right]}\right] \\
    &\hspace{10pt}+ \bbE_{W \sim {P}^{\cA_{Q}}_{W}}\left[\left|\bbE_{\overline{S}\sim P_{S}}\left[\tr\left[\hat{L}(W,\overline{S})\left(\rho_{te}(\overline{S}) \otimes \sigma^{\cA_{Q}}_{hyp}(W)\right)\right]\right] - \bbE_{S \sim \hat{P}^{\cA_{Q}}_{S|W}}\left[\tr\left[\hat{L}(W,S)\left(\rho_{te}(S) \otimes \sigma^{\cA_{Q}}_{hyp}(W)\right)\right]\right]\right|\right]\hspace{70pt}\\
    &\overset{a}{\leq}  \bbE_{\substack{W \sim {P}^{\cA_{Q}}_{W} \\ S \sim {P}^{\cA_{Q}}_{S|W}}}\left[\sqrt{2\mu^2 D\left(\sigma^{\cA_{Q}}(W,S)||\rho_{te}(S) \otimes \sigma^{\cA_{Q}}_{hyp}(W,S)\right)}\right] + \bbE_{\substack{W \sim {P}^{\cA_{Q}}_{W} \\ S \sim {P}^{\cA_{Q}}_{S|W}}}\left[\sqrt{2\mu^2 D\left( \sigma^{\cA_{Q}}_{hyp}(W,S)||\sigma^{\cA_{Q}}_{hyp}(W)\right)}\right] \\
    &\hspace{10pt}+ \bbE_{W \sim {P}^{\cA_{Q}}_{W}}\left[\left|\bbE_{\overline{S} \sim P_{S}}\left[\tr\left[\hat{L}(W,\overline{S})\left(\rho_{te}(\overline{S}) \otimes \sigma^{\cA_{Q}}_{hyp}(W)\right)\right]\right]   -\bbE_{S \sim {P}^{\cA_{Q}}_{S|W}}\left[\tr\left[\hat{L}(W,S)\left(\rho_{te}(S) \otimes \sigma^{\cA_{Q}}_{hyp}(W)\right)\right]\right]\right|\right]\\
    &\overset{b}{\leq} \bbE_{\substack{W \sim {P}^{\cA_{Q}}_{W} \\ S \sim {P}^{\cA_{Q}}_{S|W}}}\left[\sqrt{2\mu^2 D\left(\sigma^{\cA_{Q}}(W,S)||\rho_{te}(S) \otimes \sigma^{\cA_{Q}}_{hyp}(W,S)\right)}\right]+ \bbE_{\substack{W \sim {P}^{\cA_{Q}}_{W} \\ S \sim {P}^{\cA_{Q}}_{S|W}}}\left[\sqrt{2\mu^2 D\left( \sigma^{\cA_{Q}}_{hyp}(W,S)||\sigma^{\cA_{Q}}_{hyp}(W)\right)}\right] + \sqrt{2\tau^2 I[S;W]}\\
    &=\bbE_{\substack{(W,S) \sim {P}^{\cA_{Q}}_{WS} }}\left[\sqrt{2\mu^2 D\left(\sigma^{\cA_{Q}}(W,S)||\rho_{te}(S) \otimes \sigma^{\cA_{Q}}_{hyp}(W,S)\right)}+ \sqrt{2\mu^2 D\left( \sigma^{\cA_{Q}}_{hyp}(W,S)||\sigma^{\cA_{Q}}_{hyp}(W)\right)}\right] + \sqrt{2\tau^2 I[S;W]},
\end{align*}}
    where $a$ follows from \cref{loss_var_form_KL_eq,loss_var_form_KL_eq1} and $b$ follows from  \cite[Theorem $1$]{XR_2017} under the classical sub-Gaussianity assumptions mentioned in \eqref{classical_mgf_caro_mod}. This completes the proof of Theorem \ref{Caro23_result_mod}.\hfill\QED

\section{Proof of Lemma \ref{loss_var_form_KL}}\label{proof_loss_var_form_KL}
We first prove \eqref{loss_var_form_KL_eq} in two cases and later we show that the proof of \eqref{loss_var_form_KL_eq1} follows similarly. Towards this, for all $\lambda \in \bbR$ we consider the following series of inequalities,

{\allowdisplaybreaks\begin{align}
    &D\left(\sigma^{\cA_{Q}}(w,s)||\rho_{te}(s) \otimes \sigma^{\cA_{Q}}_{hyp}(w,s)\right)\nn\\
    &\overset{a}{\geq} \lambda \tr\left[\hat{L}(w,s)\sigma^{\cA_{Q}}(w,s)\right] - \log\tr\left[e^{\lambda \hat{L}(w,s)}\left(\rho_{te}(s) \otimes \sigma^{\cA_{Q}}_{hyp}(w,s)\right)\right]\\
    &= \lambda \left(\tr\left[\hat{L}(w,s)\sigma^{\cA_{Q}}(w,s)\right] - \tr\left[\hat{L}(w,s)\left(\rho_{te}(s) \otimes \sigma^{\cA_{Q}}_{hyp}(w,s)\right)\right]\right)\nn\\
    &\hspace{100pt}+ \lambda\tr\left[\hat{L}(w,s)\left(\rho_{te}(s) \otimes \sigma^{\cA_{Q}}_{hyp}(w,s)\right)\right] -\log\tr\left[e^{\lambda \hat{L}(w,s)}\left(\rho_{te}(s) \otimes \sigma^{\cA_{Q}}_{hyp}(w,s)\right)\right]\nn\\ 
    &{=} \lambda \left(\tr\left[\hat{L}(w,s)\sigma^{\cA_{Q}}(w,s)\right] - \tr\left[\hat{L}(w,s)\left(\rho_{te}(s) \otimes \sigma^{\cA_{Q}}_{hyp}(w,s)\right)\right]\right)\nn\\
    &\hspace{100pt}-\log\tr\left[e^{\lambda \left(\hat{L}(w,s) -\tr\left[\hat{L}(w,s)\left(\rho_{te}(s) \otimes \sigma^{\cA_{Q}}_{hyp}(w,s)\right)\right]\bbI\right)}\left(\rho_{te}(s) \otimes \sigma^{\cA_{Q}}_{hyp}(w,s)\right)\right]\nn\\ 
    &\overset{b}{\geq} \lambda \left(\tr\left[\hat{L}(w,s)\sigma^{\cA_{Q}}(w,s)\right] - \tr\left[\hat{L}(w,s)\left(\rho_{te}(s) \otimes \sigma^{\cA_{Q}}_{hyp}(w,s)\right)\right]\right) - \frac{\lambda^2 \mu^2}{2},\label{proof_loss_var_form_KL_eq}
\end{align}}
where $a$ follows from Fact \ref{quant_kl_var_form} and \ref{data_processing_div}, and $b$ follows from \eqref{quantum_mgf_true1_caro_mod}.

We can rewrite the inequality \eqref{proof_loss_var_form_KL_eq} as follows:

\begin{align*}
    \frac{\lambda^2 \mu^2}{2} - \lambda \left(\tr\left[\hat{L}(w,s)\sigma^{\cA_{Q}}(w,s)\right] - \tr\left[\hat{L}(w,s)\left(\rho_{te}(s) \otimes \sigma^{\cA_{Q}}_{hyp}(w,s)\right)\right]\right) + D\left(\sigma^{\cA_{Q}}(w,s)||\rho_{te}(s) \otimes \sigma^{\cA_{Q}}_{hyp}(w,s)\right) \geq 0,
\end{align*}

Since the above inequality is a non-negative quadratic equation in $\lambda$ with the coefficient $\left(\frac{\mu^2}{2}\right) \geq 0,$ therefore its discriminant must be non-positive. Thus, we have the following inequality,

\begin{align}
    \left(\tr[\hat{L}(w,s)\sigma^{\cA_{Q}}(w,s)] - \tr\left[\hat{L}(w,s)\left(\rho_{te}(s) \otimes \sigma^{\cA_{Q}}_{hyp}(w,s)\right)\right]\right)^2 &\leq 4 \left(\frac{\mu^2}{2}\right)D\left(\sigma^{\cA_{Q}}(w,s)||\rho_{te}(s) \otimes \sigma^{\cA_{Q}}_{hyp}(w,s)\right)\nn\\
    \Rightarrow \abs{\tr[\hat{L}(w,s)\sigma^{\cA_{Q}}(w,s)] - \tr\left[\hat{L}(w,s)\left(\rho_{te}(s) \otimes \sigma^{\cA_{Q}}_{hyp}(w,s)\right)\right]} &\leq \sqrt{2\mu^2 D\left(\sigma^{\cA_{Q}}(w,s)||\rho_{{te}}(s) \otimes \sigma^{\cA_{Q}}_{hyp}(w,s)\right)}.\label{proof_loss_var_form_KL_eq1}
\end{align}

We now proceed to prove \eqref{loss_var_form_KL_eq1}. Using \cref{quantum_mgf_true2_caro_mod} and a calculation similar to \cref{proof_loss_var_form_KL_eq,proof_loss_var_form_KL_eq1} we have the following inequality,

\begin{align}
    \abs{\tr\left[\hat{L}(w,s)\left(\rho_{te}(s) \otimes \sigma^{\cA_{Q}}_{hyp}(w,s)\right)\right] - \tr\left[\hat{L}(w,s)\left(\rho_{te}(s) \otimes \sigma^{\cA_{Q}}_{hyp}(w)\right)\right]}
    &\leq \sqrt{2\mu^2 D\left(\rho_{te}(s) \otimes \sigma^{\cA_{Q}}_{hyp}(w,s)||\rho_{{te}}(s) \otimes \sigma^{\cA_{Q}}_{hyp}(w)\right)}\nn\\
    &\overset{a}{=} \sqrt{2\mu^2 D\left(\sigma^{\cA_{Q}}_{hyp}(w,s)|| \sigma^{\cA_{Q}}_{hyp}(w)\right)}\nn,
\end{align}
where $a$ follows from \eqref{fact_quantum_renyi_addit_eq1} of Fact \ref{fact_quantum_renyi_addit} and Fact \ref{fact_limit_petz_divergence}. This completes the proof of Lemma \ref{loss_var_form_KL}.\hfill\QED

\section{Proof of Theorem \ref{lemma_tail_pac_naive}}\label{proof_lemma_tail_pac_naive}

    Consider the following events
    \begin{align*}
        E &:= \{(w,s) \in \cW \times \cZ^n : \abs{\text{\textnormal{gen}}(w,s)} > \eps\},\nn\\
        A(\nu) &:= \left\{(w,s) \in \cW \times \cZ^n : \log\frac{P_{WS}(w,s)}{P_{W}(w) \times P_{S}(s)} < I^{(\nu)}_{\max}[W;S]\right\},
    \end{align*}
    where $\nu \in (0,1)$ and for any $w \in \cW$ , we define $E_{w} := \{s \in \cZ^n : (w,s) \in E\}$. Then we can write the following,

    \begin{align}
        \Pr_{(W,S) \sim P_{WS}}\{E\} &= \Pr_{(W,S) \sim P_{WS}}\{E \cap A(\nu)\} + \Pr_{(W,S) \sim P_{WS}}\{E 
        \cap A^{c}(\nu)\}\nn\\
        &\leq \Pr_{(W,S) \sim P_{WS}}\{E \cap A(\nu)\} + \Pr_{(W,S) \sim P_{WS}}\{A^{c}(\nu)\}\nn\\
        &\overset{a}{\leq} \sum_{(w,s) \in E \cap A(\nu)} P_{WS}(w,s) + \nu\nn\\
        &\overset{b}{\leq}  \sum_{(w,s) \in E \cap A(\nu)} \left(P_{W}(w)\times P_{S}(s)\right) e^{I^{(\nu)}_{\max}[W;S]} + \nu\nn\\
        &\leq e^{I^{(\nu)}_{\max}[W;S]} \Pr_{(W,S) \sim P_W\times P_S}\{E \cap A(\nu)\} + \nu\nn\\
        &\leq e^{I^{(\nu)}_{\max}[W;S]} \Pr_{(W,S) \sim P_W\times P_S}\{E\} + \nu\nn\\
        &=e^{I^{(\nu)}_{\max}[W;S]} \bbE_{W \sim P_{W}}\left[ \Pr_{S \sim P_S}\{E_{W}\}\right] + \nu.\label{prob_expression_joint}
    \end{align}

For each $w \in \cW$, from \cref{class_emp_loss_intro,class_gen_ws_intro}, we now consider the following series of ineqlities,
\begin{align}
    \Pr_{S \sim P_S}\{E_{w}\} &= \Pr_{S \sim P_S}\{\abs{\text{gen}(w,S)} > \eps\}\nn\\
    &= \Pr_{S \sim P_S}\{\abs{\frac{1}{n}\sum_{i=1}^{n}l(w,Z_i) - \bbE_{\overline{Z}}[l(w,\overline{Z})]} > \eps\}\nn\\
    &\overset{a}{\leq} 2e^{-\frac{n\eps^2}{2\tau^2}},\label{prob_expression_2}
\end{align}
where, $a$ follows from Fact \ref{average_sub_gaussian} as for each $w \in \cW$, $l(w,Z)$ is $\tau$-sub-Gaussian. Thus combining \cref{prob_expression_joint,prob_expression_2}, we write the following,

\begin{align}
    \Pr_{(W,S) \sim P_{WS}}\{E\} &\leq e^{I^{(\nu)}_{\max}[W;S]}2e^{-\frac{n\eps^2}{2\tau^2}} + \nu\nn\\
    &= e^{-\frac{n\eps^2}{2\tau^2} + I^{(\nu)}_{\max}[W;S] + \log 2} + \nu,\label{prob_expression_final}
\end{align}

If we choose $\delta:= e^{-\frac{n\eps^2}{2\tau^2} + I^{(\nu)}_{\max}[W;S] + \log 2} + \nu$, then $\eps$ can be written as follows,

\begin{equation}
    \eps = \sqrt{\frac{2\tau^2}{n}\left(\log 2 +  I^{(\nu)}_{\max}[W;S] + \log\left(\frac{1}{\delta - \nu}\right)\right)},\label{eps_exp}
\end{equation}

Hence, from \cref{prob_expression_final,eps_exp} we have the following,

\begin{equation*}
    \Pr_{(W,S) \sim {P}_{WS}}\left\{\abs{\text{gen}(W,S)} \leq\sqrt{\frac{2\tau^2}{n}\left(\log 2 +  I^{(\nu)}_{\max}[W;S] + \log\left(\frac{1}{\delta - \nu}\right)\right)} \right\} \geq 1 - \delta.
\end{equation*}
This completes the proof of Theorem \ref{lemma_tail_pac_naive}.\hfill\QED

\end{document}